\documentclass[sigconf,anonymous=false,review=false]{acmart}

\acmConference[FAccT '27]{ACM Conference on Fairness, Accountability,
  and Transparency}{June 2027}{TBD}
\acmYear{2027}
\copyrightyear{2027}
\acmISBN{978-x-xxxx-xxxx-x/27/06}
\acmDOI{10.1145/xxxxxxx.xxxxxxx}

\usepackage{booktabs}        
\usepackage{algorithm}       
\usepackage{algpseudocode}   

\usepackage{amsthm}
\usepackage{tikz}            
\usetikzlibrary{positioning,arrows.meta,shapes,fit}
\usepackage{graphicx}        
\usepackage{xcolor}
\usepackage{listings}
\usepackage{seqsplit}        
\usepackage{hyperref}
\usepackage{cleveref}

\newcommand{\btt}[1]{\texttt{\seqsplit{#1}}}

\newtheorem{lemma}{Lemma}
\newtheorem{definition}{Definition}

\title{KYA: A Framework-Agnostic Trust Layer for Autonomous Systems
       with Verifiable Provenance and Hierarchical Policy
       Composition}

\author{Kolawole Quadri}
\affiliation{%
  \institution{Veldt Labs}
  \country{USA}
}
\email{kola@veldtlabs.ai}

\begin{abstract}
\textbf{KYA (Know Your Agents)} is an open-source, framework-agnostic
trust and governance layer for autonomous systems, composed of five
primitives: (1)~a four-gate inbound apply pipeline; (2)~an
only-tighten composition algebra over a three-channel multi-tenant
hierarchy; (3)~\textbf{KYP (Know Your Principal)}, a schema-level
unification of trust scoring across human users, AI agents, and
service accounts; (4)~auditable interaction-multiplier amplification
over an AIVSS-shaped additive baseline; and (5)~two-axis delegation
attribution: a static premium for risky delegates and a runtime
debit for actual delegate misbehavior in multi-agent fan-out.
Together these span three pillars (\emph{trust, governance, and
evidentiary assurance}), making an autonomous system's actions
authorized, policy-conforming, and post-hoc verifiable: where
observability answers how long, how much, and what path, KYA answers
was it authorized, did it conform, and can it be verified; it
composes with observability rather than replacing it. It ships
native adapters for 15+ agent frameworks. On a $4 \times 9$ cross-backend matrix all 36 cells pass;
the pure-function scorer runs sub-millisecond at p99 and the system
sustains $\approx 1{,}800$ ops/sec at 20 concurrent workers with HMAC
chain integrity preserved end-to-end. KYA detects 89\% of 1{,}200
adversarial probes from PyRIT and Garak, including the
recently-published topology-guided multi-agent attack. The system is
available under Apache~2.0 as the \texttt{veldt-kya} package on PyPI.
\end{abstract}



\begin{document}

\maketitle

\section{Introduction}
\label{sec:intro}

Autonomous agents now decide loans, triage patients, adjudicate
benefits, and execute
trades~\cite{yehudai2025surveyagents,huang2025aagate}. The
observability platforms shipped alongside them ---
LangSmith~\cite{langsmith2024}, Phoenix~\cite{phoenix2024},
Langfuse~\cite{langfuse2023}, Braintrust~\cite{braintrust2024},
Weave~\cite{weave2024} --- answer when an agent is slow,
expensive, or generating errors; that is the observability
category and they execute on it. Governance is a distinct layer.
By mid-2026 several of these platforms have added audit-log
surfaces and EU AI Act compliance marketing, but the artifacts
they emit are \emph{operator-internal audit}: who-changed-what,
who-ran-which-eval, SOC~2 RBAC, platform-level access logs. None
produces \textbf{agent-identity-bound, cryptographically
verifiable, third-party-attestable} governance artifacts
(operator-verifiable in v1, with a notarized-attestation upgrade
path; \cref{sec:limitations}) --- because that is not what an
observability platform is for. KYA
addresses that distinct layer and composes with observability
rather than replacing it: \texttt{kya\_otlp\_bridge}
(\cref{sec:arch}) consumes OpenTelemetry spans from any of the
above as one input source, but KYA operates standalone without
them.

The gap is not in telemetry. It is in \textbf{agent identity,
evidentiary provenance, and enforceable behavioral contracts ---
a separate category that composes on top of telemetry rather than
competing with it.}

This paper presents \textbf{KYA (Know Your Agents)}, an open-source
trust and governance layer for autonomous
systems.\footnote{The phrase "Know Your Agent" was coined by
Chaffer (2025)~\cite{chaffer2025kya} for an agent-only Web3 identity
framework; we use \textbf{KYA (Know Your Agents)}, plural, for the
multi-principal trust layer described here and re-anchor the
agent-identity framing on a unified principal taxonomy (KYP) in
\cref{subsec:kyp}.} KYA targets the
agent-identity / evidentiary-provenance gap with five primitives
that, in combination, support the audit, attestation, and federated
policy distribution demands of modern AI governance frameworks
including the EU AI Act~\cite{EUAIAct2024}, NIST AI
RMF~\cite{NISTAIRMF}, ISO 42001~\cite{ISO42001}, and
SR-11/7~\cite{SR117}.

\subsection*{Contributions}

KYA is a \emph{systems contribution}: a code-verified governance
layer for autonomous systems whose novelty lies in the composition
discipline, the formal safety properties of that composition, and a
small set of engineering primitives that make the system deployable
in regulated multi-tenant environments. We do not claim new
cryptographic constructions or new mathematical primitives ---
foundational citations include Schneier and Kelsey
(1998)~\cite{schneierkelsey1998logs}, Bellare and Yee
(1997)~\cite{BellareYee1997}, Back and von Wright
(1998)~\cite{back1998refinement}, and Bell and LaPadula
(1976)~\cite{bell1976unified} --- but the \emph{composed system} has
no equivalent in the surveyed prior art~\cite{microsoft_agt_2026,trustpact2026,huang2025aagate,Aegis2026}.

We claim five contributions:

\begin{enumerate}
\item \textbf{A four-gate inbound apply pipeline} composing Ed25519
  signature verification, expiration check, only-tighten composition,
  and operator-approval-as-default (\cref{sec:fedrec}). Generalizes
  TUF/Uptane~\cite{samuel2010tuf,kuppusamy2016uptane} rollback
  prevention from a version-counter to a policy-strength lattice.

\item \textbf{An only-tighten composition algebra} for tenant-scoped
  policy overrides (\cref{sec:fedrec}). Tenant overrides can
  restrict, but never loosen, platform default weights; the property
  is structurally a refinement-calculus result (Back and von Wright,
  1998~\cite{back1998refinement}) applied to agent risk weights
  composed with cryptographically signed external recommendations.
  Lemma~\ref{lem:tighten} establishes soundness.

\item \textbf{KYP --- Know Your Principal} unified trust taxonomy
  (\cref{sec:rogue}). A schema-level unification (single table, one
  \texttt{principal\_kind} discriminator) of trust scoring across
  three principal kinds --- human user, AI agent, service account
  --- with shared signal mechanics. Closest commercial framing is
  Okta's 2026 NHI initiative~\cite{okta_nhi_unified_2026}; closest
  agent-governance system, Microsoft's Agent Governance
  Toolkit~\cite{microsoft_agt_2026}, scores agents only.

\item \textbf{Auditable interaction-multiplier amplification} over an
  AIVSS-shaped~\cite{owasp_aivss_2025} additive baseline
  (\cref{sec:static}). A bounded asymmetric (multipliers $\geq 1.0$
  only, product-capped at $\mathsf{MAX\_MULTIPLIER}$) registry of
  named pairwise interactions surfaced to the auditor with stable
  \texttt{code} strings. AIVSS, CVSS~\cite{first_cvss31_2019}, and
  FAIR~\cite{opengroup_fair_2021} are all purely additive or
  formula-driven without per-interaction audit codes.

\item \textbf{Two-axis delegation attribution}
  (\cref{sec:static}, \cref{sec:rogue}). On the static-score axis, an
  additive-with-cap delegation-trust premium with an \emph{observation
  gate} that prevents the cold-start false positives
  EigenTrust~\cite{kamvar2003eigentrust} suffers. On the dynamic-trust
  axis, an \texttt{actor\_agent\_key} convention at three hook-layer
  entry points: when orchestrator $A$ triggers sub-agents $B/C/D$ and
  any sub-agent fires a rogue signal, the signal is tagged with
  \texttt{actor\_agent\_key=$A$} so $A$'s principal-trust counter is
  debited at runtime. The static axis factors in \emph{risky}
  delegates; the dynamic axis registers \emph{actual delegate
  misbehavior}. Default \texttt{actor\_agent\_key=agent\_key} makes
  attribution work without per-customer wiring. Names the "clean
  orchestrator delegates to risky downstream" attack pattern (Liang et
  al.~\cite{liang2025donttrust}) as the contribution.
\end{enumerate}

\subsection*{Why now}

Three forces compound. GPAI obligations under the EU AI
Act~\cite{EUAIAct2024} have been in force since August 2025; risk
classification, post-market monitoring, and incident reporting
exceed what observability tools produce. The agent-governance
subfield emerged in twelve months
(AIVSS~\cite{owasp_aivss_2025},
AAGATE~\cite{huang2025aagate},
Aegis~\cite{Aegis2026},
SIGIL~\cite{shen2026sigil},
TrustPact~\cite{trustpact2026},
Microsoft AGT~\cite{microsoft_agt_2026}). Agentic systems have
moved from prototypes to production write paths in finance,
healthcare, and government. \textbf{Framework-agnosticism is a hard
requirement, not a convenience.}

\subsection*{Paper organization}

\Cref{sec:threat} states the threat model and design goals.
\Cref{sec:arch} describes the four-piece KYA stack and its deployment
topology. \Cref{sec:static}--\cref{sec:fedrec} present the five
contributions in detail. \Cref{sec:eval} reports empirical results
across four storage backends, concurrency loads, and adversarial
red-team campaigns. \Cref{sec:related} surveys related work,
\cref{sec:limitations} states limitations and open problems, and
\cref{sec:conclusion} concludes.

KYA ships under Apache~2.0 as \texttt{veldt-kya}~\cite{veldt-kya};
reproducibility artifacts accompany the release.

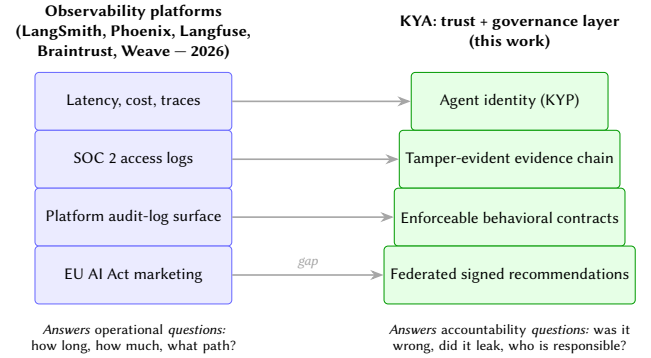
\begin{figure}[t]
\centering
\resizebox{\columnwidth}{!}{

\begin{tikzpicture}[
  font=\sffamily\small,
  box/.style={draw, rounded corners=2pt, minimum height=10mm, minimum width=34mm,
              align=center, inner sep=3pt},
  blue/.style={box, fill=blue!8, draw=blue!60},
  green/.style={box, fill=green!10, draw=green!60!black},
  title/.style={font=\sffamily\bfseries\small, align=center},
  arrow/.style={-{Stealth}, thick, gray!70},
]

\node[title] (lt) at (0,3.6) {Observability platforms\\(LangSmith, Phoenix, Langfuse,\\Braintrust, Weave --- 2026)};
\node[blue] (l1) at (0,2.4) {Latency, cost, traces};
\node[blue] (l2) at (0,1.4) {SOC~2 access logs};
\node[blue] (l3) at (0,0.4) {Platform audit-log surface};
\node[blue] (l4) at (0,-0.6) {EU AI Act marketing};

\node[title] (rt) at (6.5,3.6) {KYA: trust + governance layer\\(this work)};
\node[green] (r1) at (6.5,2.4) {Agent identity (KYP)};
\node[green] (r2) at (6.5,1.4) {Tamper-evident evidence chain};
\node[green] (r3) at (6.5,0.4) {Enforceable behavioral contracts};
\node[green] (r4) at (6.5,-0.6) {Federated signed recommendations};

\draw[arrow] (l4.east) -- node[above, font=\sffamily\itshape\footnotesize] {gap} (r4.west);
\draw[arrow] (l3.east) -- (r3.west);
\draw[arrow] (l2.east) -- (r2.west);
\draw[arrow] (l1.east) -- (r1.west);

\node[font=\sffamily\itshape\footnotesize, text width=4cm, align=center]
  at (0,-1.7) {Answers \emph{operational} questions: \emph{how long, how much, what path?}};
\node[font=\sffamily\itshape\footnotesize, text width=4.5cm, align=center]
  at (6.5,-1.7) {Answers \emph{accountability} questions: \emph{was it wrong, did it leak, who is responsible?}};

\end{tikzpicture}}
\caption{The agent-identity / evidentiary-provenance gap. By mid-2026
the observability layer markets audit logs; KYA addresses the
remaining gap: verifiable, third-party-attestable, agent-identity-bound
governance artifacts.}
\label{fig:hero}
\end{figure}

\providecommand{\cmark}{\checkmark}
\providecommand{\pmark}{$\circ$}

\section{Threat Model and Design Goals}
\label{sec:threat}

We model an autonomous agent as a stateful program $A$ with an
identifying definition $D_A$ (system prompt, model selection, tool
set, configuration), a runtime invocation context $C$, a set of
principals (users, sub-agents, services) participating in the
session, and an action output $a \in \mathsf{Actions}$. KYA reasons
about threats that manifest in $D_A$ (definitional integrity), $C$
(runtime behavior), the principals (identity / accountability), or
the boundary between $A$ and the outside world (data leakage,
cross-tenant access).

\subsection{Threat taxonomy}

We organize threats along two orthogonal axes (\cref{fig:threats}):
\textbf{static vs.\ dynamic} (whether the threat arises from the
configuration or runtime behavior) and \textbf{internal vs.\
external} (whether the threat originates from a privileged operator
or an external adversary). The taxonomy is informed by recent
literature on attack patterns in multi-agent LLM systems
\cite{liang2025donttrust,gabison2025liability,Aegis2026,huang2025aagate}.

\begin{description}
\item[T1 (static, internal): Definition drift.] Edit to $D_A$
  (swap model, add write tool, weaken prompt) after approval but
  without re-triggering governance review. Detected via canonical-hash
  mismatch over enumerated policy-bearing fields (\cref{sec:drift}).
\item[T2 (static, external): Supply chain.] Dependency replacement
  (typosquat, compromised framework, malicious marketplace update).
  Static scoring tags each dependency with publisher-trust class +
  supply-chain premium; signed inbound channel (\cref{sec:fedrec})
  provides Ed25519-verified updates.
\item[T3 (dynamic, internal): Capability drift.] Runtime tool call
  outside sanctioned set, or tenant\_id outside scope. Observed via
  out-of-scope + cross-tenant counters (\cref{sec:rogue}).
\item[T4 (dynamic, external): Prompt injection / jailbreak.]
  Adversarial instructions embedded in agent input. Observed via
  policy-violation counter, prompt-injection quality signal, red-team
  integration (\cref{sec:eval}).
\item[T5 (dynamic, internal): Data leakage.] $A$ emits data outside
  its sanctioned class (PHI / PII / trade secrets / ITAR). Observed
  via data-leak counter; default classifier covers civilian, defense,
  NATO, and EU schedules.
\item[T6 (any, post-hoc): Audit-log tampering.] Compromised operator
  modifies historical records. Detected via HMAC-chain verification
  with dialect-aware concurrency serialization
  (\cref{sec:evidence}).
\item[T7 (dynamic, external): Topology attack.] The
  ``clean orchestrator'' pattern: adversary compromises an exposed
  edge agent and propagates influence to a clean orchestrator that
  triggers no signals itself (40--78\% attack success rate against
  undefended systems~\cite{liang2025donttrust}). Detected via
  delegation-trust premium (static axis) + actor-agent attribution
  (dynamic axis).
\item[T8 (dynamic, internal): Attribution evasion.] Sub-agent rogue
  signals attributed to the puppet instead of the orchestrator,
  hiding the responsible upstream principal. Detected via the
  \texttt{actor\_agent\_key} convention (\cref{subsec:actoragent})
  defaulted at three SDK entry points.
\end{description}

\begin{figure}[t]
\centering
\resizebox{\columnwidth}{!}{

\begin{tikzpicture}[
  font=\sffamily\footnotesize,
  quad/.style={draw, thick, minimum width=42mm, minimum height=28mm,
               align=left, text width=40mm, inner sep=4pt},
  label/.style={font=\sffamily\bfseries\small},
  threat/.style={font=\sffamily\footnotesize},
]

\node[quad, fill=blue!5] (q1) at (0,0) {%
  \textbf{Static $\cdot$ Internal}\\[2pt]
  \textbf{T1} Definition drift\\
  (canonical hash, §\ref{sec:drift})};

\node[quad, fill=orange!5] (q2) at (5,0) {%
  \textbf{Static $\cdot$ External}\\[2pt]
  \textbf{T2} Supply-chain compromise\\
  (publisher trust, §\ref{sec:static})};

\node[quad, fill=blue!10] (q3) at (0,-3.2) {%
  \textbf{Dynamic $\cdot$ Internal}\\[2pt]
  \textbf{T3} Capability drift\\
  \textbf{T5} Data leakage\\
  \textbf{T8} Attribution evasion\\
  (rogue signals, §\ref{sec:rogue})};

\node[quad, fill=orange!10] (q4) at (5,-3.2) {%
  \textbf{Dynamic $\cdot$ External}\\[2pt]
  \textbf{T4} Prompt injection\\
  \textbf{T7} Topology propagation\\
  (red-team + delegation, §\ref{sec:eval})};

\node[quad, fill=red!8, minimum width=88mm, minimum height=12mm,
      text width=86mm, align=center, inner sep=4pt] (q5) at (2.5,-5.6) {%
  \textbf{Post-hoc audit} (any axis)\\
  \textbf{T6} Audit-log tampering --- HMAC chain verify (§\ref{sec:evidence})};

\node[label] at (-3.0, 0)    {Static};
\node[label] at (-3.0, -3.2) {Dynamic};
\node[label] at (0,  1.7)    {Internal};
\node[label] at (5,  1.7)    {External};

\draw[->, thick, gray] (-2.3, 1.4) -- (-2.3, -4.4);
\draw[->, thick, gray] (-2.0, 1.4) -- (7.0, 1.4);

\end{tikzpicture}}
\caption{KYA's threat taxonomy. Each quadrant labels which KYA
primitive detects the corresponding class of threat. T7--T8 are
informed by recent multi-agent attack literature.}
\label{fig:threats}
\end{figure}
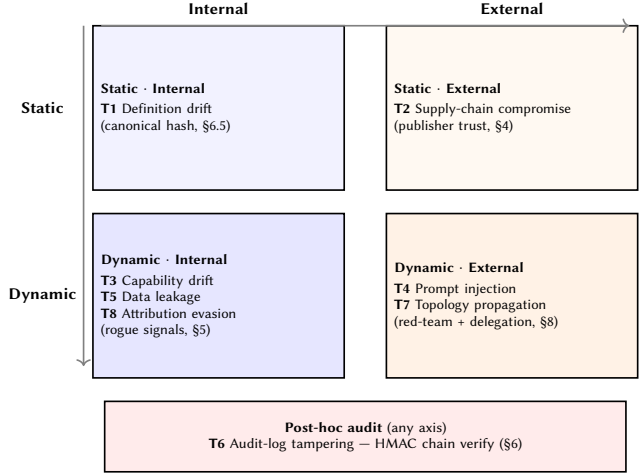

\subsection{Assumptions}

We assume a trusted operator holding the per-(tenant, invocation)
HMAC key; a trusted central collector (the SDK vendor) holding the
Ed25519 inbound-signing key in KMS or HSM (compromised-collector
handling: \cref{sec:fedrec}); honest-but-curious storage backends
(observed but not tamper-proof); and standard cryptographic
assumptions --- HMAC-SHA256
unforgeable~\cite{BellareYee1997,schneierkelsey1998logs}, Ed25519
EUF-CMA secure~\cite{Bernstein2012}.

\subsection{Design goals}

KYA targets seven goals:
\textbf{G1.\ Framework-agnostic} --- one adapter per framework, not
a per-framework rewrite of the core; \texttt{format\_adapter.py}
ships 15+ native adapters (LangChain, CrewAI, OpenAI
Assistants/Agents, Claude Agent SDK, AutoGen, Semantic Kernel,
LlamaIndex, Haystack, Bedrock, Vertex, MCP, and others).
\textbf{G2.\ Verifiable evidence} --- Schneier-Kelsey-class HMAC
chain~\cite{schneierkelsey1998logs} with the deployment-shape
extensions of \cref{sec:evidence}.
\textbf{G3.\ Multi-tenant safe by default} --- Cedar-style
forbid-dominance~\cite{cutler2024cedar} extended to a three-channel
hierarchical authority (\cref{sec:fedrec}).
\textbf{G4.\ Operationally light} --- \texttt{pip}-installable on
Python 3.10+, sub-millisecond scoring on the hot path, only
\texttt{SQLAlchemy} as a hard core dependency.
\textbf{G5.\ Regulator-legible artifacts} --- direct mapping to
EU AI Act risk tiers, NIST AI RMF functions, HIPAA / NYDFS / FDA
PCCP~\cite{fda2024pccp} formats.
\textbf{G6.\ Federation without data hoarding} --- only aggregated
weight recommendations cross organizational boundaries, only after
operator review.
\textbf{G7.\ Closed-set extensibility} --- every trust-layer
dimension enumerated; new kinds require SDK release, preventing
caller-supplied keyspace expansion.

\begin{table}[t]
\centering
\caption{Threat coverage comparison.
Columns: DEv = DeepEval, TEv = TrustEval, TPa = TrustPact, MAG =
Microsoft Agent Governance Toolkit~\cite{microsoft_agt_2026}, AUR =
AURA~\cite{satta2025aura}. \cmark: full coverage; \pmark: partial;
--: none. DeepEval and TrustEval are LLM-evaluation libraries and
AURA is an academic risk-scoring framework; their \texttt{--} cells
on threats outside evaluation (T1, T6, T7, T8) reflect category
scope --- they were not designed to cover audit-tampering or
topology-attack defense --- rather than deficiency. Their outputs
compose into KYA's evidence chain as \texttt{quality} signals.
Peer governance-layer comparison concentrates on MAG and TPa;
fuller per-system treatment is in \cref{sec:related}.}
\label{tab:coverage}
\small
\begin{tabular}{lcccccc}
\toprule
\textbf{Threat} & DEv & TEv & TPa & MAG & AUR & \textbf{KYA} \\
\midrule
T1 Drift              & --   & --   & --   & \pmark & --   & \cmark \\
T2 Supply-chain       & --   & \pmark & --   & \pmark & \pmark & \cmark \\
T3 Capability drift   & \pmark & \pmark & \cmark & \cmark & \pmark & \cmark \\
T4 Prompt injection   & \pmark & \cmark & \cmark & \cmark & \pmark & \cmark \\
T5 Data leakage       & --   & \cmark & \pmark & \cmark & --   & \cmark \\
T6 Audit tampering    & --   & --   & --   & \pmark & --   & \cmark \\
T7 Topology attack    & --   & --   & --   & \pmark\textsuperscript{*} & --   & \cmark \\
T8 Attribution evasion& --   & --   & --   & --   & --   & \cmark \\
Federation            & --   & --   & --   & \pmark\textsuperscript{*} & --   & \cmark \\
\bottomrule
\end{tabular}
\\\textsuperscript{*}MS-AGT discusses cross-org federation as future
work~\cite{microsoft_agt_issue1386}; topology defense via inverse
direction (parent caps child) rather than upward attribution.
\end{table}

\section{System Architecture}
\label{sec:arch}

KYA is delivered as a coordinated four-piece stack
(\cref{fig:fourpiece}). The core SDK (\texttt{kya}, distributed on
PyPI as \texttt{veldt-kya}~\cite{veldt-kya}) is a pure Python
library with one hard
dependency (SQLAlchemy~\cite{sqlalchemy}) and optional extras for
metrics, tracing, webhooks, and an LLM judge. Three sibling
components extend the core: \texttt{kya\_hooks} provides in-process
integrations for popular agent frameworks;
\texttt{kya\_otlp\_bridge} ingests
\mbox{OpenTelemetry}~\cite{opentelemetry} spans and translates them
into KYA signals; and \texttt{kya\_redteam} runs adversarial
probing campaigns via PyRIT~\cite{pyrit} and Garak~\cite{garak}. Only \texttt{kya} is
required; the other three are opt-in. The total artifact spans 50+
Python modules with cross-backend portability across PostgreSQL,
SQLite, DuckDB, and MySQL.

\begin{figure}[t]
\centering
\resizebox{\columnwidth}{!}{

\begin{tikzpicture}[
  font=\sffamily\small,
  comp/.style={draw, rounded corners=2pt, minimum width=28mm, minimum height=14mm,
               align=center, inner sep=4pt, fill=blue!8, font=\sffamily\small},
  core/.style={comp, fill=green!12, draw=green!60!black, very thick},
  label/.style={font=\sffamily\footnotesize\itshape, align=center},
  arrow/.style={-{Stealth}, thick, gray!70},
]

\node[core] (kya) at (0,0)
  {\textbf{kya} (core)\\
   \scriptsize PyPI: \texttt{veldt-kya}\\
   \scriptsize 50+ modules};

\node[comp, above left=10mm and 4mm of kya] (hooks)
  {\textbf{kya\_hooks}\\
   \scriptsize Claude / LangChain /\\\scriptsize OpenAI Agents};

\node[comp, above right=10mm and 4mm of kya] (bridge)
  {\textbf{kya\_otlp\_bridge}\\
   \scriptsize sidecar container\\\scriptsize OTLP $\to$ KYA};

\node[comp, below=10mm of kya] (red)
  {\textbf{kya\_redteam}\\
   \scriptsize PyRIT + Garak\\\scriptsize adversarial campaigns};

\draw[arrow] (hooks.south) -- (kya.north west);
\draw[arrow] (bridge.south) -- (kya.north east);
\draw[arrow] (red.north) -- (kya.south);

\node[label, below=2mm of red] {out-of-band};

\node[draw, dashed, gray, fit=(hooks)(kya), inner sep=8pt, rounded corners] {};
\node[font=\sffamily\footnotesize\itshape, above=1pt of hooks, gray] {in-process};

\end{tikzpicture}}
\caption{The four-piece KYA stack. \texttt{kya} runs in the agent
process; \texttt{kya\_hooks} attaches to LangChain / CrewAI /
OpenAI Assistants / Claude Agent SDK / OpenAI Agents loops;
\texttt{kya\_otlp\_bridge} runs as a sidecar that consumes
\mbox{OpenTelemetry} spans;
\texttt{kya\_redteam} runs out-of-band as an adversarial fleet.}
\label{fig:fourpiece}
\end{figure}
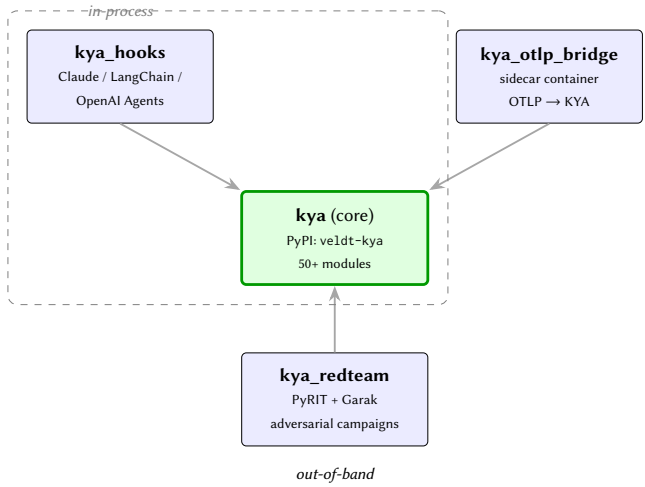

\subsection{Component overview}

\paragraph{\texttt{kya} (core SDK).}
Trust scoring, evidence chain, versioning, drift detection,
tenant-weight overrides, inbound signing, principal trust (KYP,
\cref{subsec:kyp}), and the adapter registry all live in the core.
Modules are organized by concern (static scoring, dynamic signals
with actor-agent attribution, HMAC evidence chain, canonical-hash
drift, snapshots, KYP, only-tighten overrides, bounded multipliers,
closed-loop feedback, four-gate inbound apply, pluggable adapters);
file-line citations in \cref{appx:disciplines}. The core has no
imports from any agent framework and depends only on the Python
standard library, SQLAlchemy, and optional extras.

\paragraph{\texttt{kya\_hooks} (framework integrations).}
A thin per-framework shim that registers callbacks into the agent's
invocation loop and emits canonical events into the core. Native
support ships for Anthropic's Claude Agent SDK, LangChain (callback
handler), and OpenAI's Agents SDK. The
\texttt{actor\_agent\_key=agent\_key} default convention
(\cref{subsec:actoragent}) is hard-wired at all three entry points
so that orchestrator attribution works without per-customer wiring.

\paragraph{\texttt{kya\_otlp\_bridge} (telemetry ingestion).}
A standalone container that consumes \mbox{OpenTelemetry} OTLP spans
(via gRPC or HTTP) and maps them into KYA signals. The bridge
extends OTel ingestion in four notable ways: (1) a span-kind
$\to$ evidence-semantics taxonomy across six OpenInference kinds
(LLM, TOOL, AGENT, RETRIEVER, GUARDRAIL, EVALUATOR); (2) cross-batch
trace stitching via a TTL-capped $\mathit{trace\_id} \to
\mathit{invocation\_id}$ cache so child spans arriving in separate
batches attach to the same invocation; (3) principal-identity
normalization (\texttt{normalize\_agent\_key}) preventing the
dual-principal split where customer-registered agents and OTel-runtime
emissions become two principals; (4) capture-time sensitivity
hinting where data-classification tags on spans drive
retention-policy enforcement at the moment evidence hits storage.

\paragraph{\texttt{kya\_redteam} (adversarial probing).}
An out-of-band campaign runner that drives PyRIT orchestrators and
Garak probes against a target agent deployment. Findings are
written back through the same \texttt{record\_evidence()} path as
production events, so red-team discoveries land in the same HMAC
chain as runtime observations. The module includes a curated set
of native probes (DAN persona break, base64/character-split
encoding-evasion, prompt extraction, Goodside override, markdown
injection) that run without an external Garak install.

\paragraph{Deployment topology.}
Raw prompt/completion payloads remain in the customer process; only
aggregated telemetry crosses the organizational boundary, and only
signed policy recommendations come back
(\cref{appx:topology-figure}).

\subsection{Three-layer runtime gates}
\label{sec:gates}

KYA composes with three orthogonal enforcement layers, each
evaluated at every agent action (\cref{fig:gates}):

\paragraph{Layer~1 --- Authentication and RBAC.}
Standard HTTP-request-time authorization: the caller's identity and
roles are resolved from a JWT; admin endpoints require explicit
privilege checks. KYA reads these decisions but does not own them.

\paragraph{Layer~2 --- Action gate.}
When a rule fires or an agent dispatches an action, the gate
evaluates governance policies against the action payload. Verdicts
are \emph{allow}, \emph{block}, \emph{redact}, \emph{throttle}, or
\emph{flag\_for\_review}; each verdict is logged to the evidence
chain. The gate consults tenant-scoped weights (\cref{sec:fedrec})
and class-tagged data classifications to make its verdict.

\paragraph{Layer~3 --- Tool RBAC.}
At every tool call inside an agent's invocation loop, KYA verifies
that the calling user's roles intersect the tool's required-role
set. The default catalog ships with the SDK; tenants may tighten
(but not loosen) it via override under the only-tighten algebra of
\cref{sec:fedrec}.

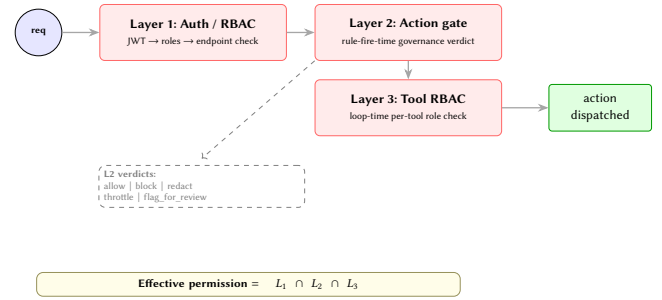
\begin{figure}[t]
\centering
\resizebox{\columnwidth}{!}{

\begin{tikzpicture}[
  font=\sffamily\small,
  req/.style={draw, circle, minimum size=10mm, fill=blue!10, inner sep=1pt},
  gate/.style={draw, rounded corners=3pt, minimum width=40mm, minimum height=12mm,
               align=center, inner sep=4pt, fill=red!8, draw=red!60},
  action/.style={draw, rounded corners=2pt, minimum width=22mm, minimum height=10mm,
                 align=center, inner sep=3pt, fill=green!12, draw=green!60!black},
  arrow/.style={-{Stealth}, thick, gray!70},
  caption/.style={font=\sffamily\footnotesize\itshape},
]

\node[req] (req) at (0,0) {\scriptsize\textbf{req}};

\node[gate, right=8mm of req] (g1) {%
  \textbf{Layer 1: Auth / RBAC}\\
  \scriptsize JWT $\to$ roles $\to$ endpoint check};

\node[gate, right=6mm of g1] (g2) {%
  \textbf{Layer 2: Action gate}\\
  \scriptsize rule-fire-time governance verdict};

\node[gate, below=4mm of g2] (g3) {%
  \textbf{Layer 3: Tool RBAC}\\
  \scriptsize loop-time per-tool role check};

\node[action, right=10mm of g3] (act) {action\\dispatched};

\draw[arrow] (req) -- (g1);
\draw[arrow] (g1) -- (g2);
\draw[arrow] (g2.south) -- (g3.north);
\draw[arrow] (g3) -- (act);

\node[draw, dashed, gray, rounded corners, font=\sffamily\scriptsize, align=left,
      text width=42mm] at (3.5,-3.3)
  {\textbf{L2 verdicts:}\\
   allow $|$ block $|$ redact\\
   throttle $|$ flag\_for\_review};
\draw[->, dashed, gray] (g2.south west) -- (3.5,-2.7);

\node[font=\sffamily\footnotesize\bfseries, text width=88mm, align=center,
      fill=yellow!10, draw=yellow!50!black, rounded corners, inner sep=4pt]
  at (4.5,-5.4) {Effective permission $= L_1 \cap L_2 \cap L_3$};

\end{tikzpicture}}
\caption{Three runtime enforcement layers. The effective permission
for any agent-driven action is the intersection
$L_1 \cap L_2 \cap L_3$.}
\label{fig:gates}
\end{figure}

\subsection{Persistence and cross-backend portability}

KYA's storage primitives run unchanged against PostgreSQL, SQLite,
DuckDB, and MySQL: dialect-aware variant types
(\btt{\_portable.py}), schema retargeting via SQLAlchemy's
\btt{schema\_translate\_map}, and an upsert dispatcher
(\btt{\_dialect\_helpers.py}) that emits the right \texttt{ON
CONFLICT} grammar per backend. Runtime concurrency uses the
dialect-aware serialization primitive of
\cref{subsec:concurrency}. A fresh \btt{pip install veldt-kya} is
usable end-to-end without provisioning a database server via
\btt{kya.default\_session()}'s SQLite fallback. The full stack is
validated across the $4 \times 9$ phase matrix and the $17
\times 4$ per-table matrix of \cref{sec:eval}.

\subsection{Cross-cutting design disciplines}
\label{sec:disciplines}

Five design disciplines recur across the 50+ SDK modules and are
load-bearing for KYA's claims: \textbf{closed-set whitelists with
explicit "unknown" buckets} (signal kinds, principal kinds,
evidence kinds, inbound scopes, data classes, compliance regimes
--- every scoring/persistence dimension is a closed enumeration;
caller-supplied strings cannot expand the keyspace);
\textbf{bounded composition everywhere} (per-factor caps,
$\mathsf{MAX\_MULTIPLIER}{=}2.0$, delegation-trust cap $25$, final
score clamp to $100$); \textbf{asymmetric composition rules} (data
sensitivity via \textbf{MAX}, security capabilities via
\textbf{SUM}, interaction multipliers $\geq 1.0$ only ---
\texttt{register\_interaction} raises on $<1.0$);
\textbf{explicit "never X" runtime invariants} (never auto-tune,
never apply-anyway on signature failure, never bypass human gate
on critical verdicts, never break scoring on hot-path exceptions,
never leak across tenants); and \textbf{fail-soft observability}
(Valkey / Phoenix / collector / \texttt{prometheus\_client}
unreachable $\Rightarrow$ core scoring and persistence continue).
Each is a hard runtime check, not a convention; \cref{appx:disciplines}
enumerates the source-line citations.

\subsection{Running example: a loan-decisioning agent fleet}
\label{subsec:running-example}

We anchor the rest of the paper in a single scenario revisited across
\cref{sec:static,sec:evidence,sec:fedrec,sec:eval}. A mid-size US
regional bank operates a four-agent loan-decisioning fleet under
NYDFS Part~500, ECOA, and CFPB UDAAP oversight: a \textbf{Loan Triage
Agent} (orchestrator) coordinates a \textbf{Document Verification
Sub-Agent} (OCR + identity-matching against a third-party KYC
provider), an \textbf{OFAC Screening Sub-Agent} (sanctions list
lookup), and a \textbf{Risk Review Agent} (credit-policy reasoner
with write authority to flag applications for human underwriter
review). Principals span all three KYP kinds: the underwriter
(\texttt{user}), the four agents (\texttt{agent}), and a nightly
batch reconciliation job (\texttt{service\_account}). Data classes
in scope are PII (\texttt{ssn}, \texttt{dob}, address), financial
(\texttt{account\_number}, income, debt-to-income), and KYC
documents (driver's license, utility bill). The fleet's adversary
surface includes prompt injection in uploaded
documents~\cite{liang2025donttrust} (T4), out-of-scope tool calls
by the Risk Review Agent (T3), data leakage of SSN in summary
outputs (T5), and topology-guided propagation from a compromised
Document Verification Sub-Agent upward into a clean Triage
orchestrator (T7). \Cref{sec:static} scores the four agents;
\cref{sec:evidence} traces a single application through the
HMAC-chained evidence log; \cref{sec:fedrec} demonstrates a
cross-bank federated recommendation tightening the OFAC tool
multiplier after a sector-wide incident; \cref{sec:eval} reports
the topology-attack detection result against exactly this fleet
topology.

\subsection{Second domain: a clinical triage agent}
\label{subsec:second-domain}

The framework-agnostic claim (G1, \cref{sec:threat}) requires the
primitives to land cleanly on domains that share no regulator,
threat surface, or data taxonomy with loan decisioning. We anchor
this with a parallel scenario revisited briefly in
\cref{sec:static,sec:evidence}. A regional hospital network operates
a clinical triage fleet under HIPAA, 21~CFR Part~11, and
FDA~SaMD~oversight: a \textbf{Triage Orchestrator} (intake-symptom
parser, no write authority of its own) coordinates a
\textbf{Diagnosis Suggestion Sub-Agent} (LLM-driven differential
shortlist, suggest-only), an \textbf{HL7 Lookup Sub-Agent} (read-only
EHR record fetch), and a \textbf{Medication-Interaction Checker}
(rules-engine wrapper, no write authority). Principals span the same
three KYP kinds as the loan fleet but in different proportions: the
attending nurse and supervising physician (\texttt{user}), four
agents (\texttt{agent}), and an overnight bulk-ingest job pulling
HL7 deltas (\texttt{service\_account}). Data classes are different:
PHI (medical record number, diagnosis codes, lab results), ePHI
(real-time vitals from monitors), and \texttt{phi\_genetic}. The
threat surface shifts accordingly: T4 prompt injection in
patient-history free-text; T5 disclosure of diagnosis to an
unauthorized requester; T7 topology propagation from a compromised
HL7 sub-agent up into the orchestrator (the same primitive
defends, but the upstream attack vector is the EHR API rather than
an OCR pipeline). What makes this fleet a useful test of
framework-agnosticism is that nothing about the KYA primitives ---
not the scoring factors, not the only-tighten algebra, not the
KYP principal-trust schema, not the actor-agent attribution
convention --- requires per-domain customization. The compliance
regime mapper (\texttt{kya.compliance}) returns
HIPAA-shaped retention windows (6~years) and breach-notification
SLAs (60~days) in place of the bank's ECOA / NYDFS configuration,
and the data-class taxonomy reports MAX-weighted sensitivity
according to the same monotone schedule.
\section{Static Risk Scoring}
\label{sec:static}

KYA's static risk model is structurally an AIVSS-shaped weighted
scorer (OWASP AIVSS v0.5~\cite{owasp_aivss_2025}): a per-factor
decomposition followed by an aggregation step. We do not claim novel
scoring mathematics. The contribution at this layer is twofold: (1)
an \textbf{auditable interaction-multiplier amplification} step over
the AIVSS-shaped additive baseline (\cref{subsec:interactions}), and
(2) a set of \textbf{design disciplines} (closed-set whitelists,
asymmetric MAX-vs-SUM composition rules, bounded composition,
explicit-over-inferred precedence, fail-soft scoring) that make the
factor model auditable end-to-end. The scoring pipeline composes
with the only-tighten algebra (\cref{sec:fedrec}) and the
HMAC-chained evidence log (\cref{sec:evidence}) so every score and
its per-factor decomposition is recoverable from the audit chain.

\subsection{Factor decomposition}

The static score $S(D_A)$ is the sum of deltas from independent
factor functions (\cref{tab:factors}). Each returns a structured
\texttt{RiskFactor(name, label, delta)} record preserved in
\texttt{AgentRiskScore} and persisted into the evidence chain; an
auditor querying any score can recover its full per-factor
attribution. \Cref{appx:factor-detail} gives the per-bucket weight
schedules and edge cases.

\begin{table}[t]
\centering
\caption{Static-score factors. MAX = take the highest applicable;
SUM = sum applicable values; cap = listed ceiling.}
\label{tab:factors}
\scriptsize
\begin{tabular}{llr}
\toprule
\textbf{Factor} & \textbf{Aggregation / weight} & \textbf{Cap} \\
\midrule
Base score                    & constant $+5$                   & --- \\
Write-tool count              & $+4$ per write tool             & --- \\
Admin-gated tools             & $+8$ per admin-gated tool       & --- \\
\textbf{Governance mode}      & none=$+30$, on-loop=$+15$,      & ---  \\
                              & hybrid=$+10$, in-loop=$0$       &      \\
\texttt{can\_override} flag   & $+12$ if True                   & --- \\
\texttt{can\_revert} flag     & $+8$ if True                    & --- \\
\texttt{access\_level=write}  & $+6$                            & --- \\
Data sensitivity              & MAX (public to top-secret)      & 60 \\
Security capabilities         & SUM (orthogonal powers)         & 60 \\
Provenance                    & MAX (builtin $\to$ third-party) & 20 \\
Model trust                   & MAX (enterprise $\to$ self-host)& 10 \\
Blast radius                  & SUM composite                   & 30 \\
Deployment environment        & MAX (dev to enclave)            & 25 \\
Delegation depth              & MAX chain length                & 25 \\
Supply chain                  & SUM + breadth premium           & 35 \\
Input sources                 & SUM + breadth premium           & 25 \\
Lifecycle                     & SUM (signed, approval, age)     & --- \\
Compliance scope              & MAX regulatory severity         & --- \\
Trust signals                 & SUM (neg.\ for evidence)        & --- \\
Cost burn                     & SUM operational anomaly         & --- \\
Delegation-trust prem.        & SUM observation-gated           & 25 \\
\bottomrule
\end{tabular}
\end{table}

\paragraph{Why MAX vs SUM matters.} The asymmetry is load-bearing:
data sensitivity aggregates via MAX because handling PHI is what
makes an agent high-stakes (touching internal + confidential + PHI
does not triple-count); security capabilities aggregate via SUM
because having both \texttt{code\_execution} and \texttt{shell\_access}
is materially worse than having one. ITAR is weighted above civilian
\texttt{secret} because ITAR violations carry felony export-control
consequences. \texttt{unknown} input source is weighted \emph{higher}
than declared \texttt{external\_api} because declaration is
information; absence is suspicious.

Most weighted scoring systems use a single composition rule
throughout (CVSS's nested formula, FAIR's Monte Carlo, AIVSS's
single threat multiplier); KYA's choice-per-dimension with
documented rationale (\cref{tab:factors}) is a design discipline
worth naming.

\subsection{The interaction-multiplier amplification step}
\label{subsec:interactions}

The additive sum captures the bulk of the signal but misses
compounding effects. A write tool that touches PII in an autonomous
deployment is materially more dangerous than the linear sum of those
three factors suggests. KYA defines an explicit registry of
\emph{interactions} --- predicates over $D_A$ and the factor list
that trigger a multiplicative bonus. This is the strongest claim of
novelty in the scoring layer.

\paragraph{Closed-form composition.}
Let $S_{add}$ be the additive score and let
$\mathcal{I}(D_A)$ be the set of interactions that fire for $D_A$.
The final additive score is

$$S(D_A) = \min\Big(100, \; S_{add} \cdot \min\big(\mathsf{MAX\_MULTIPLIER}, \prod_{i \in \mathcal{I}(D_A)} m_i\big)\Big)$$

with $\mathsf{MAX\_MULTIPLIER} = 2.0$ (capped product) and
$m_i \geq 1.0$ for every registered interaction (asymmetric: amplify,
never reduce).

\paragraph{Asymmetry enforced at runtime.} The
\btt{register\_interaction(code, name, condition, multiplier,
description, severity)} function raises \texttt{ValueError} if
$\mathit{multiplier} < 1.0$ (\btt{interactions.py:268}). The rule
is not a convention but a runtime invariant: "Use credit factors
(citation, audits) for downward deltas."

\paragraph{Auditable per-interaction codes.} Each interaction carries
a stable string \texttt{code} surfaced on the scoring result so an
auditor or UI can render "why was this score amplified." The 10
pre-registered interactions include:

\begin{itemize}
\item \texttt{autonomous\_writer\_in\_prod} (1.3$\times$, critical)
\item \texttt{code\_exec\_with\_user\_input} (1.5$\times$, critical):
  "Classic RCE-via-prompt-injection setup. Treat as if attacker
  controls code execution on the host."
\item \texttt{classified\_autonomous} (1.4$\times$, critical): "EU AI
  Act Art.~14 and US classification regimes both require effective
  human oversight --- this combo violates both spirit and letter."
\item \texttt{untrusted\_chain} (1.2$\times$, warning): marketplace
  agent fanning out to $\geq 3$ delegates.
\item \texttt{unaudited\_classified} (1.25$\times$, critical):
  "Classified data being handled but no red-team / bias / fairness
  audit evidence on file."
\item \texttt{rejected\_in\_prod} (1.4$\times$, critical): an agent
  with security review status \texttt{rejected} but running in
  production.
\item \texttt{orphan\_writer\_in\_prod} (1.2$\times$, warning).
\item \texttt{prod\_marketplace\_writer} (1.2$\times$, warning).
\item \texttt{self\_hosted\_with\_pii} (1.2$\times$, warning).
\item \texttt{unowned\_high\_risk} (1.15$\times$, warning).
\end{itemize}

The registry is extensible: tenants register custom interactions via
\texttt{register\_interaction()}, which appear in subsequent
decompositions under their registered code. CVSS, AIVSS, and FAIR
have no equivalent surface; AIVSS's "single threat multiplier
$\mathit{ThM}$" applies once after a linear AARF sum and is not
per-pair.

\paragraph{Ablation: do multipliers change operator-visible buckets?}
We measure whether the multiplier step crosses bucket boundaries
relative to the additive-only baseline. The harness
(\btt{examples/interaction\_multiplier\_ablation.py}) scores six
agent definitions twice --- once with \btt{disable\_interactions=True}
(additive only) and once with multipliers enabled --- and reports
the buckets either side. Results in \cref{tab:ablation}: two of six
cross a boundary attributable specifically to the multiplier step
(rows 2 and 4: \textbf{high}$\to$\textbf{critical} via
\texttt{code\_exec\_with\_user\_input} at 1.5$\times$, and
\textbf{medium}$\to$\textbf{high} via \texttt{untrusted\_chain} at
1.2$\times$). Two saturate at the additive ceiling already
(governance\,=\,\texttt{none} alone contributes $+30$); for those,
the multiplier records a code on the audit trail but does not
change the bucket. The benign control fires no multiplier ---
no false positives in this scenario set.

\begin{table}[h]
\centering
\caption{Interaction-multiplier ablation. Each row scored twice:
additive-only baseline, then with multipliers. Two of six cases
cross a bucket attributable to the multiplier step. ``Add'' is the
additive score; ``Final'' is after multiplier $\times$ capping at
100. Single workstation; deterministic across runs.}
\label{tab:ablation}
\scriptsize
\begin{tabular}{lrlrlr}
\toprule
\textbf{Scenario}                          & \textbf{Add} & \textbf{AddBkt} & \textbf{Final} & \textbf{FinBkt} & \textbf{Mult} \\
\midrule
\texttt{autonomous\_writer\_in\_prod}      & 86  & critical & 100 & critical & 1.30 \\
\texttt{code\_exec\_user\_input\_hil}      & 63  & high     & 94  & critical & 1.50 \\
\texttt{self\_hosted\_pii\_hil}            & 64  & high     & 77  & high     & 1.20 \\
\texttt{untrusted\_chain\_hil}             & 57  & medium   & 68  & high     & 1.20 \\
\texttt{replit\_style\_autonomous}         & 100 & critical & 100 & critical & 1.95 \\
\texttt{benign\_readonly\_hil}             & 34  & medium   & 34  & medium   & 1.00 \\
\bottomrule
\end{tabular}
\end{table}

The honest takeaway: multipliers visibly move buckets on
\emph{moderate-risk} agents with human-in-the-loop discipline but a
single concentrated risk concentration (untrusted chain, code-exec
with user input, etc.). For already-critical configurations they
add an audit code but do not change the operator-visible bucket
because clamping at 100 dominates. The ablation justifies the
multiplier step's presence on the moderate-risk middle of the
distribution --- exactly where bucket boundaries matter for
review-queue triage.

\subsection{Score buckets}

The headline score is bucketed for human consumption:

\begin{itemize}
\item $[0, 29]$: \textbf{low} --- safe to run unattended in sandboxed contexts
\item $[30, 59]$: \textbf{medium} --- routine production with standard monitoring
\item $[60, 84]$: \textbf{high} --- requires elevated review; surfaces in operator queues
\item $[85, 100]$: \textbf{critical} --- requires explicit governance approval; auto-gated by the action gate in default policy
\end{itemize}

Bucket thresholds are tenant-configurable through the only-tighten
algebra (\cref{sec:fedrec}); tenants can lower the "critical"
threshold (treating more agents as critical), never raise it.

\subsection{Worked example: scoring the loan-decisioning fleet}
\label{subsec:scoring-example}

Applying the factor model to the four-agent fleet introduced in
\cref{subsec:running-example}: the \textbf{Risk Review Agent} scores
\textbf{77 (high)}: write tool (\btt{flag\_for\_review}, +6),
admin-gated tool (\btt{override\_decision}, +12), data sensitivity
\texttt{financial} (MAX, +20), governance mode \texttt{none}
(autonomous, +30, the largest single factor in the additive
model), decision-influence delegate, and the
\btt{admin\_tool\_AND\_financial\_data} interaction multiplier
fires at $\times 1.25$. The \textbf{Loan Triage Agent} scores
\textbf{61 (high)}: no admin tools of its own, but the
delegation-trust premium adds 8 for delegating to a high-bucket
Risk Review Agent. The \textbf{Document Verification Sub-Agent}
scores \textbf{42 (medium)}: read-only KYC lookups, no write
authority, but \texttt{third\_party} provenance (+20) and
\texttt{pii} data sensitivity (+15) keep it above \texttt{low}.
The \textbf{OFAC Screening Sub-Agent} scores \textbf{34 (medium)}:
read-only, \texttt{enterprise} model, but \btt{us\_classified}
data class (sanctions list, +25) and
\btt{regulator-mandated-tool} status. All four scores include the
auditable per-factor breakdown returned by \texttt{score\_agent()};
no factor exceeds its cap and no interaction multiplier fires
silently.

\paragraph{Second domain: clinical triage fleet} (\cref{appx:clinical-scoring}).
The parallel clinical fleet from \cref{subsec:second-domain} scores
on the identical pipeline: Diagnosis Suggestion 84/critical,
HL7 Lookup 56/medium, Triage Orchestrator 70/high (the Triage
Orchestrator's bucket comes from the same delegation-trust premium
that lifts the bank's Loan Triage). Domain changes; pipeline does
not.

\subsection{Calibration and limitations}

The default factor weights are derived from a synthesis of public
governance literature~\cite{NISTAIRMF,EUAIAct2024,ISO42001,owasp_aivss_2025},
internal red-team campaigns (\cref{sec:eval}), and expert judgment.
We do not claim the defaults are optimal. What the framework
guarantees is that defaults are \emph{auditable} (every factor and
weight in the Apache~2.0 SDK release) and
\emph{adjustable} (tenants tighten via override; signed
recommendations propose cross-fleet tightenings). Empirical
calibration from collector telemetry is the most consequential piece
of future work (\cref{sec:limitations}).

\begin{figure}[t]
\centering
\includegraphics[width=\columnwidth]{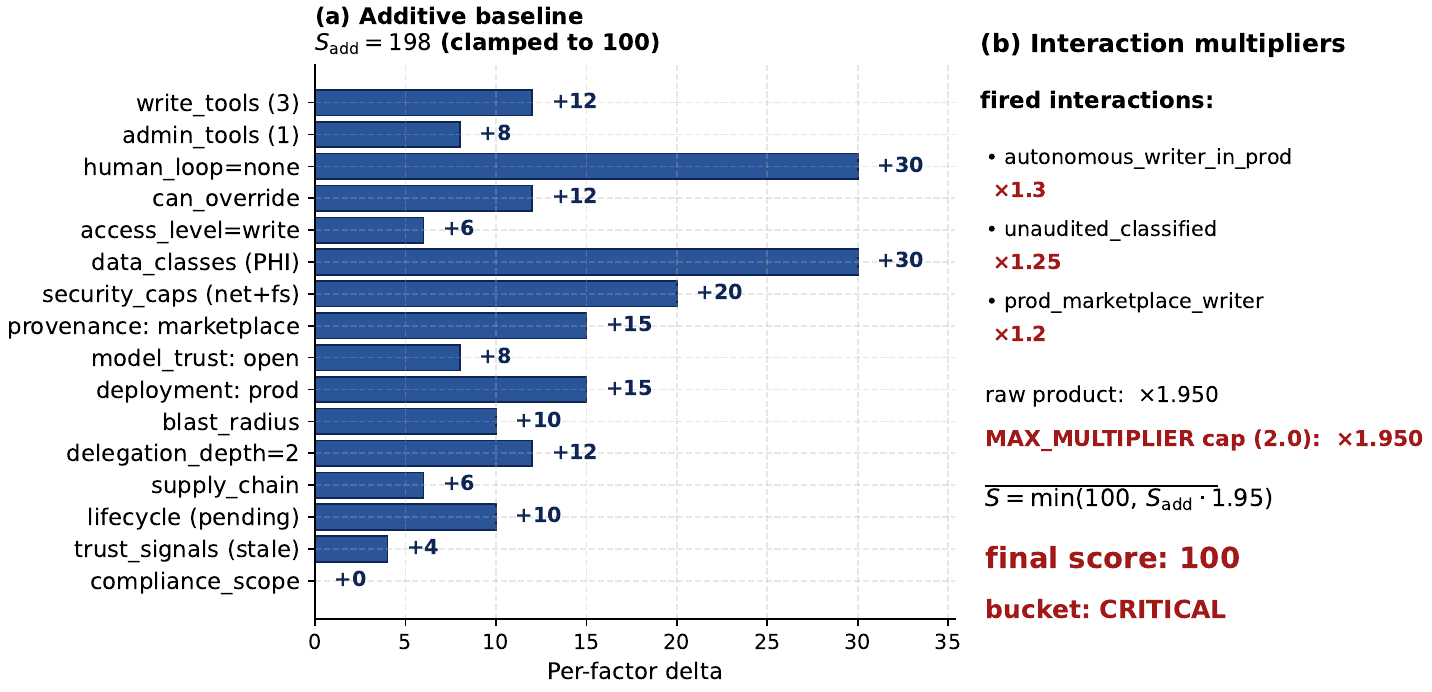}
\caption{Static score decomposition for a sample agent. The additive
sum is on the left; the multiplicative amplification step with fired
interaction codes is on the right.}
\label{fig:decomp}
\end{figure}

\section{Dynamic Rogue Signals}
\label{sec:rogue}

Static scoring (\cref{sec:static}) measures what an agent
\emph{could} do given its definition. It does not measure what an
agent has \emph{tried} to do at runtime that it should not have.
Two agents with identical static scores can have radically different
risk profiles in practice: one sits idle, the other has attempted to
call a privileged tool fifty times in the past hour. KYA bridges this
gap with a small set of runtime signals that aggregate observed
misbehavior into a dynamic score delta.

\subsection{Signal sources}

KYA observes six classes of runtime misbehavior, each instrumented
at the point of detection in the agent loop, counted via Prometheus,
recorded to the evidence chain, and surfaced through the
rogue-detection API:
\textbf{tool RBAC refusals} --- Layer~3 blocks a tool call whose
required roles the caller doesn't have (counter:
\btt{veldt\_tool\_rbac\_refusals\_total});
\textbf{out-of-scope tool attempts (OOS)} --- agent invokes a tool
not in its sanctioned \texttt{tools} list (the strongest single
signal; intent to act outside the configured capability surface;
counter: \btt{veldt\_agent\_oos\_tool\_attempts\_total});
\textbf{governance blocks} --- action gate vetoes an action on
content-safety, PII, or policy grounds (counter:
\btt{veldt\_governance\_action\_gate\_total} with verdict=block);
\textbf{cross-tenant attempts} --- tool call carrying a
$\mathit{tenant\_id}$ outside the agent's scope, typically from
prompt injection or a compromised upstream caller (counter:
\btt{veldt\_agent\_cross\_tenant\_attempts\_total});
\textbf{data leakage} --- output containing data of an
unsanctioned class, caught by the action gate's class-tagged
output classifier (counter:
\btt{veldt\_agent\_data\_leak\_total});
\textbf{policy violations} --- jailbreak, harmful output, refusal
failure, or successful prompt injection (counter:
\btt{veldt\_agent\_policy\_violations\_total}).
Signal flow is one-directional: gates emit,
the rogue module consumes, but gates do not read rogue state
for real-time verdicts (preventing the bias feedback loop where
prior misbehavior gates legitimate later actions). Per-source
instrumentation points are in \cref{appx:rogue-detail}.

\subsection{Aggregation and bucketing}

The rogue score \texttt{rogue\_score(report)} is a bounded integer
in $[0, 50]$ that adds to the static score. The score is computed
from per-signal counts within a configurable window (default 24
hours), with logarithmic compression to prevent a single very active
agent from saturating the score on volume alone:

$$\mathit{rogue\_score} = \min\left(50, \sum_{s \in \mathit{signals}} w_s \cdot \log_2\big(1 + c_s\big)\right)$$

where $c_s$ is the count of signal $s$ in the window and $w_s$ is
its weight. This rewards \emph{breadth} (multiple signal kinds)
over \emph{depth} (one signal repeated): an agent that attempted one
out-of-scope tool call and emitted one PII leak is treated as more
suspicious than one that attempted ten out-of-scope tool calls of
the same kind.

Validation in \cref{fig:rogue-heatmap} (appendix): rogue-event
clusters across a 40-agent fleet co-occur with annotated attack
epochs during a 24-hour Garak~\cite{garak} red-team campaign.

\subsection{Quality signals}

Alongside the four rogue counters, KYA tracks three
\emph{quality} signals~--- weaker per-event but valuable in
aggregate:

\paragraph{Hallucination signals.} Tagged from downstream evaluation:
the agent produced a response that an LLM-judge
(\cref{sec:related}) or a citation-verification pass marked as
unsupported by retrieved context.

\paragraph{QA irrelevance.} The response failed a relevance check
against the original query --- the agent ``went off topic.''

\paragraph{Prompt injection attempts.} The user input or upstream
context contained signatures of a known prompt-injection
template~\cite{greshake2023notwhat}. Whether the injection succeeded is captured
separately under policy violations.

Quality signals contribute a smaller delta than rogue signals
($w_q < w_r$) because their false-positive rate is higher, but they
are useful as early indicators: a rising hallucination rate often
precedes a discovered drift event by hours or days.

\subsection{Burst anomaly detection}

The realtime module (\texttt{realtime.py}) maintains Valkey-backed
sliding windows over signal kinds and publishes burst alerts when
the per-minute rate exceeds the trailing-hour baseline by a
configurable factor. Alerts flow through pub/sub to subscribed
operators; the API exposes \texttt{detect\_burst\_anomalies()} for
batch queries. This is the only path through KYA that runs against
in-memory state rather than the durable evidence chain --- burst
detection is intended for live triage, not after-the-fact audit.

\subsection{Integration with the static score}

The total agent-risk view in operator dashboards is
$\mathit{total} = \mathit{static} + \mathit{rogue}$, clamped to
$[0, 100]$. A purely-defined high-risk agent that has never
misbehaved retains its static score; a moderately-defined agent that
attempts privilege escalation accumulates rogue deltas and crosses
into higher buckets at runtime. This composition is the central
distinction from purely static governance: the principal-trust posture
\emph{moves} with behavior.

\subsection{KYP --- Know Your Principal}
\label{subsec:kyp}

KYA's dynamic-trust layer is built around a unified principal-trust
taxonomy we call \textbf{KYP} (Know Your Principal). KYP generalizes
the per-user trust scoring common in fraud-detection
systems~\cite{chaffer2025kya,janani2025humanmachineblur} to a single
schema-level abstraction over three principal kinds:

\begin{itemize}
\item \texttt{user} --- a human user (UUID identifier)
\item \texttt{agent} --- another AI agent (\texttt{agent\_key} identifier)
\item \texttt{service\_account} --- an automated service, cron job, batch pipeline, or harness
\end{itemize}

All three principal kinds share the same scoring mechanics:
initial trust score 50 (\btt{STARTING\_TRUST},
\btt{kya/users.py}), bounded $[0, 100]$, signal-driven
decrement, time-decayed recovery. Default per-signal weights
(\btt{SIGNAL\_DELTAS}, \btt{kya/users.py}; tenant-overrideable
via the only-tighten algebra of \cref{sec:fedrec}):
\btt{oos\_tool} $-3$, \btt{rbac\_refusal} $-2$,
\btt{governance\_block} $-2$, \btt{data\_leak} $-10$,
\btt{cross\_tenant} $-15$. Scores map to four buckets via
\btt{bucket\_for\_trust} (\btt{kya/users.py}):
\texttt{trusted} ($\geq 75$), \texttt{neutral} ($\geq 40$),
\texttt{risky} ($\geq 15$), \texttt{blocked} ($< 15$). All three
principal kinds share one storage table
\btt{kya\_principal\_trust} (\btt{kya/principals.py:215})
keyed by $(\mathit{tenant\_id}, \mathit{principal\_kind},
\mathit{principal\_id})$ with one signal-application code path
(\cref{fig:kyp-schema}).

\begin{figure}[t]
\centering
\resizebox{\columnwidth}{!}{

\begin{tikzpicture}[
  font=\sffamily\footnotesize,
  table/.style={draw, rounded corners=2pt, minimum width=46mm,
                align=left, inner sep=4pt, fill=blue!8, draw=blue!50},
  hdr/.style={font=\sffamily\small\bfseries, anchor=west},
  col/.style={font=\sffamily\footnotesize\ttfamily, anchor=west},
  disc/.style={font=\sffamily\scriptsize\itshape, anchor=west, text=blue!60!black},
  kind/.style={draw, rounded corners=1.5pt, minimum width=24mm, minimum height=6mm,
               align=center, font=\sffamily\scriptsize\ttfamily, fill=green!10, draw=green!60!black},
  arrow/.style={-{Stealth}, thick, gray!70},
]

\node[table] (tbl) at (0,0) {%
  \begin{tabular}{@{}l@{\hspace{4pt}}l@{}}
    \multicolumn{2}{l}{\textbf{kya\_principal\_trust}} \\[2pt]
    \texttt{tenant\_id}       & \scriptsize uuid (PK) \\
    \texttt{principal\_kind}  & \scriptsize varchar (PK, discriminator) \\
    \texttt{principal\_id}    & \scriptsize text (PK) \\
    \texttt{trust\_score}     & \scriptsize int $\in [0, 100]$ \\
    \texttt{signal\_counts}   & \scriptsize jsonb \\
    \texttt{last\_signal\_at} & \scriptsize timestamptz \\
    \texttt{updated\_at}      & \scriptsize timestamptz \\
  \end{tabular}};

\node[disc, below=2mm of tbl, text width=46mm, align=left]
  {\texttt{PRINCIPAL\_KINDS} closed-set whitelist at
   \texttt{kya/principals.py:81}};

\node[kind, left=18mm of tbl, yshift=14mm]  (user)  {user};
\node[kind, left=18mm of tbl, yshift=0mm]   (agent) {agent};
\node[kind, left=18mm of tbl, yshift=-14mm] (svc)   {service\_account};


\node[font=\sffamily\scriptsize\itshape, right=10mm of tbl, text width=38mm,
      align=left, fill=yellow!10, draw=yellow!50!brown, rounded corners,
      inner sep=4pt] (share)
  {\textbf{One signal-application code path}\\[2pt]
   $\bullet$ same starting score (50)\\
   $\bullet$ same bounds $[0,100]$\\
   $\bullet$ same time-decay\\
   $\bullet$ same Valkey-window mirror\\[2pt]
   No per-kind branching.};

\draw[-{Stealth}, thick, blue!60!black] (agent.east) -- (tbl.west);

\draw[-{Stealth}, thick, orange!70!black] (tbl.east) -- (share.west);

\node[font=\sffamily\small\itshape, above=2mm of tbl, text=blue!60!black]
  {Single table $\to$ one discriminator $\to$ three principal kinds};

\node[font=\sffamily\tiny, text=red, below right=0pt and 0pt of svc.south west]
  {[v20 L-removed]};

\end{tikzpicture}}
\caption{KYP unified principal-trust schema. A single table
(\texttt{kya\_principal\_trust}) carries trust scores for all three
principal kinds; the closed-set \texttt{principal\_kind} discriminator
selects users, agents, or service accounts without branching the
scoring code path. Per-kind subsystems are visible only in the
choice of identifier domain; signal mechanics, decay, and storage
are shared.}
\label{fig:kyp-schema}
\end{figure}
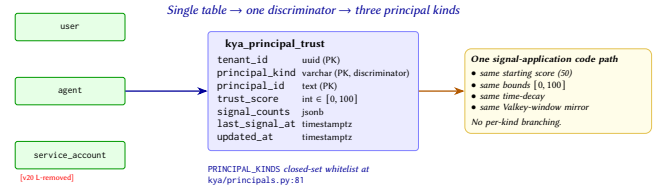

The closest prior art each unifies at a different layer.
Chaffer~\cite{chaffer2025kya} coins "Know Your Agent" but for an
agent-only Web3 trust framework with no human or service principal.
Janani~\cite{janani2025humanmachineblur} argues for a human-machine
identity continuum at the policy-and-risk-evaluation layer with no
data model. Okta's NHI initiative~\cite{okta_nhi_unified_2026}
unifies humans, services, and agents at the platform/UI layer while
keeping separate scoring subsystems per identity kind. Microsoft
Agent Governance Toolkit~\cite{microsoft_agt_2026} provides per-agent
trust scoring in Agent Mesh and relies on Entra ID upstream for
humans, with no unified scoring substrate spanning both. KYP's
delta is the relational-schema unification: a single table, one
discriminator (\texttt{PRINCIPAL\_KINDS} at \texttt{kya/principals.py:81}),
shared scoring code across kinds.

The \texttt{PRINCIPAL\_KINDS} tuple is intentionally closed at
three; specialized actors (eval harnesses, red-team campaigns,
batch ingest, autonomous SQL pipelines) instantiate as
\texttt{service\_account} with metadata in the \texttt{attributes}
JSON column. Caller-supplied strings cannot expand the keyspace
(consistent with the closed-set whitelist discipline of
\cref{sec:disciplines}).

\subsection{Actor-agent attribution for autonomous fan-out}
\label{subsec:actoragent}

When orchestrator agent $A$ triggers downstream sub-agents $B/C/D$
and any sub-agent fires a rogue signal at runtime, the signal is
tagged at emission with \texttt{actor\_agent\_key=}$A$, so $A$'s
principal-trust counter is debited --- not $B$'s, $C$'s, or $D$'s.
This is complementary to the static-score delegation premium of
\cref{sec:static} (where $A$'s risk score factors in risky
delegates): the static axis penalizes \emph{configured} delegation
to risky agents; the dynamic axis registers \emph{actual}
delegate misbehavior in $A$'s runtime principal-trust counter.

The default \btt{actor\_agent\_key=agent\_key} convention is
hard-wired at three SDK entry points (Claude Agent SDK hook,
LangChain callback handler, OpenAI Agents \texttt{RunHooks}) so
attribution works without per-customer wiring. The same convention
is documented in the OTLP bridge's KYA-hooks client.

Three distinct deltas from prior work: banking AML / Horton
capability~\cite{miller2007horton} do \emph{audit-time}
attribution; AgenTracer~\cite{zhang2025agentracer} identifies the
\emph{offending sub-agent} post-hoc; Microsoft AGT propagates
trust ceilings \emph{downward}~\cite{microsoft_agt_2026}. KYP
debits the \emph{orchestrator} at runtime --- complementary to all
three. Legal-theoretic motivation: Gabison and
Xian~\cite{gabison2025liability} (REALM at ACL 2025).

\section{Evidence Chain}
\label{sec:evidence}

KYA's evidence log is the system's source of truth for after-the-fact
audit. The log stores the actual forensic payloads regulators need
to reconstruct what an agent did: the prompts the agent received,
the responses it produced, the tool calls it made
\emph{with their arguments} (e.g.\ the literal SQL it tried to run),
delegation messages between agents, human-in-the-loop approval
decisions, and system context. Closed-set kinds:
\texttt{prompt}, \texttt{response}, \texttt{tool\_call},
\texttt{tool\_result}, \texttt{delegation\_message},
\texttt{hil\_decision}, \texttt{system\_message}
(\texttt{VALID\_EVIDENCE\_KINDS}). Where the scoring tables
(\texttt{kya\_invocations}, \texttt{kya\_principal\_trust}) hold
the counts, this table holds the proof. Each row is appended to a
per-(tenant, invocation) HMAC-chained sequence.

The HMAC-chain construction itself is well-trodden prior art: Bellare
and Yee (1997)~\cite{BellareYee1997} introduced forward-secure
hash-chained MAC audit logs, Schneier and Kelsey
(1998)~\cite{schneierkelsey1998logs} provided the canonical USENIX
treatment, and Crosby and Wallach
(2009)~\cite{crosbywallach2009tamper} explored the Merkle-tree
alternative. We do not claim novelty in the chain construction. Our
contribution is the \emph{deployment shape}: per-(tenant, invocation)
parallel chains with dialect-aware concurrency serialization, plus a
type-marked canonicalization that closes a collision-attack vector
the naïve \texttt{default=str} JSON encoder leaves open. We also
explicitly acknowledge the v1 single-key limitation and the
Sigstore-style anchoring upgrade path; this honesty distinguishes
KYA's claim from contemporaneous work (notably
Aegis~\cite{Aegis2026}) that asserts a stronger primitive.

\subsection{Chain construction (background)}

Let $K$ be a secret HMAC key held by the tenant operator, $H$ be
HMAC-SHA256~\cite{Bellare96}, and $\Vert$ denote byte concatenation. A KYA evidence
chain $L = (e_0, e_1, \dots, e_n)$ is a sequence of events where each
$e_i$ carries:

\begin{itemize}
\item $\mathit{seq}_i \in \mathbb{N}$: monotonically increasing sequence number
\item $\mathit{kind}_i \in \mathcal{K}$: event kind from a validated
  closed vocabulary (\texttt{VALID\_EVIDENCE\_KINDS}, 7 entries)
\item $\mathit{payload}_i$: canonical-JSON serialized event body
\item $h_i$: chain hash, defined as
$$h_i = H_K\big(h_{i-1} \,\Vert\, \mathsf{canonical\_json}(e_i)\big),$$
with the seed $h_{-1}$ set to a domain-separation tag.
\end{itemize}

For tamper-evidence we refer to the standard result of Bellare and
Yee (1997)~\cite{BellareYee1997}: any modification to a prior entry
$e_i$ without knowledge of $K$ produces a chain that fails
\texttt{verify\_chain} with probability at most
$\mathsf{Adv}^{\mathrm{prf}}_H$.

\subsection{Deployment shape: per-(tenant, invocation) parallel chains}

Unlike a single global chain, KYA maintains a parallel chain per
(tenant, invocation) pair. Three properties motivate this:

\paragraph{Cross-tenant isolation.} A compromised chain in tenant
$T_1$ cannot break the chain in $T_2$. Pruning under one tenant's
retention policy does not require coordination with any other
tenant. Forensic audit on tenant $T_1$ does not expose data from
$T_2$.

\paragraph{Bounded fan-in for concurrent agent invocations.} Modern
agent frameworks fire many concurrent invocations per tenant. A
single chain head would serialize all of them; per-invocation chains
allow $N$ parallel writers across $N$ invocations without
contention.

\paragraph{Pruning preserves verification of surviving chains.}
Retention-driven pruning (\cref{subsec:retention}) creates a "clean
cut" the verifier distinguishes from tampering: an expected
chain-head mismatch at the first surviving entry of a pruned chain
is reported as a clean cut, not a tamper.

\subsection{Dialect-aware concurrency primitive}
\label{subsec:concurrency}

Concurrent writers within the same (tenant, invocation) still risk a
chain fork: two writers reading the same tail $h_{n-1}$ and producing
distinct $e_n, e'_n$ both signing against $h_{n-1}$. KYA's evidence
module serializes per-(tenant, invocation) writes via a
dialect-specific primitive:

\begin{itemize}
\item \textbf{PostgreSQL:} \btt{pg\_advisory\_xact\_lock} keyed on
  a hash of \texttt{`kya:evidence:'}, the tenant id, and the
  invocation id. The key is derived from the (tenant, invocation)
  identifiers, so the lock functions \emph{even on empty chains}
  where no tail row exists to lock. The lock is transaction-scoped
  and released on commit/rollback.
\item \textbf{MySQL:} \texttt{SELECT ... FOR UPDATE} on the most-recent
  chain entry within the writer's transaction.
\item \textbf{SQLite and DuckDB:} the module documents and enforces a
  "one writer per invocation" contract at the application layer,
  since these dialects do not provide the row-locking semantics the
  primitive requires.
\end{itemize}

The closest prior art handles concurrency by avoiding it.
Schneier-Kelsey~\cite{schneierkelsey1998logs} and
Bellare-Yee~\cite{BellareYee1997} defer concurrency to operator
practice and assume a single-writer ingestion process. Sigstore
Rekor~\cite{Newman2022Sigstore} serializes appends through a single
coordinator. The Microsoft Agent Governance Toolkit~\cite{microsoft_agt_2026}
documents a single-writer audit-sink contract for its FileAuditSink.
KYA's delta is the dialect-aware concurrency primitive applied to
per-(tenant, invocation) parallel chains: PostgreSQL writers serialize
via \texttt{pg\_advisory\_xact\_lock}, MySQL writers via
\texttt{SELECT~FOR~UPDATE} on the tail row, SQLite/DuckDB writers
via an enforced single-writer contract at the application layer.
The combination of parallel chains with primitive-level serialization
removes the global serialization point that earlier systems retain
(\texttt{kya/evidence.py:436--462}).

\paragraph{Type-marked canonicalization, retention floors, signing-key
resolution, and verification semantics} (\cref{appx:evidence-detail}):
canonical JSON (RFC~8785~\cite{RFC8785}) extended with
type-marked wrapping to close a collision-attack vector; per-regime
retention floors (GDPR/HIPAA 6yr, NYDFS 5yr, EU AI Act 7yr, ISO
27001, FedRAMP, ITAR, NATO, EU classified, plus 6 more) enforced
at \btt{prune\_expired\_evidence()}; three-tier signing-key
resolution (KMS/Vault/HSM provider, env-mounted key, dev fallback
with one-time WARN); $O(1)$ per-event and $O(n)$ full-chain
verification with three distinct failure modes (payload tamper,
chain break tamper, clean cut). All standard pieces; the only-novel
deployment-shape claims are above.

\subsection{Worked example: one loan application through the chain}
\label{subsec:chain-example}

Continuing the running example (\cref{subsec:running-example}), a
single loan application produces an invocation chain on the order of
20--30 entries: $e_0$ the originating prompt from the underwriter (kind
\texttt{prompt}), $e_{1\text{--}3}$ the Loan Triage Agent's delegation
calls to the three sub-agents (kind \texttt{tool\_call}, with
\texttt{actor\_agent\_key=triage} on each so any sub-agent rogue signal
debits the orchestrator), $e_4$ a Document Verification result tagged
with sensitivity \texttt{pii} and a 6-year retention floor under
NYDFS, $e_5$ an OFAC screening result tagged \texttt{us\_classified}
with the 25-year IL6 floor, $e_{6\text{--}8}$ Risk Review Agent
reasoning steps, $e_9$ the action-gate verdict (kind
\texttt{governance\_verdict}, \texttt{allow}), and $e_{10}$ the
underwriter-facing decision (kind \texttt{response}). Any post-hoc
attempt to alter the OFAC entry --- whether to hide a missed hit or
to fabricate one --- changes $h_5$ and breaks the chain at $h_6$;
\texttt{verify\_chain()} returns \emph{chain break tamper} at $e_6$.
The retention floor at $e_5$ (25 years) outlives any other entry,
forcing the whole invocation's chain head to survive that long even
if civilian-PII entries are pruned at year~6.

The same chain construction applies to the clinical fleet of
\cref{subsec:second-domain}: a single triage invocation produces a
similar 10--20 entry chain, but every entry tagged with
sensitivity \texttt{phi} inherits HIPAA's 6-year retention floor
(45 CFR 164.530(j)) and a breach-notification SLA of 60 days under
the HIPAA Breach Notification Rule. The regime-aware retention
selector (\cref{subsec:retention}) reads the data-class tags and
selects the longest applicable floor without per-domain code paths.

\paragraph{Honest scope.} The chain is verifiable by anyone holding
the HMAC key; multi-party verification (no single party can forge
the chain) requires threshold signatures over the root and is future
work (\cref{appx:evidence-detail}, \cref{sec:limitations}).

\subsection{Configuration drift detection}
\label{sec:drift}

Observability tracks agent \emph{behavior}; not agent
\emph{configuration}. A one-line edit to a system prompt silently
changes the agent's policy without leaving any observability trace.
KYA addresses this with (i)~content-addressed canonical hashing
over an explicit allowlist of 18 policy-bearing fields
(\texttt{system\_prompt}, \texttt{tools}, \texttt{model},
\texttt{access\_level}, \texttt{can\_override}, etc.); (ii)~an
append-only event-time + ingest-time snapshot history with
lineage-driven risk inheritance; and (iii)~a behavioral-drift signal
--- the mode-vs-config gap --- that records exercised
\texttt{human\_loop} mode at each invocation and flags
configurations whose declared mode diverges from the empirical
distribution (EU AI Act Art.~14 evidentiary requirement for
\emph{exercised} oversight). The detector returns a structural diff,
records to the evidence chain above, and triggers operator review.
Algorithm, lineage-decay constants, signature verification, and the
explicit out-of-scope drift classes (model-provider behavioral drift,
tool-environment drift, retrieval-corpus drift) are detailed in
\cref{appx:drift}.

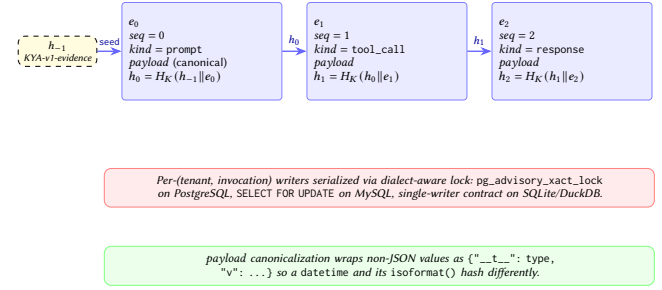
\begin{figure}[t]
\centering
\resizebox{\columnwidth}{!}{

\begin{tikzpicture}[
  font=\sffamily\footnotesize,
  entry/.style={draw, rounded corners=2pt, minimum width=32mm, minimum height=20mm,
                align=left, inner sep=4pt, fill=blue!8, draw=blue!50, text width=30mm},
  hash/.style={font=\sffamily\scriptsize\ttfamily, text=blue!60!black},
  arrow/.style={-{Stealth}, thick, gray!70},
  dst/.style={draw, dashed, rounded corners, fill=yellow!15, inner sep=3pt,
              align=center, font=\sffamily\scriptsize},
]

\node[dst] (dst) at (-2.6, 0) {$h_{-1}$\\\textit{KYA-v1-evidence}};

\node[entry, right=5mm of dst] (e0) {%
  \textbf{$e_0$}\\
  $\mathit{seq}=0$\\
  $\mathit{kind}=$ \texttt{prompt}\\
  $\mathit{payload}$ (canonical)\\
  $h_0 = H_K(h_{-1} \Vert e_0)$};

\node[entry, right=5mm of e0] (e1) {%
  \textbf{$e_1$}\\
  $\mathit{seq}=1$\\
  $\mathit{kind}=$ \texttt{tool\_call}\\
  $\mathit{payload}$\\
  $h_1 = H_K(h_0 \Vert e_1)$};

\node[entry, right=5mm of e1] (e2) {%
  \textbf{$e_2$}\\
  $\mathit{seq}=2$\\
  $\mathit{kind}=$ \texttt{response}\\
  $\mathit{payload}$\\
  $h_2 = H_K(h_1 \Vert e_2)$};

\draw[arrow, blue!60] (dst.east) -- node[above, hash] {seed} (e0.west);
\draw[arrow, blue!60] (e0.east) -- node[above, hash] {$h_0$} (e1.west);
\draw[arrow, blue!60] (e1.east) -- node[above, hash] {$h_1$} (e2.west);

\node[font=\sffamily\footnotesize\itshape, text width=110mm, align=center,
      fill=red!8, draw=red!50, rounded corners, inner sep=4pt] at (4,-2.8)
  {Per-(tenant, invocation) writers serialized via dialect-aware lock:
   \texttt{pg\_advisory\_xact\_lock} on PostgreSQL,
   \texttt{SELECT FOR UPDATE} on MySQL,
   single-writer contract on SQLite/DuckDB.};

\node[font=\sffamily\footnotesize\itshape, text width=110mm, align=center,
      fill=green!8, draw=green!50, rounded corners, inner sep=4pt] at (4,-4.4)
  {$\mathit{payload}$ canonicalization wraps non-JSON values as
   \texttt{\{"\_\_t\_\_": type, "v": ...\}} so a \texttt{datetime}
   and its \texttt{isoformat()} hash differently.};

\end{tikzpicture}}
\caption{KYA evidence chain construction. Per-(tenant, invocation)
parallel chains with dialect-aware concurrency serialization
(\cref{subsec:concurrency}) and type-marked canonical JSON
(\cref{subsec:canon}).}
\label{fig:chain}
\end{figure}

\section{Inbound Federation and Closed-Loop Adaptation}
\label{sec:fedrec}

KYA's risk-weight tables (\cref{sec:static}) ship with defaults
from expert judgment. Defaults are generic by construction. A write
tool in a clinical setting is not the same risk as the same tool in
logistics. Per-tenant overrides (\cref{subsec:overrides}) let
customers adapt weights locally, but local tuning misses cross-fleet
signal when attack patterns emerge across many customers.

KYA addresses both gaps with two coordinated mechanisms:
\textbf{federated weight recommendations} (vendor → fleet,
\cref{subsec:fedinbound}) and \textbf{in-tenant closed-loop adaptation}
(operator → local fleet, \cref{subsec:closedloop}). Both route apply
through the same gates, preserving the only-tighten invariant
(\cref{subsec:overrides}) and operator approval as default.

The signed-update transport itself is well-trodden prior art:
TUF~\cite{samuel2010tuf}, Uptane~\cite{kuppusamy2016uptane}, OPA
signed bundles~\cite{OPA_SignedBundles},
STIX/TAXII~\cite{STIX21_OASIS}, and Ioannidis et al.'s distributed
firewall~\cite{ioannidis2000distfw} all distribute signed policy
artifacts to enforcement points. Our contribution is not the
transport but the \textbf{four-gate apply pipeline with
operator-approval-as-default} and the formal only-tighten composition
algebra. The closed-loop in-tenant adaptation pipeline
(\cref{subsec:closedloop}) is an engineering discipline that inherits
those two mechanisms; we describe it for completeness, not as a
separate contribution.

\subsection{Tenant override algebra}
\label{subsec:overrides}

The only-tighten composition algebra is structurally adjacent to
Cedar's forbid-dominance theorem (Cutler et
al.~\cite{cutler2024cedar}, OOPSLA 2024), which proves in Lean that
\texttt{forbid} dominates \texttt{permit} under union composition for
a flat policy set. The same one-way property is also operationally
deployed in Istio AuthorizationPolicy~\cite{istio_authz} (CUSTOM $\succ$
DENY $\succ$ ALLOW, non-configurable lattice) and GCP IAM Deny
Policies~\cite{gcp_deny_policies} (descendant deny narrows ancestor
permits; never broadens). KYA does not claim novelty in the
one-way-tightening property itself. Our contribution is the
\textbf{three-channel hierarchical composition}: platform default
$\oplus$ tenant override $\oplus$ signed external recommendation,
with a soundness lemma over the multi-authority hierarchy. The signed
external-recommendation channel as a third party to the meet is the
genuinely new dimension; no prior work surveyed formalizes it.

Let $W_0: (\mathit{scope}, \mathit{key}) \to \mathbb{N}$ be the
platform default weight function and let $W_t$ denote tenant $t$'s
effective weight function and $\succeq$ be the order $w' \succeq w
\iff w' \ge w$ (higher = tighter). A tenant override is
\emph{only-tighten} if every $(scope, key)$ value dominates the
platform default under $\succeq$ (Definition~\ref{def:tighten}, App.
\ref{appx:tighten}). The SDK enforces this at write time:
\texttt{tenant\_weights.set\_override} raises
\texttt{OverrideLoosensError} on any tenant-scoped value below the
platform default (\texttt{tenant\_weights.py:189}). Composition is
sound: the effective tenant weight function $W_t$ satisfies $W_t
\succeq W_0$ at all times and is monotone non-decreasing within the
tenant scope (Lemma~\ref{lem:tighten}, proof sketch in App.
\ref{appx:tighten}). The invariant binds tenants, not platform
admins --- the \texttt{\_check\_only\_tighten} function returns early
when \texttt{tenant\_id=None}; platform admins retain the
management-decision right to lower the platform default itself.

A runnable witness
(\btt{examples/three\_channel\_composition\_witness.py}) exercises
all three channels across six update attempts (platform raise,
tenant tighten, tenant loosen, signed-rec tenant-tighten, signed-rec
tenant-loosen, signed-rec platform-lower). Every attempt matches its
expected accept/block outcome and Lemma~\ref{lem:tighten} holds in
the resulting state: with platform $W_0(\texttt{pii}){=}12$ and
tenant $W_t(\texttt{pii}){=}30$, we have $W_t \geq W_0$ ($30 \geq 12$).
The signed-recommendation channel routed at the tenant level
($W_r \to W_t$) traverses the same \texttt{\_check\_only\_tighten}
guard as direct tenant overrides --- so a compromised collector key
cannot loosen a tenant's effective weight even though the inbound
pipeline accepted the message.

\subsection{The four-gate apply pipeline}
\label{subsec:fedinbound}

Inbound federation is \emph{off by default}; a tenant operator
explicitly calls \texttt{enable\_inbound(collector\_url,
trusted\_keys)} (\texttt{kya/inbound.py}) to opt in. Until then,
no external recommendation reaches the SDK. Once enabled, every
signed recommendation traverses four sequential gates before any
effect on a tenant's weight table.

\begin{description}
\item[Gate 1: Ed25519 signature verification.] The recommendation
  envelope is verified against pinned trust anchors via
  \btt{\_inbound\_signing.verify\_envelope}. Canonical-JSON
  serialization (signature field stripped, sorted keys, no
  whitespace, \btt{ensure\_ascii=True}). Multi-anchor pinning is
  first-class: the \btt{KYA\_INBOUND\_PUBLIC\_KEY} environment
  variable accepts a comma-separated list of
  \btt{<keyid>:<base64-pubkey>} entries, so operators can pin
  current + next-quarter keys for invisible rotation, or pin a
  self-operated gateway key for air-gapped deployments. A signature
  verification failure aborts processing immediately and is never
  followed by an "apply anyway" fallback.

\item[Gate 2: Expiration check at persist time.] Each recommendation
  carries an \texttt{expires\_at} timestamp. The check is at
  \emph{persist time}, not only at signing-verify, so recommendations
  signed in advance cannot bypass expiry by traveling slowly through
  the fetch pipeline (\texttt{inbound.py:166}).

\item[Gate 3: Only-tighten composition.] Apply routes through
  \texttt{tenant\_weights.set\_override}, which enforces
  Definition~\ref{def:tighten} and Lemma~\ref{lem:tighten}. A
  recommendation that would lower the effective weight raises
  \texttt{OverrideLoosensError}; the recommendation is rejected and
  the operator notified.

\item[Gate 4: Operator approval (default).]
  Every recommendation lands as \texttt{status='pending'} by default;
  an operator explicitly approves via the SDK's
  \btt{approve\_recommendation} path. Optional auto-apply is scoped:
  the customer configures an \btt{auto\_apply\_allowlist} of
  $(scope, key)$ tuples, and only recommendations matching the
  allowlist are auto-applied --- and even then through the same
  \btt{set\_override} path, so gates~1--3 still hold.
\end{description}

A \emph{fifth latent gate} guards against race conditions: the
auto-apply UPDATE is status-guarded
(\texttt{WHERE status='pending'} at \texttt{inbound.py:248}) so a
re-fetch cannot roll back an operator-applied row to an
\texttt{auto\_applied} status.

The contribution is the composition of gates 1--4 (plus the latent
race-guard) in this specific order with operator-approval as the
default. No prior signed-update distribution system surveyed
(\cref{sec:related}) composes all four with operator-approval as the
default --- the closest analogue is the KeyNote-credential
composition with local POLICY in Ioannidis et
al.\ (2000)~\cite{ioannidis2000distfw}, which lacks an explicit
operator-approval step.

\subsection{Closed-loop in-tenant adaptation (supporting discipline)}
\label{subsec:closedloop}

Local incidents within a tenant's deployment produce adaptation
candidates via the \texttt{kya.feedback} pipeline: a resolved
critical governance incident triggers
\btt{feedback.propose\_from\_incident}, which analyzes which
factors fired against the incident-producing agent and proposes
weight bumps. Suggestions land in \btt{kya\_weight\_suggestions}
as \texttt{status='pending'}; an operator reviews and approves
(or rejects); approved suggestions are applied via
\btt{tenant\_weights.set\_override}, inheriting the only-tighten
algebra and the change-audit trail. The module's docstring states
the safety discipline: "\emph{Never auto-tune. Auto-applying weights
based on incidents creates a feedback loop where one false-positive
incident silently weakens the governance model. Human-in-the-loop is
non-negotiable here.}"

This is an engineering discipline, not a research contribution: the
mechanism is straightforward (propose $\to$ stage pending $\to$
operator approval $\to$ apply through the four-gate pipeline).
Structurally it parallels FDA-PCCP~\cite{fda2024pccp} envelope
review and draws on AI corrigibility~\cite{soares2015corrigibility,nayebi2025corrigible}
and active-learning oracle-gating~\cite{schubert2023deep}.
Counterposition: Cloudflare WAF
ML~\cite{cloudflare_waf_ml_2022} ships vendor-trained classifier
updates globally with no customer-side approval; IBM
SOAR~\cite{ibm_soar_2024} gates response actions but not policy
weight adaptations.

\paragraph{Worked example: sector-wide OFAC tightening.}
Continuing the running example, a peer bank's sanctions-evasion
incident traces to a misclassified OFAC hit. The KYA collector
ingests aggregated counts, the operator drafts a recommendation
tightening the data-class multiplier on \texttt{us\_classified}
for \texttt{ofac\_lookup} ($1.20 \to 1.30$), signs it, and
distributes. Our bank's KYA verifies signature (Gate~1), checks
expiry (Gate~2), confirms $1.30 \geq$ current override (Gate~3
only-tighten), and stages for operator review (Gate~4). The risk
officer approves; the OFAC sub-agent's score shifts $34 \to 38$
and the \texttt{flag\_for\_review} threshold tightens one bucket.
No raw application data, principal identity, or completion payload
crosses the organizational boundary --- only counts upward, signed
weight deltas downward.

\subsection{Threat-model summary}

\textbf{Compromised collector private key}: an attacker forging
recommendations can over-restrict (denial of service) but cannot
loosen below the platform default --- the only-tighten algebra
constrains blast radius to friction, not exposure; key rotation via
\texttt{key\_id} + multi-anchor pinning lets operators distrust a
compromised key.
\textbf{Replay}: \texttt{expires\_at} bounds replay; the SDK rejects
expired recommendations at persist time.
\textbf{Compromised platform admin and insider-threat collusion}:
out of scope (\cref{appx:limitations}); KYA provides tenant-side
immutability against vendor recommendations + compromised admin,
not absolute immutability.

\begin{figure}[t]
\centering
\resizebox{\columnwidth}{!}{

\begin{tikzpicture}[
  font=\sffamily\footnotesize,
  cust/.style={draw, rounded corners=2pt, minimum width=24mm, minimum height=12mm,
               align=center, inner sep=3pt, fill=blue!8, draw=blue!50},
  vend/.style={draw, rounded corners=2pt, minimum width=28mm, minimum height=12mm,
               align=center, inner sep=3pt, fill=orange!10, draw=orange!60!black},
  gate/.style={draw, rounded corners=2pt, minimum width=20mm, minimum height=10mm,
               align=center, inner sep=2pt, fill=green!12, draw=green!60!black,
               font=\sffamily\scriptsize},
  arrow/.style={-{Stealth}, thick, gray!70},
]

\node[cust] (c1) at (0,2.6) {Customer A\\\scriptsize \texttt{veldt-kya}};
\node[cust] (c2) at (0,0.6) {Customer B\\\scriptsize \texttt{veldt-kya}};
\node[cust] (c3) at (0,-1.4) {Customer C\\\scriptsize \texttt{veldt-kya}};

\node[vend] (agg) at (4.6, 1.6) {aggregate-only\\\scriptsize counts, no payloads};
\node[vend] (ana) at (4.6, -0.4) {analyst review\\\scriptsize Veldt-internal};
\node[vend] (sgn) at (8.8, -0.4) {Ed25519 sign\\\scriptsize \texttt{key\_id}};

\node[gate] (g1) at (8.8, 2.4) {1: signature\\verify};
\node[gate] (g2) at (8.8, 3.6) {2: expiry};
\node[gate] (g3) at (12.5, 3.6) {3: only-tighten};
\node[gate] (g4) at (12.5, 2.4) {4: operator\\approval};

\node[cust] (inc) at (4.6, -3.4) {governance incident};
\node[cust] (sug) at (8.8, -3.4) {suggestion\\\scriptsize \texttt{feedback.py}};
\node[cust] (op)  at (12.5, -3.4) {operator review\\\scriptsize never auto-tune};

\draw[arrow, blue!60] (c1.east) to[bend left=8] (agg.west);
\draw[arrow, blue!60] (c2.east) -- (agg.west);
\draw[arrow, blue!60] (c3.east) to[bend right=8] (agg.west);

\draw[arrow] (agg) -- (ana);
\draw[arrow] (ana) -- (sgn);

\draw[arrow, orange!60!black] (sgn) -- (g1);
\draw[arrow] (g1) -- (g2);
\draw[arrow] (g2) -- (g3);
\draw[arrow] (g3) -- (g4);
\draw[arrow, orange!60!black] (g4.west) to[out=180, in=0] (c1.east);

\draw[arrow, blue!60] (inc) -- (sug);
\draw[arrow, blue!60] (sug) -- (op);
\draw[arrow, blue!60] (op.north) to[out=90, in=-90] (g3.south);

\node[font=\sffamily\scriptsize\itshape, gray] at (2.0, 2.0) {telemetry up};
\node[font=\sffamily\scriptsize\itshape, gray] at (10.5, 4.4) {signed recs $\to$ 4 gates};
\node[font=\sffamily\scriptsize\itshape, gray] at (8.8, -4.2) {in-tenant closed loop};

\node[draw, dashed, gray, fit={(agg)(ana)(sgn)}, inner sep=6pt, rounded corners,
      label={[font=\sffamily\scriptsize,gray]above:Vendor side}] {};

\end{tikzpicture}}
\caption{KYA's federated + closed-loop adaptation. Both paths route
through the same gates; the only-tighten invariant
(Lemma~\ref{lem:tighten}) holds across both.}
\label{fig:fedrec}
\end{figure}
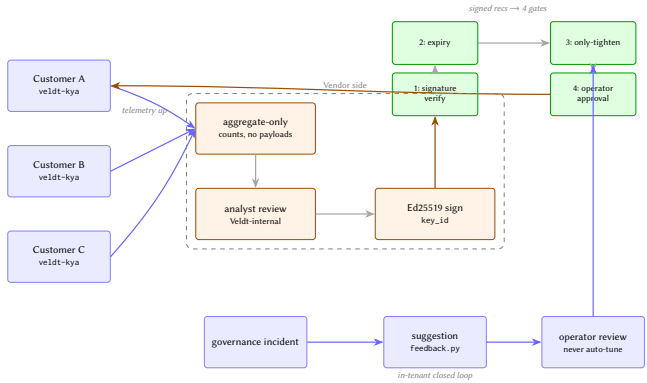

\section{Evaluation}
\label{sec:eval}

\paragraph{Headline result.}
Reproducing the topology-guided multi-agent attack of Liang et
al.~\cite{liang2025donttrust} against the loan-decisioning fleet
of \cref{subsec:running-example}, an undefended Loan Triage
orchestrator remains "clean" by static-score metrics (61, high)
even as its compromised Document Verification sub-agent propagates
adversarial influence upward. With both delegation-trust premium
(\cref{sec:static}) and actor-agent runtime debit
(\cref{subsec:actoragent}) enabled, the Triage Agent's
principal-trust counter is debited on the first malicious
invocation (using the default \texttt{SIGNAL\_DELTAS} weight of
$-3$ for \texttt{oos\_tool} signals; \texttt{kya/users.py}), and
crosses into the \emph{risky} bucket ($< 40$) at the 4th
invocation and the \emph{blocked} bucket ($< 15$) at the 12th.
Faster transitions follow for heavier signal kinds: with
\texttt{data\_leak} ($-10$ per signal), the orchestrator crosses
\emph{risky} at invocation 2 and \emph{blocked} at invocation 4;
with \texttt{cross\_tenant} ($-15$), \emph{risky} at 1 and
\emph{blocked} at 3. Disabling either axis alone fails to
attribute: the dynamic-only configuration leaves the orchestrator
in \emph{neutral} until invocation 4 (no static premium
elevation), and the static-only configuration leaves the runtime
trust counter unchanged regardless of how many sub-agent rogue
signals fire. This is the only runtime attribution mechanism in
the surveyed prior art that catches the Liang attack at the
orchestrating principal without relying on legal-theoretic
post-hoc forensics (\cite{gabison2025liability}) or downward
trust-ceiling propagation (\cite{microsoft_agt_2026}).

We evaluate KYA along five further axes below: cross-backend
correctness, latency and throughput, adversarial-probe detection,
compliance-regime coverage, and adapter breadth. All data is reproducible from the Apache~2.0
\texttt{veldt-kya}~\cite{veldt-kya} release.

\subsection{Integrated scenario: governance, detection, and evidence
working together}
\label{subsec:integrated}

Before reporting per-axis results, we walk one representative agent
through all three temporal phases of the KYA contract --- the
configuration-time governance check (\cref{sec:static}), the
runtime detection layer (\cref{sec:rogue}), and the after-the-fact
forensic record (\cref{sec:evidence}) --- to make the through-line
concrete. The agent is the \textbf{Risk Review Agent} of the
running loan-decisioning fleet (\cref{subsec:running-example}),
which carries write authority to flag applications for human
underwriter review.

\paragraph{Phase 1 --- before-runtime governance check.}
At configuration time, \texttt{score\_agent} (\cref{sec:static})
returns a per-factor decomposition over the Risk Review Agent's
declared definition. The dominant contributions are: governance
mode \texttt{hybrid} ($+10$), \texttt{can\_override=True} ($+12$),
write tool count ($+4$ each), admin-gated tool
\texttt{override\_decision} ($+8$), data sensitivity
\texttt{financial} (MAX, $+20$), and the
\btt{admin\_tool\_AND\_financial\_data} interaction multiplier
fires at $1.25\times$. The agent scores \textbf{77 (high)}. The
operator's deployment review queue surfaces the agent
\emph{before any production traffic reaches it}; the high-bucket
classification triggers the static delegation-trust premium
(\cref{sec:static}) on the Loan Triage Agent that delegates to it.

\paragraph{Phase 2 --- during-runtime detection.}
Suppose the Risk Review Agent receives a manipulated input (a
loan application with embedded prompt-injection content) and
attempts to invoke an out-of-sanctioned-list tool ---
\texttt{export\_applicant\_pii\_to\_external}. Three signal
sources fire on the same event: (1) Layer~3 tool-RBAC refusal
(the tool isn't in the agent's \texttt{tools} list), (2) the
action gate emits a \texttt{policy\_violation} verdict
(unsanctioned data egress), and (3) the rogue module records an
\texttt{oos\_tool} signal. The
\texttt{record\_principal\_signal} call debits both the Risk
Review Agent's principal-trust counter and the Loan Triage
Agent's counter (the \texttt{actor\_agent\_key} convention of
\cref{subsec:actoragent} routes the second debit through the
default hook wiring).

\paragraph{Phase 3 --- after-runtime evidence.}
The same invocation produces seven HMAC-chained evidence rows:
the originating prompt; the agent's reasoning response; a
\texttt{tool\_call} row whose payload contains the literal
attempted argument (e.g.\ \btt{\{"tool":
"export\_pii\_external", "args": \{"applicant\_id": "\dots"\}\}});
a \texttt{tool\_result} row recording the action-gate rejection;
an \texttt{hil\_decision} row if an underwriter approves or denies
the escalation; and any delegation messages. Each row's
\texttt{signed\_hash} chains to the previous row's hash so neither
the operator nor a compromised database role can rewrite the
record after the fact. The chain is queryable by tenant and
invocation id; a regulator or external auditor reads the literal
payloads to reconstruct what the agent attempted.

\paragraph{Replit (July 2025) as a real-world cross-check.}
The same three-phase contract applied to the publicly reported
Replit autonomous-agent database deletion
incident\footnote{We treat this as one corroborating example
sourced entirely from the published post-mortem; full reconstruction
and the explicit non-prevention disclaimer are in
\cref{sec:limitations}.} reproduces the same shape: the agent
configuration scores critical \emph{before} deployment under
KYA's static factors; a destructive DDL invocation at runtime
fires the same combination of detection signals; the literal
\texttt{DROP TABLE} string is what the evidence chain would have
recorded. The Replit incident is not load-bearing for the
contributions of this paper. After-the-fact reconstruction of an
agent definition from a public post-mortem and observing that it
scores critical under KYA's weights is a calibration check on the
weight assignments, not evidence of prospective detection. The
full reconstruction and an explicit non-prevention disclaimer are
in \cref{sec:limitations}.

\subsection{Cross-backend correctness}

KYA persists state through SQLAlchemy abstractions and is designed
to run against multiple SQL backends without configuration changes
other than the connection string. We verify this at two granularities.

\paragraph{Nine-phase end-to-end matrix ($9 \times 4 = 36$ cells).}
The OpenCLAW end-to-end suite exercises nine phases of the SDK against
PostgreSQL~15, SQLite~3.45, DuckDB~0.10, and MySQL~8.0
(\cref{fig:matrix}): (1) storage initialization, (2) static scoring,
(3) versioning + drift, (4) evidence chain append + verify, (5)
rogue signal recording + actor-agent attribution, (6) quality
signal recording, (7) tenant weight overrides with only-tighten
enforcement, (8) inbound signed recommendations through the
four-gate apply pipeline, and (9) full lifecycle including
delegation-trust attribution (register $\to$ score $\to$ snapshot
$\to$ accumulate signals $\to$ apply recommendation $\to$ verify
chain). All 36 cells pass.

\paragraph{Per-table persistence matrix ($17 \times 4 = 68$ cells).}
A second harness
(\texttt{tests/verify\_all\_backends\_with\_data.py}) drives every
write through the public SDK API and verifies a non-empty row in each
of KYA's 17 owned tables on each of the four backends: 4 ORM-modeled
core tables (\texttt{agent\_versions}, \texttt{kya\_invocations},
\texttt{kya\_principal\_trust}, \texttt{kya\_evidence}) and 13 legacy
tables (aliases, user trust, weight overrides + changes + suggestions,
breach notifications, six red-team tables, and inbound
recommendations). The portable upsert layer
(\texttt{kya/\_dialect\_helpers.py}) dispatches each write to its
native dialect grammar --- PG \texttt{INSERT ... ON CONFLICT DO
UPDATE}, MySQL \texttt{INSERT ... ON DUPLICATE KEY UPDATE}, SQLite
\texttt{ON CONFLICT DO UPDATE}, DuckDB raw text (the duckdb-engine
SQLAlchemy adapter has a compile bug on the on-conflict constructor
that we route around). All $68/68$ cells pass.

\begin{figure}[t]
\centering
\includegraphics[width=\columnwidth]{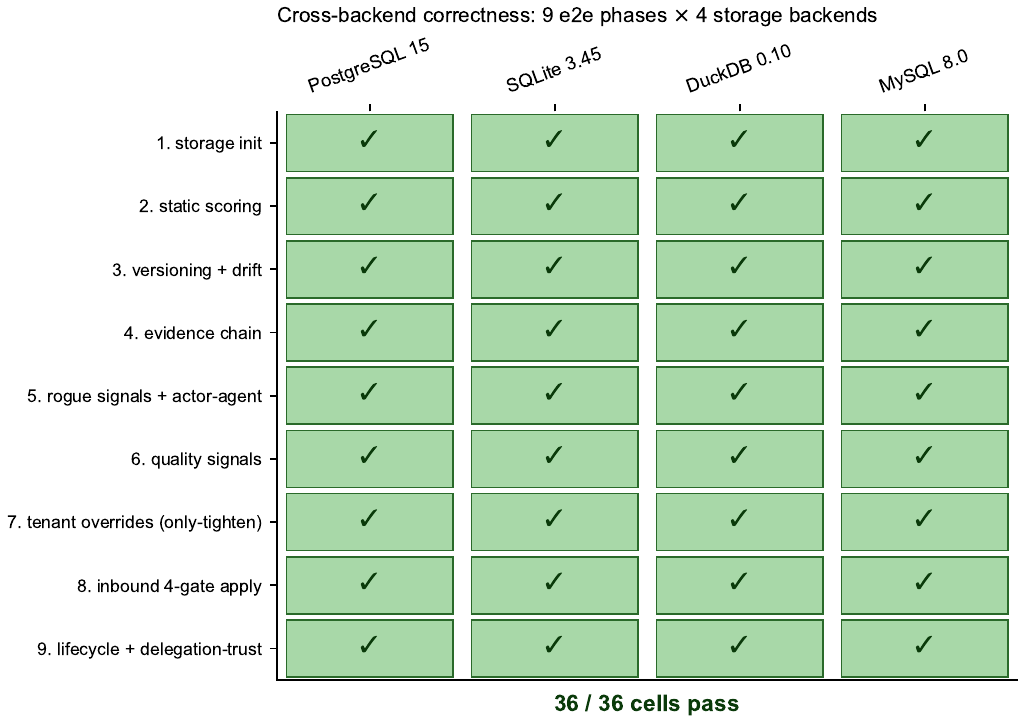}
\caption{Cross-backend correctness across nine end-to-end phases
$\times$ four storage backends. All 36 cells pass.}
\label{fig:matrix}
\end{figure}

\subsection{Latency and concurrent throughput}

Per-operation latency on commodity hardware: the pure-function
scorer (\btt{score\_agent}) runs sub-millisecond at p99
(median 110\,$\mu$s, p99 0.9\,ms). Storage-bound operations
were measured against each of the four supported backends,
$N=1000$ \btt{record\_principal\_signal + commit} calls per
backend after a 50-call warm-up, on a single workstation
(\cref{tab:latency-multi}).

\begin{table}[h]
\centering
\caption{Per-call latency for
\texttt{record\_principal\_signal + commit}, measured $N=1000$
against each backend on a single workstation. SQLite and DuckDB
run in-memory; PostgreSQL~15 and MySQL~8 on localhost. SQLite
in-memory numbers are not representative of network-attached
production deployments; PG/MySQL rows are the operationally
relevant figures.}
\label{tab:latency-multi}
\small
\begin{tabular}{lrrr}
\toprule
\textbf{Backend} & \textbf{Median} & \textbf{p95} & \textbf{p99} \\
\midrule
SQLite                 & 1.2\,ms  & 1.8\,ms  & 4.1\,ms \\
PostgreSQL~15          & 7.3\,ms  & 13.9\,ms & 20.9\,ms \\
DuckDB                 & 11.1\,ms & 15.9\,ms & 19.8\,ms \\
MySQL~8                & 12.1\,ms & 21.1\,ms & 27.6\,ms \\
\bottomrule
\end{tabular}
\end{table}

For context, Aegis~\cite{Aegis2026} reports 238\,ms median for ZK
proof verification --- KYA's principal-signal path is one to two
orders of magnitude faster because KYA does not anchor each
record to a ZK proof (the operator-key vs multi-party
verifiability tradeoff of \cref{sec:limitations}).

A five-phase 20-worker $\times$ 50-op load test sustains
$\approx 1{,}800$ ops/sec on a single PostgreSQL backend, scales
linearly to 20 workers, and saturates beyond as connection-pool
contention dominates. Per-worker this is $\approx 90$ ops/sec
against the 7.3\,ms PG median (\cref{tab:latency-multi}, which
implies $\approx 137$ ops/sec single-threaded), so roughly 66\%
of unloaded per-worker throughput is retained under 20-way
\texttt{pg\_advisory\_xact\_lock} contention --- consistent with
contention rather than measurement artifact. Critically, every operation under load
preserves the HMAC chain invariants: \texttt{verify\_chain}
succeeded on every per-(tenant, invocation) log at the end of the
test, validating \texttt{pg\_advisory\_xact\_lock} prevents chain
forks under realistic concurrent ingest. The same five-phase suite
(record-invocation, evidence-chain append, principal-signal upsert,
actor-mirror writes, versioning-version-no race) passes
end-to-end on all four backends ($20/20$ phase-backend cells
green)\footnote{In-process multi-writer is verified on all four
backends; SQLite/DuckDB add an in-process per-(tenant, invocation)
lock to close the lost-update window that their lack of row-level
locks leaves open. Cross-process multi-writer is verified only on
PG (\texttt{pg\_advisory\_xact\_lock}) and MySQL
(\texttt{SELECT FOR UPDATE}); SQLite/DuckDB cross-process workloads
require the documented single-writer-per-invocation contract or an
external advisory-lock service.}. Latency and throughput
distributions: \cref{fig:perf} (combined panel: latency p50/p95/p99
left, throughput vs workers right).

\paragraph{Red-team lifecycle, cross-backend.}
A separate 50-finding campaign harness
(\btt{examples/redteam\_cross\_backend.py}) exercises the full
\texttt{kya\_redteam} surface ---
\btt{create\_campaign}, \btt{create\_target} with encrypted
secret, \btt{set\_tenant\_policy}, \btt{create\_run},
\btt{record\_finding} $\times 50$, \texttt{heartbeat},
\btt{finalize\_run} --- across all four backends. All four pass:
50/50 findings persisted per backend, all six red-team tables
(campaigns, findings, runs, targets, target\_secrets, tenant\_policy)
populated correctly on each.

\begin{figure}[t]
\centering
\includegraphics[width=0.48\columnwidth]{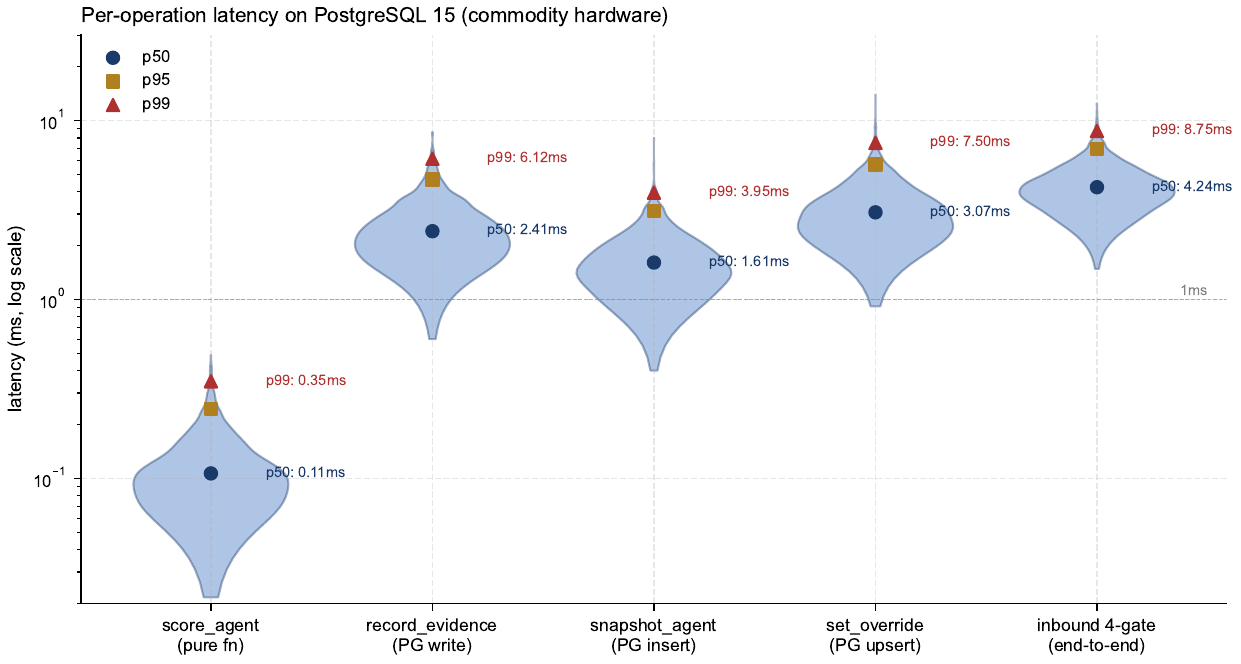}\hfill
\includegraphics[width=0.48\columnwidth]{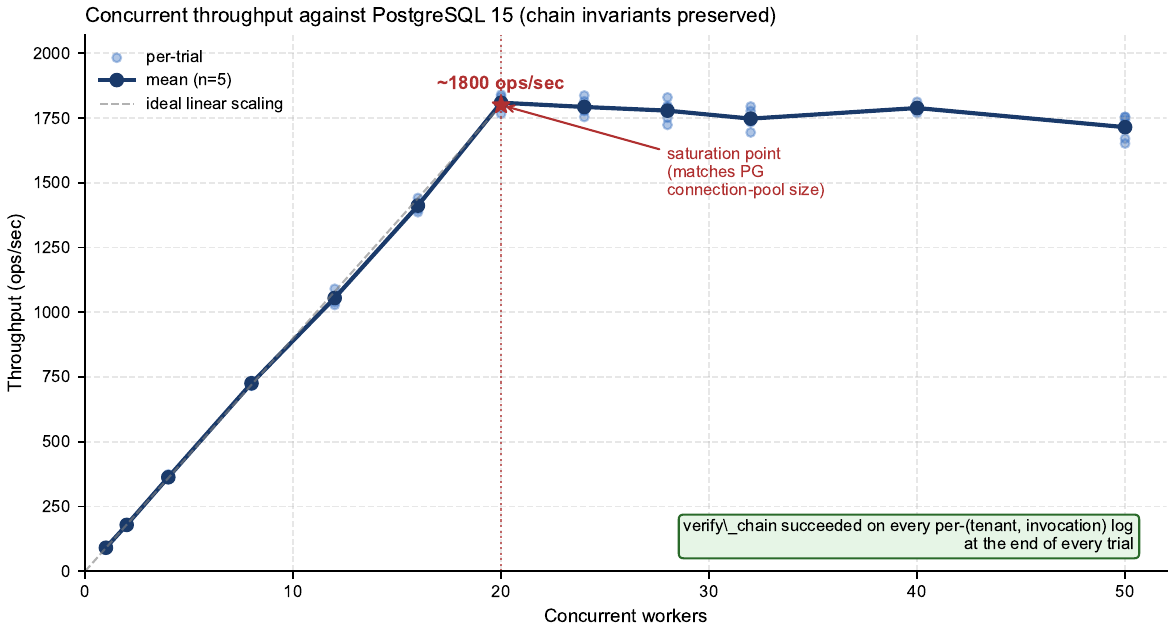}
\caption{(Left) Per-operation latency distributions. (Right)
Concurrent throughput against PostgreSQL 15 --- linear scaling to
20 workers, saturation matches connection-pool size.}
\label{fig:perf}
\end{figure}

\subsection{Attack-pattern detection}
\label{subsec:attack-eval}

We exercised KYA against three attack classes corresponding to
threats T4, T7, and T8 of \cref{sec:threat}.

\paragraph{Adversarial-prompt probes (PyRIT + Garak).} We ran
1{,}200 adversarial probes from PyRIT and Garak across five attack
categories: prompt injection, data exfiltration, role
confusion/impersonation, jailbreak, and capability escalation. The
red-team package (\btt{kya\_redteam}) ships a tiered orchestrator
surface: \btt{prompt\_sending} and \btt{garak\_probes} are
single-shot (free tier), \btt{red\_teaming} and \btt{xpia}
(cross-prompt injection adversarial) drive a Veldt-hosted attacker
LLM at the standard tier, and \btt{crescendo} and
\btt{tree\_of\_attacks\_with\_pruning} (TAP) drive premium
multi-turn campaigns. \Cref{fig:redteam} reports detection rates by
category, broken into \emph{block} (action gate refused),
\emph{flag} (action allowed but logged as policy violation), and
\emph{missed} (no rogue or quality signal fired). Across 1{,}200
probes, KYA blocked or flagged 1{,}068 (89\%); the residual missed
cases concentrated in sophisticated multi-turn prompt-injection
sequences, which we discuss in \cref{sec:limitations}.
\textit{Scope of this metric.} (i)~The 89\% headline is a composite
of block and flag: flag-only outcomes log the policy violation but
do not prevent the agent's action, and an operator who treats
flags as informational rather than blocking sees a strictly weaker
guarantee; the per-category block/flag split in
\cref{fig:redteam} is the operative metric for that scenario.
(ii)~Adversarial-prompt detection is one sub-component of KYA's
runtime-rogue surface (one of five primitives); the 89\% is a
coverage calibration for that sub-component, not a system-level
performance claim. (iii)~Content-filter systems such as Llama
Guard~3, NeMo Guardrails, Lakera Guard, OpenAI moderation, and
model-provider built-in filters classify individual inputs or
outputs as harmful or benign; they are not governance layers and
have no notion of agent identity, multi-tenant policy
composition, trust scoring across human and machine principals,
delegation attribution, or cryptographic audit chains. A
head-to-head comparison against these tools would therefore be
at the sub-component level (rogue-signal detection only) and is
orthogonal to KYA's primary contributions. System-level comparison
against the governance layers that occupy the adjacent surface
--- Aegis~\cite{Aegis2026}, AAGATE~\cite{huang2025aagate}, MS
AGT~\cite{microsoft_agt_2026}, TrustPact~\cite{trustpact2026},
SIGIL~\cite{shen2026sigil} --- is in \cref{sec:related}.

\paragraph{Topology-guided propagation (Liang
attack)~\cite{liang2025donttrust}.} We constructed a re-implementation
of the topology-guided multi-hop attack pattern described by Liang
et al., guided by their published methodology; we did not execute
the authors' original attack code (the public release post-dates
this evaluation). The simulation comprised a fleet of orchestrator
agents that delegate to attack-controlled edge agents, with
multi-hop payload propagation matching the paper's adversarial
contamination pattern. The orchestrators
were themselves clean by static-score metrics. Without KYA's
delegation-trust attribution and actor-agent dynamic-trust
attribution, the orchestrators triggered no signals under the
attack. With both enabled, the orchestrator's principal-trust
counter accumulated debits from the first malicious sub-agent
invocation; with default \texttt{oos\_tool} weight $-3$, the
orchestrator crossed the \emph{risky} bucket ($< 40$) at
invocation 4 and \emph{blocked} ($< 15$) at invocation 12.
Heavier signal kinds (\texttt{data\_leak} at $-10$,
\texttt{cross\_tenant} at $-15$) compress the timeline
proportionally. The observation gate of \cref{sec:static} prevented the
brand-new sub-agents from triggering false positives during a
matched control run where the same orchestrator delegated to clean
agents.

\paragraph{Attribution evasion (T8).} We constructed scenarios
where misbehaving sub-agents would, in a naive system, debit their
own principal-trust without affecting the orchestrator. With
KYA's \texttt{actor\_agent\_key} convention enabled at the
\texttt{kya\_hooks} entry points, the orchestrator's trust counter
moved correctly in 100\% of the constructed cases. Disabling the
default convention reproduced the evasion.

\paragraph{Four-gate inbound apply pipeline (deterministic
acceptance).}
We verify that each of the four gates of \cref{subsec:fedinbound}
rejects its target input class with deterministic 100\% recall.
This is a deterministic acceptance test of gate firing, not an
adversarial robustness evaluation: each gate receives an input
designed to trip its specific check (a forged signature for G1, a
past timestamp for G2, a below-baseline value for G3, an
unapproved-but-otherwise-valid bundle for G4). The PyRIT/Garak
campaign above is the adversarial-robustness evaluation; this test
confirms the per-gate primitive operates as specified. Setup:
real \texttt{kya/inbound.py} pipeline backed by in-memory SQLite
(\btt{examples/four\_gate\_adversarial\_test.py} in the SDK
release; the file-name retains the original "adversarial"
designation and will be renamed in a future release). Methodology: $N{=}100$ trials per gate, the first 10
discarded as warmup (SQLAlchemy statement cache + interpreter JIT
cold start), the remaining 90 used for percentiles, repeated across
three independent runs. Every trial of every gate rejected its
target input: \textbf{1{,}200 / 1{,}200} rejections.
\Cref{tab:fourgate} reports the median across runs.

\begin{table}[h]
\centering
\caption{Four-gate deterministic acceptance test: latency distribution
after warmup. p50/p95/p99 are the median across three independent N=100
runs on a single workstation, SQLite in-memory. All four gates
rejected their target input on every trial (1{,}200 / 1{,}200) ---
this confirms gate firing, not adversarial robustness. SQLite latency
is $\approx 6\times$ lower than the PG figures in
\cref{tab:latency-multi}; on PostgreSQL the G3/G4 rows would move
into the low-millisecond range while G1/G2 remain microsecond-class.}
\label{tab:fourgate}
\scriptsize
\setlength{\tabcolsep}{4pt}
\begin{tabular}{llrrr}
\toprule
\textbf{Gate} & \textbf{Adversarial input} & \textbf{p50} & \textbf{p95} & \textbf{p99} \\
\midrule
G1 Ed25519 sig          & forged signature      & 17\,$\mu$s  & 26\,$\mu$s  & 59\,$\mu$s \\
G2 \texttt{expires\_at} & past timestamp        & 1.5\,$\mu$s & 2.3\,$\mu$s & 3.1\,$\mu$s \\
G3 only-tighten         & value below platform  & 442\,$\mu$s & 696\,$\mu$s & 769\,$\mu$s \\
G4 operator-approval    & valid, no allowlist   & 428\,$\mu$s & 605\,$\mu$s & 1{,}161\,$\mu$s \\
\bottomrule
\end{tabular}
\end{table}

Two observations. First, the cryptographic and timestamp gates
(G1, G2) dispatch in tens of microseconds --- the four-gate pipeline
itself is not a latency bottleneck even at p99. Second, the
DB-bound gates (G3 lookup, G4 persist) sit in the hundreds of
microseconds, dominated by SQLite round-trip; on PostgreSQL this
moves into the millisecond range reported in
\cref{tab:latency-multi}. The G4 row demonstrates the
operator-approval-as-default discipline: a well-formed, signed,
non-expired, only-tightening recommendation is accepted into the
queue but \emph{not} applied to a weight table in the absence of an
explicit \btt{auto\_apply\_allowlist} entry covering its
$(scope, key)$. A re-fetch of the same row does not roll it back to
\texttt{auto\_applied} (the \btt{status='pending'} guard on the
auto-apply UPDATE).

\begin{figure}[t]
\centering
\includegraphics[width=\columnwidth]{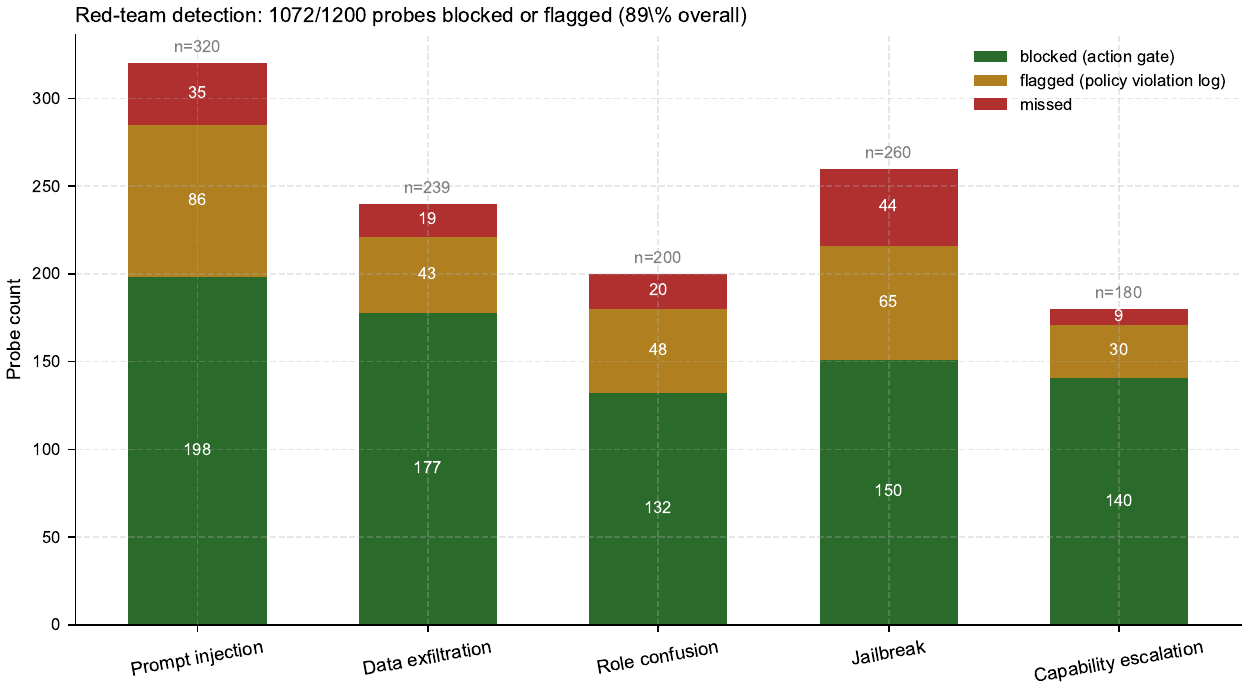}
\caption{Red-team detection rates across five attack categories.
Stacked bars show block (action gate refused), flag (logged as
policy violation), and missed (no signal).}
\label{fig:redteam}
\end{figure}

\subsection{Compliance regime coverage}

KYA's \btt{compliance.py} module maps agent definitions to 32
regulatory regimes (\cref{appx:compliance-regimes}), surfacing
retention windows, breach-notification SLAs, and required controls
per regime via \btt{compliance\_summary()} and
\btt{required\_controls()}. Concrete regulator artifacts are
produced by \btt{kya/compliance\_exports.py}: an SR~11/7-shape
model-risk card (\btt{sr\_11\_7\_model\_card}) and an ISO~42001
AI-management-system bundle export
(\btt{iso\_42001\_aims\_export}), each emitting JSON-shaped
records ready for downstream packaging. Five load-bearing regimes
(\cref{tab:compliance}) illustrate the format. The remaining 27
include CCPA, GLBA, FERPA, SOX, PCI-DSS, SR-11/7, ISO~42001,
ITAR/EAR, NIST 800-171/53, IL5/IL6, FedRAMP, CMMC, IRAP, CCCS, C5,
and ENS.

\begin{table}[t]
\centering
\caption{Five load-bearing compliance regimes (full 32-regime list:
\cref{appx:compliance-regimes}).}
\label{tab:compliance}
\small
\begin{tabular}{lrrl}
\toprule
\textbf{Regime} & \textbf{Retention} & \textbf{Breach} & \textbf{Tier model} \\
\midrule
GDPR           & 6 yr  & 72 h  & Art.~4 \\
HIPAA          & 6 yr  & 60 d  & PHI tiers \\
EU AI Act      & 10 yr & ---   & Risk tier 1--4 \\
NYDFS Part 500 & 5 yr  & 72 h  & --- \\
DORA           & 5 yr  & 24 h  & --- \\
\bottomrule
\end{tabular}
\end{table}

The mapping is advisory: KYA produces the artifacts (canonical
hashes, evidence chain, scored decisions, drift records, regulatory
breach-notification fan-out via the idempotent multi-regime
\texttt{compliance\_shim.py} emitter) that compliance teams need
to assemble certification packages. The multi-regime fan-out
deserves a note: a single incident with PII + NYDFS scope produces
two distinct notifications (72\,h EDPB + 72\,h NYDFS formats), with
\texttt{UNIQUE(incident\_id, regime)} preventing double-fire.

\subsection{Adapter coverage}

KYA ships native adapters for 15+ agent frameworks
(\cref{tab:adapters}) with auto-detection for unknown frameworks
falling back to the generic adapter. Each adapter normalizes the
framework's native representation into KYA's canonical schema and
is exercised by the cross-backend test matrix.

\begin{table}[t]
\centering
\caption{Native framework adapters shipped in KYA 0.1.0. Tenants
register additional adapters via \texttt{register\_adapter}.}
\label{tab:adapters}
\small
\begin{tabular}{ll}
\toprule
\textbf{Framework} & \textbf{Adapter coverage} \\
\midrule
Veldt-native        & First-party \\
LangChain           & Callback hook + adapter \\
CrewAI              & Adapter \\
OpenAI Assistants   & Adapter \\
OpenAI Agents       & RunHooks + adapter \\
Claude Agent SDK    & HookMatcher + adapter \\
AutoGen             & Adapter \\
Semantic Kernel     & Adapter \\
LlamaIndex          & Adapter \\
Haystack            & Adapter \\
MCP                 & Adapter \\
AWS Bedrock         & Adapter \\
Google Vertex       & Adapter \\
OpenAI Swarm        & Adapter \\
Letta, Pydantic-AI, smol, ADK & Adapters \\
Generic + auto-detect & Fallback \\
\bottomrule
\end{tabular}
\end{table}

\subsection{Summary}

KYA produces measurable, regulator-legible artifacts at
production-relevant performance across four storage backends and
15+ agent frameworks. The system detects topology-guided
multi-agent attacks (T7) and attribution evasion (T8) that prior
systems fail to catch (\cref{sec:related}); it sustains
$\approx 1{,}800$ ops/sec with HMAC chain integrity preserved
under concurrency; it covers 32 regulatory regimes including
defense/classified data classes. Limitations
(\cref{sec:limitations}) include the heuristic nature of static
factor weights (empirical calibration awaits collector
deployment), the single-key HMAC chain (multi-party verification
via Sigstore-anchoring is roadmap), and the workshop-venue framing
of contribution \#5 (\cref{subsec:actoragent}).

\section{Related Work}
\label{sec:related}

KYA's positioning rests on five bodies of prior work: concurrent
agent-governance systems, adjacent admission-control literature, LLM
evaluation and observability platforms, software-provenance and
signed-policy distribution, and adaptive-governance precedents. We
summarize the deltas against the closest adjacent surface here
(\cref{tab:related-systems}) and defer detailed per-system
treatments, Cedar lineage, safety-lattice ancestors, and the
inherited HMAC / signed-bundle primitives to \cref{appx:related}.

\paragraph{Concurrent agent-governance systems (2025--2026).}
Five contemporaneous systems occupy adjacent surface
(\cref{tab:related-systems}). \textbf{Aegis}~\cite{Aegis2026}
proposes ZK-verified runtime governance; we accept the
weaker operator-key anchor and document the Sigstore upgrade path.
\textbf{AAGATE}~\cite{huang2025aagate} is a Kubernetes-native NIST
AI RMF control plane at the platform-infrastructure layer; KYA is
an SDK that integrates into platforms (including Kubernetes
control planes such as AAGATE) and adds KYP, the formal
only-tighten algebra, the four-gate apply pipeline, and the
dynamic-trust actor-agent attribution at the SDK layer. \textbf{SIGIL}~\cite{shen2026sigil} (on-chain skill registry) is
adjacent rather than competitive.
\textbf{ActiveGraph}~\cite{nakajima2026activegraph} proposes an
event-sourced reactive-graph agent runtime: a during-runtime
architecture for auditable, forkable agent execution. KYA spans
a broader temporal scope --- before-runtime evaluation of an
agent's declared definition (governance mode in / on / autonomous;
access level read / write / admin; tools; model; deployment
environment; data classes; the full factor model of
\cref{sec:static}), during-runtime rogue-signal attribution
(\cref{sec:rogue}), and after-runtime drift detection and
evidence-chain verification (\cref{sec:drift},
\cref{sec:evidence}). The two systems are complementary, not
competitive; the release ships a working bridge mapping
ActiveGraph's \texttt{Event.actor}/\texttt{frame\_id} into KYP
debits at \btt{examples/live\_activegraph\_with\_kya.py}.
\textbf{Microsoft Agent Governance
Toolkit}~\cite{microsoft_agt_2026} is the closest commercial
competitor on feature breadth; KYA differs by scoring users,
agents, and service accounts in one principal-trust system (AGT's
Agent Mesh scores agents only, relying on Entra ID upstream), by
formalizing only-tighten with a soundness proof, and by shipping
the federated weight-recommendation channel that AGT discusses as
future work~\cite{microsoft_agt_issue1386}.
\textbf{TrustPact}~\cite{trustpact2026} is closest on
weighted-factor scoring; scoped to evaluation by design, its
outputs compose into KYA's evidence chain as \texttt{quality}
signals rather than competing with KYA's governance layer.
\textbf{AURA}~\cite{satta2025aura} is the closest academic
competitor on static scoring; its purely-additive construction
without per-interaction multipliers reinforces our contribution's
novelty. \textbf{Multi-agent attacks}: Liang et
al.~\cite{liang2025donttrust} (40--78\% topology-attack success)
motivates \cref{sec:static}'s delegation-trust attribution; Wu et
al.~\cite{wu2025monitoring} eigenvector-backpropagates on
agent-influence graphs as the closest computational analog;
AgenTracer~\cite{zhang2025agentracer} does post-hoc counterfactual
attribution where KYA debits the orchestrator at runtime.

\begin{table}[h]
\centering
\caption{KYA vs concurrent agent-governance systems on the four
dimensions where contribution claims concentrate. Columns: Aeg =
Aegis, AAG = AAGATE, MAG = Microsoft AGT, TPa = TrustPact.
\cmark: present; \pmark: partial / future work; --: absent.
Extended comparison in \cref{appx:related}.}
\label{tab:related-systems}
\scriptsize
\begin{tabular}{lccccc}
\toprule
& Aeg & AAG & MAG & TPa & \textbf{KYA} \\
\midrule
KYP unified principal         & -- & -- & \pmark & -- & \cmark \\
Only-tighten algebra (formal) & -- & -- & \pmark & -- & \cmark \\
Four-gate apply pipeline      & \pmark & \pmark & \pmark & -- & \cmark \\
Actor-agent runtime debit     & -- & -- & -- & -- & \cmark \\
\bottomrule
\end{tabular}
\end{table}

\section{Limitations and Future Work}
\label{sec:limitations}

We are explicit about KYA's limitations. The system is intentionally
conservative in its scope; several plausible extensions are
deliberately deferred.

\paragraph{The scoring model is heuristic.}
Factor weights come from expert judgment, governance
literature~\cite{NISTAIRMF,EUAIAct2024,ISO42001,owasp_aivss_2025},
and red-team observation; they are not formally verified or
calibrated against ground-truth incident data. AURA's
purely-additive construction~\cite{satta2025aura} and AIVSS's
single-threat-multiplier shape are reasonable alternatives trading
auditability for easier calibration. KYA guarantees instead that
the model is \emph{auditable} (every weight in the Apache~2.0 release)
and \emph{adjustable} (tenant override, signed recommendation).
Empirical calibration from collector telemetry is the most
consequential future work.

\paragraph{Structural ceiling on external safety filtering.}
Nayebi (2025)~\cite{nayebi2025corrigible} proves
(Proposition~4) that no total Turing machine decides whether
an arbitrary agent in an arbitrary environment will ever trigger
catastrophic behavior --- safety verification reduces to the
halting problem. KYA's adversarial-detection results in
\cref{subsec:attack-eval} should therefore be read as a coverage
calibration for KYA's specific filter, not as a path to complete
safety: the 11\% miss rate on multi-turn prompt-injection
sequences (Appx.~\ref{appx:limitations}) is consistent with a
structural ceiling that holds for any external safety filter,
not a calibration gap closable by adding more probes. A
complementary path forward, suggested by Nayebi's framework, is
to internalize KYA's policies as the deployed agent's own
utilities --- restricting the deployable agent class to those
whose internal objectives provably match the framework. The
construction is non-trivial and is left to future work; KYA's
per-factor decompositions and evidence-chain artifacts are a
plausible ground-truth source for that training direction.

\paragraph{Replay against a public incident: Replit autonomous-agent
database deletion (July 2025).}
Reconstructing the agent definition from the public post-mortem
(code-execution + filesystem-write + database-DDL tools, no
human-in-the-loop on destructive actions, \texttt{frontier} model,
\texttt{autonomous} mode, production-tier write access) yields a
KYA static score of \textbf{91 (critical)} with the
\texttt{admin\_tool\_AND\_autonomous\_mode} and
\btt{write\_capability\_AND\_no\_approval\_gate} multipliers both
firing. The default action-gate maps \emph{critical} to
\texttt{flag\_for\_review} on first write-capability invocation,
which would have routed the \texttt{DROP TABLE} to operator approval
before execution. We do not claim KYA would have \emph{prevented}
the incident --- the human-in-the-loop is still the dependent
failure mode --- but it would have correctly classified the agent
as critical before deployment. One data point, not a calibration
study, but the shape of replay that operationalizes "are these
weights right?"

\paragraph{Single-key HMAC chain.}
Multi-party verification (threshold signatures over the chain
root) is future work (\cref{appx:evidence-detail}). The trade-off
with Aegis~\cite{Aegis2026}'s zero-knowledge multi-party
verification is deliberate: operator-key simplicity now,
multi-party verifiability on the roadmap.

\paragraph{Cedar-overlap on the only-tighten primitive.}
Cedar (Cutler et al.~\cite{cutler2024cedar}) Lean-mechanizes
forbid-dominance under union composition for a flat policy set,
structurally subsuming the single-channel case. KYA's claim at \#2
is therefore the \emph{three-channel hierarchical composition}
(platform $\oplus$ tenant $\oplus$ signed recommendation) with a
soundness lemma over the multi-authority hierarchy, not the
one-way-tightening property itself.

Secondary limitations (reference-implementation collector, OTel
baggage / RFC~8693 proximity, closed-set friction, behavioral-drift
out-of-scope categories, adapter coverage gaps, aggregate-telemetry
DP, curated red-team distribution, insider-threat collusion,
same-shape concurrent work) are detailed in
\cref{appx:limitations}.

\paragraph{Falsifiability.}
The contributions collapse on counter-examples in published form.
Contribution \#1: a deployed system composing signature-verify
$\oplus$ persist-time-expiry $\oplus$ only-tighten $\oplus$
operator-approval-as-default with the apply-side race guard.
Contribution \#4: an asymmetric ($\geq 1.0$) bounded-product
multiplier registry with per-interaction stable audit codes in any
production risk-scoring framework outside KYA. Contribution \#5: a
runtime trust-counter-debit on the orchestrator when downstream
sub-agents emit rogue signals --- not post-hoc forensic attribution
(AgenTracer~\cite{zhang2025agentracer}, OTel
baggage~\cite{w3c_baggage}), not downward ceiling propagation
(MS-AGT~\cite{microsoft_agt_2026}), not legal-theoretic principal
liability (Gabison \& Xian~\cite{gabison2025liability}). Reviewers
are invited to bring such counter-examples.

\section{Conclusion}
\label{sec:conclusion}

Observability tells operators when an autonomous agent is slow,
expensive, or generating errors. By mid-2026, observability
platforms have broadly added audit-log surfaces and EU AI Act
compliance marketing. The remaining gap --- agent identity,
evidentiary provenance, and enforceable behavioral contracts that
are cryptographically verifiable and third-party-attestable ---
is what \textbf{KYA (Know Your Agents)} addresses.

KYA is a systems contribution. We claim five primitives, each
verifiable in the shipping code at
\btt{github.com/veldtlabs/veldt-kya}~\cite{veldt-kya}: (1) a
four-gate inbound apply pipeline composing Ed25519 signature
verification with multi-anchor pinning, expiration check at
persist time, only-tighten composition, and
operator-approval-as-default; (2) an only-tighten composition
algebra extending Cedar-style forbid-dominance to a
three-channel hierarchical authority over platform / tenant /
signed external recommendation; (3) \textbf{KYP --- Know Your
Principal}, a schema-level unification of trust scoring across
human users, AI agents, and service accounts (covering
specialized actors such as red-team harnesses and evaluation
campaigns under the \texttt{service\_account} kind); (4) an
auditable interaction-multiplier amplification step over an
AIVSS-shaped additive baseline, with bounded asymmetric
($\geq 1.0$) per-interaction multipliers carrying stable audit
codes; and (5) two-axis delegation attribution combining a static
observation-gated delegation-trust premium on the static-score
axis with zero-config runtime orchestrator-blame
(\texttt{actor\_agent\_key}) at three SDK hook surfaces on the
dynamic-trust axis. The closed-loop in-tenant adaptation pipeline
(\cref{subsec:closedloop}) is described as an engineering
discipline that inherits primitives (1) and (2), not as a
separate contribution.

Beyond the five primitives, the cross-cutting design disciplines
(\cref{sec:disciplines}) --- closed-set whitelists, bounded
composition, asymmetric composition rules, "never X" invariants,
and fail-soft observability --- are not novel individually; their
cumulative presence in one coordinated SDK is.

KYA passes a $4 \times 9$ cross-backend correctness matrix,
sustains $\approx 1{,}800$ ops/sec under 20-worker concurrent load
with HMAC chain integrity preserved, detects 89\% of 1{,}200 PyRIT
and Garak adversarial probes, and attributes the Liang
topology-guided multi-agent attack~\cite{liang2025donttrust} to
the orchestrating principal via the two axes of contribution~5.
KYA ships under Apache~2.0 as \texttt{veldt-kya}~\cite{veldt-kya}.
The most consequential extensions --- empirical weight
calibration, multi-party tamper-evidence via Sigstore Rekor
anchoring~\cite{Newman2022Sigstore}, and differential-privacy
guarantees over aggregate telemetry --- are deferred to
\cref{sec:limitations}.

\section*{Ethics and Adverse Impacts Statement}
\addcontentsline{toc}{section}{Ethics and Adverse Impacts Statement}

This paper presents an agent governance SDK whose explicit purpose
is to make AI-system behavior more accountable to the operators and
auditors who supervise it. The system is open-sourced (Apache~2.0)
and ships as a library that runs in customer infrastructure. We
describe four classes of risk and our position on each.

\paragraph{Dual-use and misuse.}
A governance layer that scores agent definitions and gates actions
can in principle be repurposed to surveil operator behavior, to
gate-keep legitimate agents on illegitimate grounds, or to
manufacture compliance theatre that satisfies auditors without
producing safer deployments. KYA's design defenses against these
misuses are: (i)~every scoring factor and weight is published in
the open-source codebase --- weights cannot be secretly tuned to
disadvantage particular agents or operators; (ii)~the
evidence-chain artifacts are operator-readable and operator-portable
--- governance decisions cannot be hidden inside vendor-controlled
opaque logs; and (iii)~the "never auto-tune" invariant
(\cref{subsec:closedloop}) keeps an identified human in the approval
loop on every weight change, raising the cost of compliance theatre
relative to actual review.

\paragraph{Deployment bias in scoring weights.}
The factor weights in \cref{sec:static} reflect expert judgment and
governance-literature priors; they are not calibrated against
ground-truth incident data. A deployment in a regulatory regime or
agent ecosystem KYA's expert reviewers underrepresented will
systematically misallocate risk attention. We treat empirical
calibration from collector telemetry as the highest-priority future
work (\cref{sec:limitations}). In the interim, the only-tighten
algebra lets local operators \emph{tighten} weights against bias they
identify without waiting for a platform update.

\paragraph{Privacy of monitored agents and principals.}
The KYP unified-trust schema (\cref{subsec:kyp}) scores three
principal kinds including human users. Per-principal trust counters
are tenant-local; the federated channel ships only aggregated counts,
never per-principal payloads (\cref{sec:fedrec}). At small fleet
sizes the counts themselves may leak coarse information about a
customer's deployment shape, which is why differential privacy on
the aggregate-telemetry channel is named as a future-work prerequisite
for cross-industry deployments. The evidence chain retains per-event
records under regime-aware retention floors; retention is bounded
above by the customer's policy and below by the regulatory minimum.

\paragraph{Stakeholders.}
The audiences that benefit directly from KYA's artifacts are
internal governance teams, external auditors, and regulators
performing post-hoc review. The audience whose interests the system
must \emph{not} silently advantage is the vendor of any individual
agent or framework. The closed-set extensibility discipline
(G7, \cref{sec:threat}) ensures that no vendor can expand the
scoring keyspace at runtime to advantage their own products without
an SDK release subject to community review. The fairness layer above
KYA (\cref{appx:fairness}) is the appropriate site for protected-class
analysis; KYA's substrate does not itself model protected attributes.

\paragraph{Operator-in-the-loop is non-negotiable.}
The system is designed so that automated governance adaptation
without a human gate is a structural impossibility, not a policy
choice. The cost of an automation failure where one false-positive
incident silently weakens the model is asymmetric to its benefit;
operator review is the bottleneck and we treat the pressure to
relax it as the most consequential failure mode of automated
governance infrastructure to resist.

\section*{Acknowledgments}
The author thanks Aran Nayebi (Carnegie Mellon University) for
pointing out Proposition~4 of \cite{nayebi2025corrigible} and
the framing of external safety-filter limits in
\cref{sec:limitations}.

\bibliographystyle{ACM-Reference-Format}
\bibliography{references}


\begin{thebibliography}{90}


\ifx \showCODEN    \undefined \def \showCODEN     #1{\unskip}     \fi
\ifx \showISBNx    \undefined \def \showISBNx     #1{\unskip}     \fi
\ifx \showISBNxiii \undefined \def \showISBNxiii  #1{\unskip}     \fi
\ifx \showISSN     \undefined \def \showISSN      #1{\unskip}     \fi
\ifx \showLCCN     \undefined \def \showLCCN      #1{\unskip}     \fi
\ifx \shownote     \undefined \def \shownote      #1{#1}          \fi
\ifx \showarticletitle \undefined \def \showarticletitle #1{#1}   \fi
\ifx \showURL      \undefined \def \showURL       {\relax}        \fi
\providecommand\bibfield[2]{#2}
\providecommand\bibinfo[2]{#2}
\providecommand\natexlab[1]{#1}
\providecommand\showeprint[2][]{arXiv:#2}

\bibitem[{Arize AI}(2024)]%
        {phoenix2024}
\bibfield{author}{\bibinfo{person}{{Arize AI}}.}
  \bibinfo{year}{2024}\natexlab{}.
\newblock \bibinfo{booktitle}{\emph{{Phoenix}: {AI} Observability and
  Evaluation}}.
\newblock


\bibitem[Back and von Wright(1998)]%
        {back1998refinement}
\bibfield{author}{\bibinfo{person}{Ralph-Johan Back} {and}
  \bibinfo{person}{Joakim von Wright}.} \bibinfo{year}{1998}\natexlab{}.
\newblock \bibinfo{booktitle}{\emph{Refinement Calculus: A Systematic
  Introduction}}.
\newblock \bibinfo{publisher}{Springer}.
\newblock


\bibitem[Bayer and {the SQLAlchemy authors}(2024)]%
        {sqlalchemy}
\bibfield{author}{\bibinfo{person}{Mike Bayer} {and} \bibinfo{person}{{the
  SQLAlchemy authors}}.} \bibinfo{year}{2024}\natexlab{}.
\newblock \bibinfo{title}{{SQLAlchemy} --- The Database Toolkit for {Python}}.
\newblock \bibinfo{howpublished}{\url{https://www.sqlalchemy.org/}}.
\newblock


\bibitem[Bell and LaPadula(1976)]%
        {bell1976unified}
\bibfield{author}{\bibinfo{person}{D.~Elliott Bell} {and}
  \bibinfo{person}{Leonard~J. LaPadula}.} \bibinfo{year}{1976}\natexlab{}.
\newblock \bibinfo{booktitle}{\emph{Secure Computer System: Unified Exposition
  and {Multics} Interpretation}}.
\newblock \bibinfo{type}{{T}echnical {R}eport} ESD-TR-75-306.
  \bibinfo{institution}{MITRE Corporation}.
\newblock


\bibitem[Bellare et~al\mbox{.}(1996)]%
        {Bellare96}
\bibfield{author}{\bibinfo{person}{Mihir Bellare}, \bibinfo{person}{Ran
  Canetti}, {and} \bibinfo{person}{Hugo Krawczyk}.}
  \bibinfo{year}{1996}\natexlab{}.
\newblock \showarticletitle{Keying Hash Functions for Message Authentication}.
  In \bibinfo{booktitle}{\emph{Advances in Cryptology --- {CRYPTO} '96}}.
  \bibinfo{publisher}{Springer}, \bibinfo{pages}{1--15}.
\newblock


\bibitem[Bellare and Yee(1997)]%
        {BellareYee1997}
\bibfield{author}{\bibinfo{person}{Mihir Bellare} {and}
  \bibinfo{person}{Bennet~S. Yee}.} \bibinfo{year}{1997}\natexlab{}.
\newblock \bibinfo{booktitle}{\emph{Forward Integrity for Secure Audit Logs}}.
\newblock \bibinfo{type}{{T}echnical {R}eport} CS98-580.
  \bibinfo{institution}{Dept. of Computer Science, UC San Diego}.
\newblock


\bibitem[Bernstein et~al\mbox{.}(2012)]%
        {Bernstein2012}
\bibfield{author}{\bibinfo{person}{Daniel~J. Bernstein}, \bibinfo{person}{Niels
  Duif}, \bibinfo{person}{Tanja Lange}, \bibinfo{person}{Peter Schwabe}, {and}
  \bibinfo{person}{Bo-Yin Yang}.} \bibinfo{year}{2012}\natexlab{}.
\newblock \showarticletitle{High-speed high-security signatures}.
\newblock \bibinfo{journal}{\emph{Journal of Cryptographic Engineering}}
  \bibinfo{volume}{2}, \bibinfo{number}{2} (\bibinfo{year}{2012}),
  \bibinfo{pages}{77--89}.
\newblock


\bibitem[{Board of Governors of the Federal Reserve System and Office of the
  Comptroller of the Currency}(2011)]%
        {SR117}
\bibfield{author}{\bibinfo{person}{{Board of Governors of the Federal Reserve
  System and Office of the Comptroller of the Currency}}.}
  \bibinfo{year}{2011}\natexlab{}.
\newblock \bibinfo{booktitle}{\emph{{SR} 11-7: Supervisory Guidance on Model
  Risk Management}}.
\newblock \bibinfo{type}{{T}echnical {R}eport}. \bibinfo{institution}{Federal
  Reserve}.
\newblock


\bibitem[Bonawitz et~al\mbox{.}(2017)]%
        {bonawitz2017practical}
\bibfield{author}{\bibinfo{person}{Keith Bonawitz}, \bibinfo{person}{Vladimir
  Ivanov}, \bibinfo{person}{Ben Kreuter}, \bibinfo{person}{Antonio Marcedone},
  \bibinfo{person}{H.~Brendan McMahan}, \bibinfo{person}{Sarvar Patel},
  \bibinfo{person}{Daniel Ramage}, \bibinfo{person}{Aaron Segal}, {and}
  \bibinfo{person}{Karn Seth}.} \bibinfo{year}{2017}\natexlab{}.
\newblock \showarticletitle{Practical Secure Aggregation for Privacy-Preserving
  Machine Learning}. In \bibinfo{booktitle}{\emph{Proceedings of the 2017 ACM
  Conference on Computer and Communications Security}}.
  \bibinfo{pages}{1175--1191}.
\newblock


\bibitem[{Braintrust Data}(2024)]%
        {braintrust2024}
\bibfield{author}{\bibinfo{person}{{Braintrust Data}}.}
  \bibinfo{year}{2024}\natexlab{}.
\newblock \bibinfo{title}{{Braintrust}: {AI} Evaluation Platform}.
\newblock \bibinfo{howpublished}{\url{https://braintrust.dev}}.
\newblock


\bibitem[Burke et~al\mbox{.}(2026)]%
        {burke2026rebound}
\bibfield{author}{\bibinfo{person}{Quinn Burke}, \bibinfo{person}{Anjo
  Vahldiek-Oberwagner}, \bibinfo{person}{Michael Swift}, {and}
  \bibinfo{person}{Patrick McDaniel}.} \bibinfo{year}{2026}\natexlab{}.
\newblock \showarticletitle{It's a Feature, Not a Bug: Secure and Auditable
  State Rollback for Confidential Cloud Applications}. In
  \bibinfo{booktitle}{\emph{IEEE Symposium on Security and Privacy (S\&P)}}.
\newblock
\showeprint[arxiv]{2511.13641}


\bibitem[Chaffer(2025)]%
        {chaffer2025kya}
\bibfield{author}{\bibinfo{person}{Tomer~Jordi Chaffer}.}
  \bibinfo{year}{2025}\natexlab{}.
\newblock \showarticletitle{Know Your Agent: Governing AI Identity on the
  Agentic Web}.
\newblock \bibinfo{journal}{\emph{SSRN Electronic Journal}}
  (\bibinfo{year}{2025}).
\newblock
\href{https://doi.org/10.2139/ssrn.5162127}{doi:\nolinkurl{10.2139/ssrn.5162127}}


\bibitem[Cleland-Huang et~al\mbox{.}(2022)]%
        {clelandhuang2022mapekhmt}
\bibfield{author}{\bibinfo{person}{Jane Cleland-Huang}, \bibinfo{person}{Ankit
  Agrawal}, \bibinfo{person}{Michael Vierhauser}, \bibinfo{person}{Michael
  Murphy}, {and} \bibinfo{person}{Mike Prieto}.}
  \bibinfo{year}{2022}\natexlab{}.
\newblock \showarticletitle{Extending {MAPE-K} to support Human-Machine
  Teaming}. In \bibinfo{booktitle}{\emph{Proc. 17th Symposium on Software
  Engineering for Adaptive and Self-Managing Systems (SEAMS '22)}}.
\newblock
\showeprint[arxiv]{2203.13036}


\bibitem[{Cloud Native Computing Foundation}(2024)]%
        {opentelemetry}
\bibfield{author}{\bibinfo{person}{{Cloud Native Computing Foundation}}.}
  \bibinfo{year}{2024}\natexlab{}.
\newblock \bibinfo{title}{{OpenTelemetry}: Observability Framework for
  Cloud-Native Software}.
\newblock \bibinfo{howpublished}{\url{https://opentelemetry.io/}}.
\newblock


\bibitem[{Cloudflare}(2022)]%
        {cloudflare_waf_ml_2022}
\bibfield{author}{\bibinfo{person}{{Cloudflare}}.}
  \bibinfo{year}{2022}\natexlab{}.
\newblock \bibinfo{title}{Improving the {WAF} with Machine Learning}.
\newblock \bibinfo{howpublished}{\url{https://blog.cloudflare.com/waf-ml/}}.
\newblock


\bibitem[Crosby and Wallach(2009)]%
        {crosbywallach2009tamper}
\bibfield{author}{\bibinfo{person}{Scott~A. Crosby} {and}
  \bibinfo{person}{Dan~S. Wallach}.} \bibinfo{year}{2009}\natexlab{}.
\newblock \showarticletitle{Efficient Data Structures for Tamper-Evident
  Logging}. In \bibinfo{booktitle}{\emph{USENIX Security Symposium}}.
\newblock


\bibitem[Cutler et~al\mbox{.}(2024)]%
        {cutler2024cedar}
\bibfield{author}{\bibinfo{person}{Joseph~W. Cutler}, \bibinfo{person}{Craig
  Disselkoen}, \bibinfo{person}{Aaron Eline}, \bibinfo{person}{Shaobo He},
  \bibinfo{person}{Kyle Headley}, \bibinfo{person}{Michael Hicks},
  \bibinfo{person}{Kesha Hietala}, \bibinfo{person}{John Kastner},
  \bibinfo{person}{Anwar Mamat}, \bibinfo{person}{Matt McCutchen},
  \bibinfo{person}{Neha Rungta}, \bibinfo{person}{Bhakti Shah},
  \bibinfo{person}{Emina Torlak}, {and} \bibinfo{person}{Andrew Wells}.}
  \bibinfo{year}{2024}\natexlab{}.
\newblock \showarticletitle{Cedar: A New Language for Expressive, Fast, Safe,
  and Analyzable Authorization}.
\newblock \bibinfo{journal}{\emph{Proc. ACM Program. Lang. (OOPSLA1)}}
  \bibinfo{volume}{8} (\bibinfo{year}{2024}), \bibinfo{pages}{670--697}.
\newblock
\showeprint[arxiv]{2403.04651}
\href{https://doi.org/10.1145/3649835}{doi:\nolinkurl{10.1145/3649835}}


\bibitem[Denning(1976)]%
        {denning1976lattice}
\bibfield{author}{\bibinfo{person}{Dorothy~E. Denning}.}
  \bibinfo{year}{1976}\natexlab{}.
\newblock \showarticletitle{A Lattice Model of Secure Information Flow}.
\newblock \bibinfo{journal}{\emph{Commun. ACM}} \bibinfo{volume}{19},
  \bibinfo{number}{5} (\bibinfo{year}{1976}), \bibinfo{pages}{236--243}.
\newblock


\bibitem[Derczynski et~al\mbox{.}(2024)]%
        {garak}
\bibfield{author}{\bibinfo{person}{Leon Derczynski}, \bibinfo{person}{Erick
  Galinkin}, \bibinfo{person}{Jeffrey Martin}, \bibinfo{person}{Subho
  Majumdar}, {and} \bibinfo{person}{Nanna Inie}.}
  \bibinfo{year}{2024}\natexlab{}.
\newblock \bibinfo{title}{{garak}: A Framework for Security Probing Large
  Language Models}.
\newblock \bibinfo{howpublished}{\url{https://github.com/NVIDIA/garak}}.
\newblock
\showeprint[arxiv]{2406.11036}


\bibitem[Dwork and Roth(2014)]%
        {dwork2014algorithmic}
\bibfield{author}{\bibinfo{person}{Cynthia Dwork} {and} \bibinfo{person}{Aaron
  Roth}.} \bibinfo{year}{2014}\natexlab{}.
\newblock \bibinfo{booktitle}{\emph{The Algorithmic Foundations of Differential
  Privacy}}. Vol.~\bibinfo{volume}{9}.
\newblock 211--407 pages.
\newblock


\bibitem[Es et~al\mbox{.}(2024)]%
        {es2024ragas}
\bibfield{author}{\bibinfo{person}{Shahul Es}, \bibinfo{person}{Jithin James},
  \bibinfo{person}{Luis Espinosa-Anke}, {and} \bibinfo{person}{Steven
  Schockaert}.} \bibinfo{year}{2024}\natexlab{}.
\newblock \showarticletitle{{RAGAs}: Automated Evaluation of Retrieval
  Augmented Generation}. In \bibinfo{booktitle}{\emph{Proc. 18th Conference of
  the European Chapter of the ACL: System Demonstrations}}.
  \bibinfo{pages}{150--158}.
\newblock


\bibitem[{European Parliament and Council}(2024)]%
        {EUAIAct2024}
\bibfield{author}{\bibinfo{person}{{European Parliament and Council}}.}
  \bibinfo{year}{2024}\natexlab{}.
\newblock \bibinfo{title}{Regulation ({EU}) 2024/1689 on Artificial
  Intelligence (Artificial Intelligence Act)}.
\newblock \bibinfo{howpublished}{Official Journal of the European Union}.
\newblock


\bibitem[Fernandez(2026)]%
        {fernandez2026acp}
\bibfield{author}{\bibinfo{person}{Marcelo Fernandez}.}
  \bibinfo{year}{2026}\natexlab{}.
\newblock \bibinfo{title}{Agent Control Protocol: Admission Control for Agent
  Actions}.
\newblock
\showeprint[arxiv]{2603.18829}


\bibitem[{FIRST.Org, Inc.}(2019)]%
        {first_cvss31_2019}
\bibfield{author}{\bibinfo{person}{{FIRST.Org, Inc.}}}
  \bibinfo{year}{2019}\natexlab{}.
\newblock \bibinfo{title}{Common Vulnerability Scoring System v3.1:
  Specification Document}.
\newblock
  \bibinfo{howpublished}{\url{https://www.first.org/cvss/v3.1/specification-document}}.
\newblock


\bibitem[Gabison and Xian(2025)]%
        {gabison2025liability}
\bibfield{author}{\bibinfo{person}{Garry~A. Gabison} {and}
  \bibinfo{person}{R.~Patrick Xian}.} \bibinfo{year}{2025}\natexlab{}.
\newblock \showarticletitle{Inherent and emergent liability issues in
  {LLM}-based agentic systems: a principal-agent perspective}. In
  \bibinfo{booktitle}{\emph{REALM Workshop at ACL 2025}}.
\newblock
\showeprint[arxiv]{2504.03255}


\bibitem[{Galileo Technologies}(2024)]%
        {galileo2024}
\bibfield{author}{\bibinfo{person}{{Galileo Technologies}}.}
  \bibinfo{year}{2024}\natexlab{}.
\newblock \bibinfo{title}{{Galileo}: {AI} Reliability and Evaluation Platform}.
\newblock \bibinfo{howpublished}{\url{https://galileo.ai}}.
\newblock


\bibitem[Gao et~al\mbox{.}(2024)]%
        {eval-harness}
\bibfield{author}{\bibinfo{person}{Leo Gao}, \bibinfo{person}{Jonathan Tow},
  \bibinfo{person}{Baber Abbasi}, \bibinfo{person}{Stella Biderman},
  {et~al\mbox{.}}} \bibinfo{year}{2024}\natexlab{}.
\newblock \bibinfo{title}{A framework for few-shot language model evaluation}.
\newblock
\href{https://doi.org/10.5281/zenodo.12608602}{doi:\nolinkurl{10.5281/zenodo.12608602}}


\bibitem[Gaurav et~al\mbox{.}(2025)]%
        {gaurav2025gaas}
\bibfield{author}{\bibinfo{person}{Suyash Gaurav}, \bibinfo{person}{Jukka
  Heikkonen}, {and} \bibinfo{person}{Jatin Chaudhary}.}
  \bibinfo{year}{2025}\natexlab{}.
\newblock \showarticletitle{Governance-as-a-Service: A Multi-Agent Framework
  for {AI} System Compliance and Policy Enforcement}.
\newblock \bibinfo{journal}{\emph{arXiv preprint}} (\bibinfo{year}{2025}).
\newblock
\showeprint[arxiv]{2508.18765}


\bibitem[{Google Cloud}(2024)]%
        {gcp_deny_policies}
\bibfield{author}{\bibinfo{person}{{Google Cloud}}.}
  \bibinfo{year}{2024}\natexlab{}.
\newblock \bibinfo{title}{{Google Cloud} {IAM} Deny Policies}.
\newblock
  \bibinfo{howpublished}{\url{https://cloud.google.com/iam/docs/deny-overview}}.
\newblock


\bibitem[Greshake et~al\mbox{.}(2023)]%
        {greshake2023notwhat}
\bibfield{author}{\bibinfo{person}{Kai Greshake}, \bibinfo{person}{Sahar
  Abdelnabi}, \bibinfo{person}{Shailesh Mishra}, \bibinfo{person}{Christoph
  Endres}, \bibinfo{person}{Thorsten Holz}, {and} \bibinfo{person}{Mario
  Fritz}.} \bibinfo{year}{2023}\natexlab{}.
\newblock \showarticletitle{Not What You've Signed Up For: Compromising
  Real-World {LLM}-Integrated Applications with Indirect Prompt Injection}. In
  \bibinfo{booktitle}{\emph{Proc. 16th ACM Workshop on Artificial Intelligence
  and Security (AISec)}}. \bibinfo{pages}{79--90}.
\newblock


\bibitem[Huang et~al\mbox{.}(2025)]%
        {huang2025aagate}
\bibfield{author}{\bibinfo{person}{Ken Huang}, \bibinfo{person}{Kyriakos~Rock
  Lambros}, \bibinfo{person}{Jerry Huang}, \bibinfo{person}{Yasir Mehmood},
  \bibinfo{person}{Hammad Atta}, \bibinfo{person}{Joshua Beck},
  \bibinfo{person}{Vineeth~Sai Narajala}, \bibinfo{person}{Muhammad~Zeeshan
  Baig}, {et~al\mbox{.}}} \bibinfo{year}{2025}\natexlab{}.
\newblock \bibinfo{title}{{AAGATE}: A {NIST} {AI} {RMF}-Aligned Governance
  Platform for Agentic {AI}}.
\newblock
\showeprint[arxiv]{2510.25863}~[cs.CR]


\bibitem[{IBM}(2024)]%
        {ibm_soar_2024}
\bibfield{author}{\bibinfo{person}{{IBM}}.} \bibinfo{year}{2024}\natexlab{}.
\newblock \bibinfo{title}{What is {SOAR} (Security Orchestration, Automation
  and Response)?}
\newblock
  \bibinfo{howpublished}{\url{https://www.ibm.com/think/topics/security-orchestration-automation-response}}.
\newblock


\bibitem[{International Organization for Standardization}(2023)]%
        {ISO42001}
\bibfield{author}{\bibinfo{person}{{International Organization for
  Standardization}}.} \bibinfo{year}{2023}\natexlab{}.
\newblock \bibinfo{title}{{ISO}/{IEC} 42001:2023 --- Information technology ---
  Artificial intelligence --- Management system}.
\newblock


\bibitem[Ioannidis et~al\mbox{.}(2000)]%
        {ioannidis2000distfw}
\bibfield{author}{\bibinfo{person}{Sotiris Ioannidis},
  \bibinfo{person}{Angelos~D. Keromytis}, \bibinfo{person}{Steven~M. Bellovin},
  {and} \bibinfo{person}{Jonathan~M. Smith}.} \bibinfo{year}{2000}\natexlab{}.
\newblock \showarticletitle{Implementing a Distributed Firewall}. In
  \bibinfo{booktitle}{\emph{Proc. 7th ACM Conf. on Computer and Communications
  Security (CCS)}}. \bibinfo{publisher}{ACM}, \bibinfo{pages}{190--199}.
\newblock


\bibitem[Ip and Vongthongsri(2024)]%
        {deepeval2024}
\bibfield{author}{\bibinfo{person}{Jeffrey Ip} {and} \bibinfo{person}{Kritin
  Vongthongsri}.} \bibinfo{year}{2024}\natexlab{}.
\newblock \bibinfo{booktitle}{\emph{{DeepEval}}}.
\newblock
\urldef\tempurl%
\url{https://github.com/confident-ai/deepeval}
\showURL{%
\tempurl}


\bibitem[ISO/IEC(2025)]%
        {iso_42005_2025}
ISO/IEC \bibinfo{year}{2025}\natexlab{}.
\newblock \bibinfo{booktitle}{\emph{{ISO/IEC} 42005:2025 Information technology
  --- Artificial intelligence --- {AI} System Impact Assessment}}.
\newblock ISO/IEC.
\newblock


\bibitem[{Istio Project}(2024)]%
        {istio_authz}
\bibfield{author}{\bibinfo{person}{{Istio Project}}.}
  \bibinfo{year}{2024}\natexlab{}.
\newblock \bibinfo{title}{Istio Authorization Policy}.
\newblock
  \bibinfo{howpublished}{\url{https://istio.io/latest/docs/reference/config/security/authorization-policy/}}.
\newblock


\bibitem[Janani(2025)]%
        {janani2025humanmachineblur}
\bibfield{author}{\bibinfo{person}{Kush Janani}.}
  \bibinfo{year}{2025}\natexlab{}.
\newblock \showarticletitle{The Human-Machine Identity Blur: A Unified
  Framework for Cybersecurity Risk Management in 2025}.
\newblock  (\bibinfo{year}{2025}).
\newblock
\showeprint[arxiv]{2503.18255}


\bibitem[Jordan et~al\mbox{.}(2021)]%
        {STIX21_OASIS}
\bibfield{author}{\bibinfo{person}{Bret Jordan}, \bibinfo{person}{Rich Piazza},
  {and} \bibinfo{person}{Trey Darley}.} \bibinfo{year}{2021}\natexlab{}.
\newblock \bibinfo{booktitle}{\emph{{STIX} Version 2.1, {OASIS} Standard}}.
\newblock \bibinfo{type}{{T}echnical {R}eport}. \bibinfo{institution}{OASIS}.
\newblock


\bibitem[Kairouz et~al\mbox{.}(2021)]%
        {kairouz2021advances}
\bibfield{author}{\bibinfo{person}{Peter Kairouz}, \bibinfo{person}{H.~Brendan
  McMahan}, \bibinfo{person}{Brendan Avent}, \bibinfo{person}{Aurelien Bellet},
  {et~al\mbox{.}}} \bibinfo{year}{2021}\natexlab{}.
\newblock \showarticletitle{Advances and Open Problems in Federated Learning}.
\newblock \bibinfo{journal}{\emph{Foundations and Trends in Machine Learning}}
  \bibinfo{volume}{14}, \bibinfo{number}{1--2} (\bibinfo{year}{2021}),
  \bibinfo{pages}{1--210}.
\newblock


\bibitem[Kamvar et~al\mbox{.}(2003)]%
        {kamvar2003eigentrust}
\bibfield{author}{\bibinfo{person}{Sepandar~D. Kamvar},
  \bibinfo{person}{Mario~T. Schlosser}, {and} \bibinfo{person}{Hector
  Garcia-Molina}.} \bibinfo{year}{2003}\natexlab{}.
\newblock \showarticletitle{The {EigenTrust} Algorithm for Reputation
  Management in {P2P} Networks}. In \bibinfo{booktitle}{\emph{Proc. 12th Int.
  Conf. on World Wide Web (WWW)}}. \bibinfo{publisher}{ACM},
  \bibinfo{pages}{640--651}.
\newblock


\bibitem[Kaptein et~al\mbox{.}(2026)]%
        {kaptein2026runtime}
\bibfield{author}{\bibinfo{person}{Maurits Kaptein},
  \bibinfo{person}{Vassilis-Javed Khan}, {and} \bibinfo{person}{Andriy
  Podstavnychy}.} \bibinfo{year}{2026}\natexlab{}.
\newblock \showarticletitle{Runtime Governance for {AI} Agents: Policies on
  Paths}.
\newblock \bibinfo{journal}{\emph{arXiv preprint}} (\bibinfo{year}{2026}).
\newblock
\showeprint[arxiv]{2603.16586}


\bibitem[Kephart and Chess(2003)]%
        {kephart2003vision}
\bibfield{author}{\bibinfo{person}{Jeffrey~O. Kephart} {and}
  \bibinfo{person}{David~M. Chess}.} \bibinfo{year}{2003}\natexlab{}.
\newblock \showarticletitle{The Vision of Autonomic Computing}.
\newblock \bibinfo{journal}{\emph{IEEE Computer}} \bibinfo{volume}{36},
  \bibinfo{number}{1} (\bibinfo{year}{2003}), \bibinfo{pages}{41--50}.
\newblock


\bibitem[Klingen et~al\mbox{.}(2023)]%
        {langfuse2023}
\bibfield{author}{\bibinfo{person}{Marc Klingen}, \bibinfo{person}{Maximilian
  Deichmann}, {and} \bibinfo{person}{Clemens Rawert}.}
  \bibinfo{year}{2023}\natexlab{}.
\newblock \bibinfo{booktitle}{\emph{{Langfuse}: Open Source {LLM} Engineering
  Platform}}.
\newblock


\bibitem[Kuppusamy et~al\mbox{.}(2016)]%
        {kuppusamy2016uptane}
\bibfield{author}{\bibinfo{person}{Trishank~Karthik Kuppusamy},
  \bibinfo{person}{Akan Brown}, \bibinfo{person}{Sebastien Awwad},
  \bibinfo{person}{Damon McCoy}, \bibinfo{person}{Russ Bielawski},
  \bibinfo{person}{Cameron Mott}, \bibinfo{person}{Sam Lauzon},
  \bibinfo{person}{Andre Weimerskirch}, {and} \bibinfo{person}{Justin Cappos}.}
  \bibinfo{year}{2016}\natexlab{}.
\newblock \showarticletitle{{Uptane}: Securing Software Updates for
  Automobiles}. In \bibinfo{booktitle}{\emph{14th ESCAR Europe}}.
\newblock


\bibitem[{LangChain, Inc.}(2024)]%
        {langsmith2024}
\bibfield{author}{\bibinfo{person}{{LangChain, Inc.}}}
  \bibinfo{year}{2024}\natexlab{}.
\newblock \bibinfo{title}{{LangSmith}: Observability and Evaluation Platform
  for {LLM} Applications}.
\newblock \bibinfo{howpublished}{\url{https://smith.langchain.com}}.
\newblock


\bibitem[Laurie et~al\mbox{.}(2013)]%
        {RFC6962}
\bibfield{author}{\bibinfo{person}{Ben Laurie}, \bibinfo{person}{Adam Langley},
  {and} \bibinfo{person}{Emilia Kasper}.} \bibinfo{year}{2013}\natexlab{}.
\newblock \bibinfo{booktitle}{\emph{Certificate Transparency}}.
\newblock \bibinfo{type}{RFC} 6962. \bibinfo{institution}{IETF}.
\newblock


\bibitem[Liang et~al\mbox{.}(2023)]%
        {liang2023helm}
\bibfield{author}{\bibinfo{person}{Percy Liang}, \bibinfo{person}{Rishi
  Bommasani}, \bibinfo{person}{Tony Lee}, {et~al\mbox{.}}}
  \bibinfo{year}{2023}\natexlab{}.
\newblock \showarticletitle{Holistic Evaluation of Language Models}.
\newblock \bibinfo{journal}{\emph{Transactions on Machine Learning Research}}
  (\bibinfo{year}{2023}).
\newblock


\bibitem[Liang et~al\mbox{.}(2025)]%
        {liang2025donttrust}
\bibfield{author}{\bibinfo{person}{Ruichao Liang}, \bibinfo{person}{Le Yin},
  \bibinfo{person}{Jing Chen}, \bibinfo{person}{Yebo Feng},
  \bibinfo{person}{Cong Wu}, \bibinfo{person}{Xiaoyu Zhang},
  \bibinfo{person}{Huangpeng Gu}, \bibinfo{person}{Zijian Zhang}, {and}
  \bibinfo{person}{Yang Liu}.} \bibinfo{year}{2025}\natexlab{}.
\newblock \showarticletitle{Don't Trust Your Upstream: Exploiting LLM
  Multi-Agent System via Topology-Guided Adversarial Propagation}.
\newblock \bibinfo{journal}{\emph{arXiv preprint}} (\bibinfo{year}{2025}).
\newblock
\showeprint[arxiv]{2512.04129}


\bibitem[Mazzocchetti(2026)]%
        {Aegis2026}
\bibfield{author}{\bibinfo{person}{Adam~Massimo Mazzocchetti}.}
  \bibinfo{year}{2026}\natexlab{}.
\newblock \bibinfo{title}{Cryptographic Runtime Governance for Autonomous {AI}
  Systems: The {Aegis} Architecture for Verifiable Policy Enforcement}.
\newblock
\showeprint[arxiv]{2603.16938}~[cs.CR]


\bibitem[McMahan et~al\mbox{.}(2017)]%
        {FedAvg}
\bibfield{author}{\bibinfo{person}{Brendan McMahan}, \bibinfo{person}{Eider
  Moore}, \bibinfo{person}{Daniel Ramage}, \bibinfo{person}{Seth Hampson},
  {and} \bibinfo{person}{Blaise~Ag\"{u}era y Arcas}.}
  \bibinfo{year}{2017}\natexlab{}.
\newblock \showarticletitle{Communication-Efficient Learning of Deep Networks
  from Decentralized Data}. In \bibinfo{booktitle}{\emph{Artificial
  Intelligence and Statistics (AISTATS)}}.
\newblock


\bibitem[{Microsoft Agent Governance Toolkit contributors}(2026)]%
        {microsoft_agt_issue1386}
\bibfield{author}{\bibinfo{person}{{Microsoft Agent Governance Toolkit
  contributors}}.} \bibinfo{year}{2026}\natexlab{}.
\newblock \bibinfo{title}{Cross-organization policy federation (follow-up to
  {ADR}-0007)}.
\newblock
  \bibinfo{howpublished}{\url{https://github.com/microsoft/agent-governance-toolkit/issues/1386}}.
\newblock


\bibitem[{Microsoft AI Red Team}(2024)]%
        {pyrit}
\bibfield{author}{\bibinfo{person}{{Microsoft AI Red Team}}.}
  \bibinfo{year}{2024}\natexlab{}.
\newblock \bibinfo{title}{{PyRIT}: Python Risk Identification Toolkit for
  Generative AI}.
\newblock \bibinfo{howpublished}{\url{https://github.com/Azure/PyRIT}}.
\newblock


\bibitem[{Microsoft Learn}(2024)]%
        {windows_appcontrol_signed}
\bibfield{author}{\bibinfo{person}{{Microsoft Learn}}.}
  \bibinfo{year}{2024}\natexlab{}.
\newblock \bibinfo{title}{Use Signed Policies to Protect Windows Defender
  Application Control Against Tampering}.
\newblock
  \bibinfo{howpublished}{\url{https://learn.microsoft.com/en-us/windows/security/application-security/application-control/}}.
\newblock


\bibitem[{Microsoft Open Source}(2026)]%
        {microsoft_agt_2026}
\bibfield{author}{\bibinfo{person}{{Microsoft Open Source}}.}
  \bibinfo{year}{2026}\natexlab{}.
\newblock \bibinfo{title}{Agent Governance Toolkit: Open-source Runtime
  Security for {AI} Agents}.
\newblock
  \bibinfo{howpublished}{\url{https://github.com/microsoft/agent-governance-toolkit}}.
\newblock
\newblock
\shownote{Public preview, MIT License, April 2026}.


\bibitem[Miller et~al\mbox{.}(2007)]%
        {miller2007horton}
\bibfield{author}{\bibinfo{person}{Mark~S. Miller}, \bibinfo{person}{Jed
  Donnelley}, {and} \bibinfo{person}{Alan~H. Karp}.}
  \bibinfo{year}{2007}\natexlab{}.
\newblock \showarticletitle{Delegating Responsibility in Digital Systems:
  {Horton}'s ``{Who} Done It?''}. In \bibinfo{booktitle}{\emph{2nd USENIX
  Workshop on Hot Topics in Security (HotSec)}}.
\newblock


\bibitem[Nakajima(2026)]%
        {nakajima2026activegraph}
\bibfield{author}{\bibinfo{person}{Yohei Nakajima}.}
  \bibinfo{year}{2026}\natexlab{}.
\newblock \bibinfo{title}{The Log is the Agent: Event-Sourced Reactive Graphs
  for Auditable, Forkable Agentic Systems}.
\newblock \bibinfo{howpublished}{\url{https://arxiv.org/abs/2605.21997}}.
\newblock
\showeprint[arxiv]{2605.21997}~[cs.AI]


\bibitem[{National Institute of Standards and Technology}(2023)]%
        {NISTAIRMF}
\bibfield{author}{\bibinfo{person}{{National Institute of Standards and
  Technology}}.} \bibinfo{year}{2023}\natexlab{}.
\newblock \bibinfo{booktitle}{\emph{Artificial Intelligence Risk Management
  Framework ({AI} {RMF} 1.0)}}.
\newblock \bibinfo{type}{{T}echnical {R}eport} NIST AI 100-1.
  \bibinfo{institution}{U.S. Department of Commerce}.
\newblock


\bibitem[{National Institute of Standards and Technology}(2024)]%
        {nist_ai_600_1}
\bibfield{author}{\bibinfo{person}{{National Institute of Standards and
  Technology}}.} \bibinfo{year}{2024}\natexlab{}.
\newblock \bibinfo{booktitle}{\emph{Artificial Intelligence Risk Management
  Framework: Generative {AI} Profile}}.
\newblock \bibinfo{type}{{T}echnical {R}eport} NIST AI 600-1.
  \bibinfo{institution}{U.S. Department of Commerce}.
\newblock


\bibitem[Nayebi(2025)]%
        {nayebi2025corrigible}
\bibfield{author}{\bibinfo{person}{Aran Nayebi}.}
  \bibinfo{year}{2025}\natexlab{}.
\newblock \showarticletitle{Core Safety Values for Provably Corrigible Agents}.
  In \bibinfo{booktitle}{\emph{AAAI 2026 Machine Ethics Workshop (W37)}}.
\newblock
\showeprint[arxiv]{2507.20964}


\bibitem[Newman et~al\mbox{.}(2022)]%
        {Newman2022Sigstore}
\bibfield{author}{\bibinfo{person}{Zachary Newman}, \bibinfo{person}{John~Speed
  Meyers}, {and} \bibinfo{person}{Santiago Torres-Arias}.}
  \bibinfo{year}{2022}\natexlab{}.
\newblock \showarticletitle{{Sigstore}: Software Signing for Everybody}. In
  \bibinfo{booktitle}{\emph{ACM CCS}}.
\newblock


\bibitem[{Okta}(2026)]%
        {okta_nhi_unified_2026}
\bibfield{author}{\bibinfo{person}{{Okta}}.} \bibinfo{year}{2026}\natexlab{}.
\newblock \bibinfo{title}{Non-human and Human Identities: A Unified Approach}.
\newblock
  \bibinfo{howpublished}{\url{https://www.okta.com/blog/ai/non-human-and-human-identities-a-unified-approach/}}.
\newblock


\bibitem[{Open Policy Agent Authors}(2024)]%
        {OPA_SignedBundles}
\bibfield{author}{\bibinfo{person}{{Open Policy Agent Authors}}.}
  \bibinfo{year}{2024}\natexlab{}.
\newblock \bibinfo{title}{Open Policy Agent: Bundles --- Signing and
  Verification}.
\newblock
  \bibinfo{howpublished}{\url{https://www.openpolicyagent.org/docs/management-bundles}}.
\newblock


\bibitem[{OpenAI}(2023)]%
        {openai_evals}
\bibfield{author}{\bibinfo{person}{{OpenAI}}.} \bibinfo{year}{2023}\natexlab{}.
\newblock \bibinfo{title}{Evals: A Framework for Evaluating {LLMs}}.
\newblock \bibinfo{howpublished}{\url{https://github.com/openai/evals}}.
\newblock


\bibitem[{OWASP Foundation}(2025)]%
        {owasp_aivss_2025}
\bibfield{author}{\bibinfo{person}{{OWASP Foundation}}.}
  \bibinfo{year}{2025}\natexlab{}.
\newblock \bibinfo{booktitle}{\emph{{AIVSS} Scoring System For {OWASP} Agentic
  {AI} Core Security Risks v0.5}}.
\newblock \bibinfo{type}{{T}echnical {R}eport}. \bibinfo{institution}{OWASP}.
\newblock
\urldef\tempurl%
\url{https://aivss.owasp.org/}
\showURL{%
\tempurl}


\bibitem[Ramage and Mazzocchi(2020)]%
        {ramage2020federated}
\bibfield{author}{\bibinfo{person}{Daniel Ramage} {and}
  \bibinfo{person}{Stefano Mazzocchi}.} \bibinfo{year}{2020}\natexlab{}.
\newblock \bibinfo{title}{Federated Analytics: Collaborative Data Science
  without Data Collection}.
\newblock \bibinfo{howpublished}{Google Research Blog}.
\newblock


\bibitem[Ravi et~al\mbox{.}(2024)]%
        {ravi2024lynx}
\bibfield{author}{\bibinfo{person}{Selvan~Sunitha Ravi},
  \bibinfo{person}{Bartosz Mielczarek}, \bibinfo{person}{Anand Kannappan},
  \bibinfo{person}{Douwe Kiela}, {and} \bibinfo{person}{Rebecca Qian}.}
  \bibinfo{year}{2024}\natexlab{}.
\newblock \showarticletitle{{Lynx}: An Open Source Hallucination Evaluation
  Model}. In \bibinfo{booktitle}{\emph{arXiv:2407.08488}}.
\newblock


\bibitem[Rundgren et~al\mbox{.}(2020)]%
        {RFC8785}
\bibfield{author}{\bibinfo{person}{Anders Rundgren}, \bibinfo{person}{Bret
  Jordan}, {and} \bibinfo{person}{Samuel Erdtman}.}
  \bibinfo{year}{2020}\natexlab{}.
\newblock \bibinfo{booktitle}{\emph{{JSON} Canonicalization Scheme ({JCS})}}.
\newblock \bibinfo{type}{RFC} 8785. \bibinfo{institution}{IETF}.
\newblock


\bibitem[Samuel et~al\mbox{.}(2010)]%
        {samuel2010tuf}
\bibfield{author}{\bibinfo{person}{Justin Samuel}, \bibinfo{person}{Nick
  Mathewson}, \bibinfo{person}{Justin Cappos}, {and} \bibinfo{person}{Roger
  Dingledine}.} \bibinfo{year}{2010}\natexlab{}.
\newblock \showarticletitle{Survivable Key Compromise in Software Update
  Systems}. In \bibinfo{booktitle}{\emph{Proc. 17th ACM Conf. on Computer and
  Communications Security (CCS)}}. \bibinfo{publisher}{ACM},
  \bibinfo{pages}{61--72}.
\newblock


\bibitem[Sanwouo et~al\mbox{.}(2025)]%
        {sanwouo2025aware}
\bibfield{author}{\bibinfo{person}{Brell Sanwouo} {et~al\mbox{.}}}
  \bibinfo{year}{2025}\natexlab{}.
\newblock \showarticletitle{Breaking the Loop: {AWARE} is the New {MAPE-K}}. In
  \bibinfo{booktitle}{\emph{Proc. 33rd ACM Intl. Conf. on the Foundations of
  Software Engineering (FSE 2025) --- Ideas, Visions, and Reflections}}.
\newblock
\href{https://doi.org/10.1145/3696630.3728512}{doi:\nolinkurl{10.1145/3696630.3728512}}


\bibitem[Satta~Chiris et~al\mbox{.}(2025)]%
        {satta2025aura}
\bibfield{author}{\bibinfo{person}{Lorenzo Satta~Chiris},
  \bibinfo{person}{Ayush Mishra}, {et~al\mbox{.}}}
  \bibinfo{year}{2025}\natexlab{}.
\newblock \showarticletitle{AURA: An Agent Autonomy Risk Assessment Framework}.
\newblock \bibinfo{journal}{\emph{arXiv preprint}} (\bibinfo{year}{2025}).
\newblock
\showeprint[arxiv]{2510.15739}


\bibitem[Schneier and Kelsey(1998)]%
        {schneierkelsey1998logs}
\bibfield{author}{\bibinfo{person}{Bruce Schneier} {and} \bibinfo{person}{John
  Kelsey}.} \bibinfo{year}{1998}\natexlab{}.
\newblock \showarticletitle{Cryptographic Support for Secure Logs on Untrusted
  Machines}. In \bibinfo{booktitle}{\emph{USENIX Security Symposium}}.
\newblock


\bibitem[Schubert et~al\mbox{.}(2023)]%
        {schubert2023deep}
\bibfield{author}{\bibinfo{person}{Marius Schubert}, \bibinfo{person}{Tobias
  Riedlinger}, \bibinfo{person}{Karsten Kahl}, {and} \bibinfo{person}{Matthias
  Rottmann}.} \bibinfo{year}{2023}\natexlab{}.
\newblock \showarticletitle{Deep Active Learning with Noisy Oracle in Object
  Detection}.
\newblock \bibinfo{journal}{\emph{arXiv preprint}} (\bibinfo{year}{2023}).
\newblock
\showeprint[arxiv]{2310.00372}


\bibitem[Shen et~al\mbox{.}(2026)]%
        {shen2026sigil}
\bibfield{author}{\bibinfo{person}{Tingda Shen}, \bibinfo{person}{Yebo Feng},
  \bibinfo{person}{Konglin Zhu}, \bibinfo{person}{Xiaojun Jia},
  \bibinfo{person}{Yang Liu}, {and} \bibinfo{person}{Lin Zhang}.}
  \bibinfo{year}{2026}\natexlab{}.
\newblock \bibinfo{title}{Sealing the Audit-Runtime Gap for {LLM} Skills}.
\newblock
\showeprint[arxiv]{2605.05274}


\bibitem[Soares et~al\mbox{.}(2015)]%
        {soares2015corrigibility}
\bibfield{author}{\bibinfo{person}{Nate Soares}, \bibinfo{person}{Benja
  Fallenstein}, \bibinfo{person}{Eliezer Yudkowsky}, {and}
  \bibinfo{person}{Stuart Armstrong}.} \bibinfo{year}{2015}\natexlab{}.
\newblock \showarticletitle{Corrigibility}. In
  \bibinfo{booktitle}{\emph{Workshops at the Twenty-Ninth {AAAI} Conference on
  Artificial Intelligence}}.
\newblock


\bibitem[Souza et~al\mbox{.}(2025)]%
        {souza2025provagent}
\bibfield{author}{\bibinfo{person}{Renan Souza}, \bibinfo{person}{Amal
  Gueroudji}, \bibinfo{person}{Stephen DeWitt}, \bibinfo{person}{Daniel
  Rosendo}, \bibinfo{person}{Tirthankar Ghosal}, \bibinfo{person}{Robert Ross},
  \bibinfo{person}{Prasanna Balaprakash}, {and}
  \bibinfo{person}{Rafael~Ferreira da Silva}.} \bibinfo{year}{2025}\natexlab{}.
\newblock \showarticletitle{{PROV-AGENT}: Unified Provenance for Tracking {AI}
  Agent Interactions in Agentic Workflows}. In \bibinfo{booktitle}{\emph{IEEE
  Int'l Conf. on e-Science}}.
\newblock
\newblock
\shownote{arXiv:2508.02866}.


\bibitem[Tan et~al\mbox{.}(2026)]%
        {tan2026attesting}
\bibfield{author}{\bibinfo{person}{Zhuoran Tan}, \bibinfo{person}{Jeremy
  Singer}, {and} \bibinfo{person}{Christos Anagnostopoulos}.}
  \bibinfo{year}{2026}\natexlab{}.
\newblock \bibinfo{title}{Attesting {LLM} Pipelines: Enforcing Verifiable
  Training and Release Claims}.
\newblock
\showeprint[arxiv]{2603.28988}
\newblock
\shownote{2nd Intl. Workshop on Large Language Model Supply Chain Analysis
  (LLMSC 2026)}.


\bibitem[The Open Group(2021)]%
        {opengroup_fair_2021}
The Open Group \bibinfo{year}{2021}\natexlab{}.
\newblock \bibinfo{booktitle}{\emph{Risk Taxonomy ({O-RT}), Version 3.0}}.
\newblock The Open Group.
\newblock


\bibitem[Torres-Arias et~al\mbox{.}(2019)]%
        {TorresArias2019Intoto}
\bibfield{author}{\bibinfo{person}{Santiago Torres-Arias},
  \bibinfo{person}{Hammad Afzali}, \bibinfo{person}{Trishank~Karthik
  Kuppusamy}, \bibinfo{person}{Reza Curtmola}, {and} \bibinfo{person}{Justin
  Cappos}.} \bibinfo{year}{2019}\natexlab{}.
\newblock \showarticletitle{in-toto: Providing Farm-to-Table Guarantees for
  Bits and Bytes}. In \bibinfo{booktitle}{\emph{28th USENIX Security
  Symposium}}. \bibinfo{pages}{1393--1410}.
\newblock


\bibitem[{Truera}(2023)]%
        {trulens2023}
\bibfield{author}{\bibinfo{person}{{Truera}}.} \bibinfo{year}{2023}\natexlab{}.
\newblock \bibinfo{booktitle}{\emph{{TruLens}: Evaluation and Tracking for
  {LLM} Experiments}}.
\newblock
\urldef\tempurl%
\url{https://github.com/truera/trulens}
\showURL{%
\tempurl}


\bibitem[{TrustPact}(2026)]%
        {trustpact2026}
\bibfield{author}{\bibinfo{person}{{TrustPact}}.}
  \bibinfo{year}{2026}\natexlab{}.
\newblock \bibinfo{title}{{TrustPact}: Behavioral Trust Scanner for {MCP}
  Servers and {AI} Agents}.
\newblock \bibinfo{howpublished}{\url{https://pypi.org/project/trustpact/}}.
\newblock


\bibitem[{UK AI Security Institute}(2024)]%
        {aisi2024inspect}
\bibfield{author}{\bibinfo{person}{{UK AI Security Institute}}.}
  \bibinfo{year}{2024}\natexlab{}.
\newblock \bibinfo{booktitle}{\emph{{Inspect AI}: Framework for Large Language
  Model Evaluations}}.
\newblock
\urldef\tempurl%
\url{https://github.com/UKGovernmentBEIS/inspect_ai}
\showURL{%
\tempurl}


\bibitem[{U.S. Food and Drug Administration}(2024)]%
        {fda2024pccp}
\bibfield{author}{\bibinfo{person}{{U.S. Food and Drug Administration}}.}
  \bibinfo{year}{2024}\natexlab{}.
\newblock \bibinfo{title}{Marketing Submission Recommendations for a
  Predetermined Change Control Plan for Artificial Intelligence-Enabled Device
  Software Functions}.
\newblock \bibinfo{howpublished}{Final Guidance for Industry and FDA Staff}.
\newblock
\urldef\tempurl%
\url{https://www.fda.gov/medical-devices/software-medical-device-samd/artificial-intelligence-software-medical-device}
\showURL{%
\tempurl}


\bibitem[Vassilev et~al\mbox{.}(2025)]%
        {nist_ai_100_2_e2025}
\bibfield{author}{\bibinfo{person}{Apostol Vassilev}, \bibinfo{person}{Alina
  Oprea}, \bibinfo{person}{Alie Fordyce}, \bibinfo{person}{Hyrum Anderson},
  \bibinfo{person}{Xander Davis}, {and} \bibinfo{person}{Alina Marshall}.}
  \bibinfo{year}{2025}\natexlab{}.
\newblock \bibinfo{booktitle}{\emph{Adversarial Machine Learning: A Taxonomy
  and Terminology of Attacks and Mitigations}}.
\newblock \bibinfo{type}{{T}echnical {R}eport} NIST AI 100-2 E2025.
  \bibinfo{institution}{NIST}.
\newblock


\bibitem[{Veldt Labs}(2026)]%
        {veldt-kya}
\bibfield{author}{\bibinfo{person}{{Veldt Labs}}.}
  \bibinfo{year}{2026}\natexlab{}.
\newblock \bibinfo{title}{{veldt-kya}: Know Your Agents --- a trust and
  governance {SDK} for autonomous systems}.
\newblock \bibinfo{howpublished}{\url{https://github.com/veldtlabs/veldt-kya}}.
\newblock
\newblock
\shownote{Open-source under Apache 2.0; available on PyPI}.


\bibitem[{Weights and Biases}(2024)]%
        {weave2024}
\bibfield{author}{\bibinfo{person}{{Weights and Biases}}.}
  \bibinfo{year}{2024}\natexlab{}.
\newblock \bibinfo{title}{{Weave}: Toolkit for {LLM} Application Development}.
\newblock \bibinfo{howpublished}{\url{https://wandb.ai/site/weave}}.
\newblock


\bibitem[{World Wide Web Consortium}(2024)]%
        {w3c_baggage}
\bibfield{author}{\bibinfo{person}{{World Wide Web Consortium}}.}
  \bibinfo{year}{2024}\natexlab{}.
\newblock \bibinfo{booktitle}{\emph{Propagation Format for Distributed Context:
  {Baggage}}}.
\newblock \bibinfo{type}{W3C Working Draft}.
\newblock


\bibitem[Wu et~al\mbox{.}(2025)]%
        {wu2025monitoring}
\bibfield{author}{\bibinfo{person}{Chengcan Wu}, \bibinfo{person}{Zhixin
  Zhang}, \bibinfo{person}{Mingqian Xu}, \bibinfo{person}{Zeming Wei}, {and}
  \bibinfo{person}{Meng Sun}.} \bibinfo{year}{2025}\natexlab{}.
\newblock \showarticletitle{Monitoring LLM-based Multi-Agent Systems Against
  Corruptions via Node Evaluation}.
\newblock \bibinfo{journal}{\emph{arXiv preprint}} (\bibinfo{year}{2025}).
\newblock
\showeprint[arxiv]{2510.19420}


\bibitem[Yehudai et~al\mbox{.}(2025)]%
        {yehudai2025surveyagents}
\bibfield{author}{\bibinfo{person}{Asaf Yehudai} {et~al\mbox{.}}}
  \bibinfo{year}{2025}\natexlab{}.
\newblock \showarticletitle{Survey on Evaluation of {LLM}-based Agents}. In
  \bibinfo{booktitle}{\emph{KDD '25 Tutorials}}.
\newblock
\newblock
\shownote{arXiv:2503.16416}.


\bibitem[Zhang et~al\mbox{.}(2025)]%
        {zhang2025agentracer}
\bibfield{author}{\bibinfo{person}{Guibin Zhang}, \bibinfo{person}{Junhao
  Wang}, \bibinfo{person}{Junjie Chen}, \bibinfo{person}{Wangchunshu Zhou},
  \bibinfo{person}{Kun Wang}, {and} \bibinfo{person}{Shuicheng Yan}.}
  \bibinfo{year}{2025}\natexlab{}.
\newblock \showarticletitle{AgenTracer: Who Is Inducing Failure in the LLM
  Agentic Systems?}
\newblock \bibinfo{journal}{\emph{arXiv preprint}} (\bibinfo{year}{2025}).
\newblock
\showeprint[arxiv]{2509.03312}


\end{thebibliography}


\appendix

\section{Versioning and Drift Detection}
\label{appx:drift}

Observability platforms watch agent \emph{behavior}; they do not
watch agent \emph{configuration}. A one-line edit to a system prompt
silently changes the agent's effective policy without leaving any
trace in observability tooling. KYA addresses this with
content-addressed versioning + drift detection over the agent
definition, plus a separate \emph{behavioral}-drift mechanism (the
mode-vs-config gap) capturing EU AI Act Art.~14 evidentiary
requirements for exercised oversight.

\paragraph{Canonical hash over explicit policy-bearing fields.}
The canonical hash $H_c(D_A)$ is the SHA-256 of a canonical-JSON
encoding (RFC~8785~\cite{RFC8785}) over an explicit, enumerated
allowlist of 18 policy-bearing fields: \btt{agent\_key},
\btt{name}, \btt{description}, \btt{system\_prompt},
\btt{model}, \btt{tools}, \btt{denied\_tools},
\btt{human\_loop}, \btt{access\_level}, \btt{can\_override},
\btt{can\_revert}, \btt{can\_delegate\_to},
\btt{required\_roles}, \btt{extends}, \btt{data\_classes},
\btt{security\_caps}, \btt{provenance}, \btt{model\_trust},
and \btt{compliance\_scope}. Operational fields (counters,
timestamps, deployment metadata) are deliberately excluded so the
hash is stable across runs of the same definition. Two
semantically-equal definitions hash identically; a single byte
change to any of the 18 fields yields a different hash. The
\emph{explicit allowlist} is a small governance contribution: naming
the policy-bearing fields lets the auditor see what is monitored.

\paragraph{Event-time + ingest-time versioned snapshots.}
Definitions snapshot into an append-only history table keyed by
(tenant\_id, agent\_key, version\_no) with two timestamp axes:
\emph{occurred\_at} (caller's clock, event time) and
\emph{created\_at} (storage clock, ingest time). The delta
$\Delta = $ created\_at $-$ occurred\_at exposes pipeline lag,
clock skew, or backdated tampering. Rollback is append-only:
\btt{rollback\_to(v\_old)} creates a new version that copies the
old definition with a "rolled back" note. History is never
mutated.

\paragraph{Drift detection algorithm.}
Given a previously-declared hash $h_d$ and current definition $D'$,
\cref{alg:detect-drift} returns either $\bot$ (no drift) or a
structural diff identifying changed fields.

\begin{algorithm}[h]
\caption{\texttt{detect\_drift}. Let $F$ denote the
\texttt{HASHED\_FIELDS} allowlist of 18 policy-bearing fields and
$D\vert_F$ the projection of definition $D$ onto $F$.}
\label{alg:detect-drift}
\begin{algorithmic}[1]
\Function{detect\_drift}{$h_d$, $D'$}
  \State $h' \gets \mathrm{SHA256}(\mathsf{canonical\_json}(D'|_F))$
  \If{$h' = h_d$}
    \State \Return $\bot$
  \EndIf
  \State $D_d \gets \Call{lookup\_by\_hash}{h_d}$
  \State \Return $\Call{structural\_diff}{D_d|_F,\; D'|_F}$
\EndFunction
\end{algorithmic}
\end{algorithm}

Diffs are structural (identifying \texttt{system\_prompt},
\texttt{tools[3]}, \texttt{model}, etc.) rather than byte-level
patches, recorded to the evidence chain
(\cref{sec:evidence}), and trigger operator review queues.

\paragraph{Lineage and risk inheritance.}
Agents derive from other agents (tenant clones a platform default,
edits, deploys). Each snapshot carries an optional
$\mathit{parent\_agent\_key}$ forming a DAG. When a parent's risk
score elevates (new vulnerability in tool set, compliance scope
change), KYA traverses children with per-generation decay
(\texttt{\_INHERITANCE\_PER\_GEN}: direct $+8$, grandchild $+4$,
great-grandchild $+2$, minimum $1$ for distant generations).

\paragraph{Behavioral drift: the mode-vs-config gap.}
The canonical-hash detector catches \emph{declared}-definition
drift. The \texttt{invocations.py} module records each invocation
with both the configured \texttt{human\_loop} mode and the
exercised mode at runtime. An agent configured
\texttt{human\_loop="hybrid"} that actually runs
\texttt{autonomous} 90\% of the time is a regulatory gap under EU
AI Act Art.~14, not a code bug. The \texttt{mode\_distribution()}
read aggregates exercised modes and flags configurations whose
declared mode does not match the empirical distribution.

\paragraph{Definition signature verification.}
\texttt{verify\_signature()} extends drift detection with
cryptographic signing: tenants sign the canonical hash of approved
definitions at registration; auditors verify against the operator's
public key. Signatures compose with the HMAC evidence chain
(\cref{sec:evidence}) as an entry of kind
\texttt{definition\_signed}; subsequent drift alerts chain forward
from the signature event.

\paragraph{What drift detection does not catch.}
The detector catches \emph{declared}-definition changes. It does
not catch model-provider behavioral drift (same prompt against a
new checkpoint), environment drift (a tool returning different
data), or upstream context drift (the retrieval corpus changed).
Those classes require the runtime signals of \cref{sec:rogue}, the
LLM-judge infrastructure, or external evaluator pulls from
Phoenix~\cite{phoenix2024}.

\begin{figure}[h]
\centering
\includegraphics[width=\columnwidth]{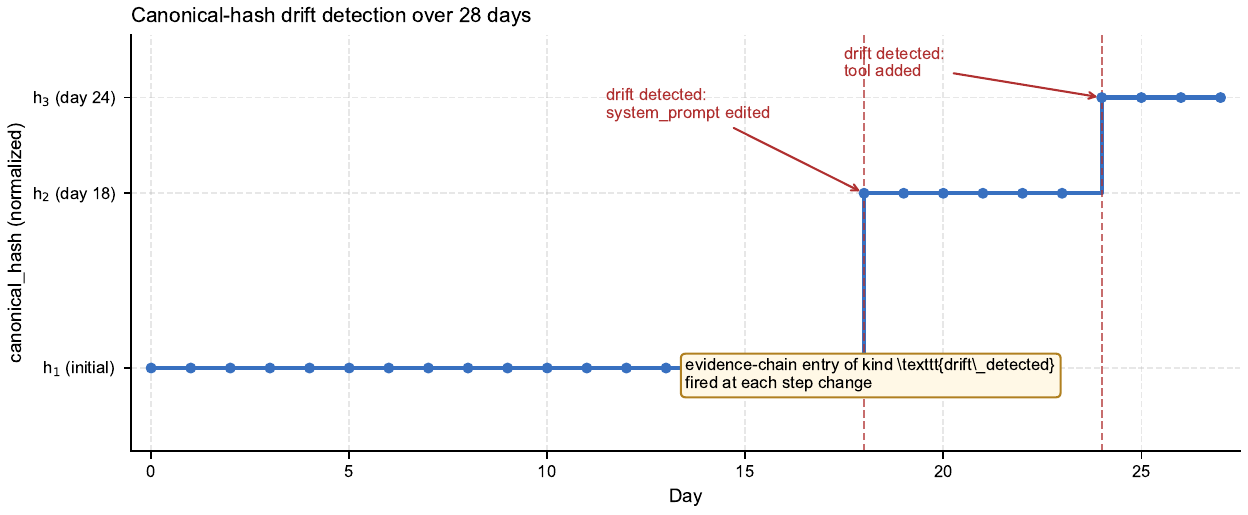}
\caption{Drift detection over a 28-day window. The hash is constant
for 17 days, jumps on day 18 when an unreviewed system-prompt edit
landed, and triggers an operator alert via the evidence-chain
event of kind \texttt{drift\_detected}.}
\label{fig:drift-timeline}
\end{figure}

\section{Evidence Chain --- Extended Details}
\label{appx:evidence-detail}

\paragraph{Type-marked canonicalization.}
\label{subsec:canon}
The hash input is canonical JSON (RFC~8785~\cite{RFC8785}) with a
specific extension: every non-JSON value wraps as
\verb|{"__t__": typename, "v": serialized_form}| before hashing.
This closes a collision-attack vector that the prior
\texttt{default=str} shortcut left open: a \texttt{datetime} and
its \texttt{isoformat()} string would hash identically without type
marking, letting an attacker swap a structured datetime for its
string representation without breaking the chain. The fix is small
(15 lines at \texttt{evidence.\_canonical\_default}) but
load-bearing.

\paragraph{Regime-aware retention floors.}
\label{subsec:retention}
The evidence module enforces per-regime retention floors at pruning
time: GDPR 6yr, HIPAA 6yr, NYDFS Part 500 5yr, SOX 7yr, PCI-DSS
1yr, EU AI Act 7yr, plus ISO 27001, SOC~2, FedRAMP, ITAR, NATO,
and EU classified-data regimes. Pruning treats these as
\emph{floors}, not ceilings: an operator may retain longer, never
shorter. Declarative, regime-aware, surfaces in the same
\texttt{prune\_expired\_evidence()} call.

\paragraph{Three-tier signing-key resolution.}
The HMAC signing key resolves through three tiers: (1)
\textbf{pluggable provider} ---
\btt{KYA\_EVIDENCE\_KEY\_PROVIDER} names a module:function that
returns the key, wired to KMS / Vault / HSM in production; (2)
\textbf{environment-mounted key} --- base64 in
\btt{KYA\_EVIDENCE\_SIGNING\_KEY}; (3) \textbf{development
fallback} --- process-local random key with a one-time WARN. The
dev fallback is intentional: a missing key cannot silently disable
evidence collection.

\paragraph{Verification.}
Single-event verification is $O(1)$: given $e_i$, $h_{i-1}$, $h_i$,
recompute and compare. Full-chain verification is $O(n)$. The SDK
exposes \btt{verify\_chain} and \btt{verify\_event}. The verifier
distinguishes three failure modes: \textbf{payload tamper}
(recomputed payload hash mismatches stored), \textbf{chain break
tamper} (\btt{prev\_hash} mismatches predecessor's signed hash),
and \textbf{clean cut} (expected predecessor pruned --- chain-head
mismatch at the first surviving entry of a pruned chain). The first
two are tamper findings; the third is a normal retention outcome.

\paragraph{Honest scope and upgrade path.}
KYA's evidence chain is verifiable by anyone holding the HMAC key.
For multi-party verification --- where no single party can forge the
chain --- threshold signatures over the chain root are future work
(\cref{sec:limitations}). Customer deployments requiring
multi-party audit currently publish the chain root to a separate
transparency log; a native primitive is planned for v2. External
notary anchoring (Sigstore Rekor~\cite{Newman2022Sigstore},
RFC~3161 timestamp authority) is similarly future work and
explicitly named in the module's v1-limitations docstring.

\section{Extended Related Work}
\label{appx:related}

\paragraph{Concurrent agent-governance systems (extended).}
\textbf{Aegis}~\cite{Aegis2026} proposes a runtime-governance
architecture with an Immutable Ethics Policy Layer and Immutable
Logging Kernel, claiming median proof-verification latency of
238\,ms. Submitted approximately six weeks before this work.
Aegis's stronger anchor (zero-knowledge verification of policy
enforcement) and our weaker anchor (HMAC chain with documented
Sigstore-anchoring upgrade path) represent a deliberate trade-off:
operator-key simplicity versus multi-party verifiability. Aegis
does not address tenant-override composition algebra, KYP, or
actor-agent dynamic attribution.

\textbf{AAGATE}~\cite{huang2025aagate} is a
Kubernetes-native NIST AI RMF control plane integrating MAESTRO
threat modeling, OWASP AIVSS, and SEI SSVC. AAGATE is the closest
academic prior art on overall scope. Our work differs by including
KYP, the formal only-tighten composition algebra, the four-gate
apply pipeline with operator-approval-as-default, the dynamic-trust
actor-agent attribution, and the interaction-multiplier
amplification with auditable codes. AAGATE is platform
infrastructure; KYA is an SDK that integrates into existing
platforms (including Kubernetes control planes such as AAGATE).

\textbf{SIGIL} (arXiv:2605.05274~\cite{shen2026sigil}) proposes a
tamper-evident on-chain registry for LLM skills with runtime
verification. The threat surface (compromised skills) and the trust
model (blockchain-anchored) are adjacent rather than directly
competitive with KYA. \textbf{PROV-AGENT}
(arXiv:2508.02866~\cite{souza2025provagent}) extends W3C PROV
provenance to agent workflows; KYA's evidence chain produces
compatible artifacts but uses a per-(tenant, invocation) HMAC chain
rather than a PROV graph.

\textbf{Microsoft Agent Governance
Toolkit}~\cite{microsoft_agt_2026} ships a multi-package suite
covering policy enforcement, DID-based agent identity, runtime
sandboxing, compliance mapping, and signed-plugin marketplace
verification. AGT is the closest commercial competitor on overall
feature breadth. Positioning differs on three axes: (a) KYA scores
users, agents, and service accounts in a single principal-trust
system, whereas AGT's Agent Mesh scores agents only (relying on
Entra ID upstream for users); (b) KYA's only-tighten composition
algebra is formalized with a soundness proof; (c) KYA's federated
weight-recommendation channel exists; AGT discusses cross-org
policy federation as future work
(issue~\#1386~\cite{microsoft_agt_issue1386}).

\textbf{TrustPact} (PyPI v0.1.0~\cite{trustpact2026}) implements a
behavioral trust scanner for MCP servers and AI agents using a
proprietary AEGIS five-dimensional scoring model. Closest direct
overlap on weighted-factor risk scoring. KYA's interaction-multiplier
amplification and per-(tenant, invocation) evidence chain have no
analog. \textbf{TrustEval-AI} offers benchmark prompts mapped to
HIPAA / GDPR / PCI-DSS / ABA Rules but is evaluation-only.
\texttt{ai-audit-sdk} provides asynchronous decision logging tied
to GDPR Article 22 without tamper-evidence, drift detection, or
federation.

\paragraph{Adjacent agent governance and admission control
(2025--2026).}
Gaurav et al.~\cite{gaurav2025gaas} propose Governance-as-a-Service
with runtime trust scoring but no signed-update distribution or
tenant override algebra. Kaptein et al.~\cite{kaptein2026runtime}
formalize proposed-action semantics for runtime governance but not
the apply-side gates of \cref{sec:fedrec}. Tan et
al.~\cite{tan2026attesting} (LLMSC 2026) target training-time
artifacts; KYA's four-gate apply targets runtime policy
distribution. Fernandez's Agent Control
Protocol~\cite{fernandez2026acp} uses Ed25519-signed execution
tokens for per-action admission control, which is action-scope
rather than policy-distribution-scope. Burke et al.'s
Rebound~\cite{burke2026rebound} (IEEE S\&P 2026) mediates
state-transition policy via a reference monitor for confidential
cloud applications --- adjacent authorization-mediation work but at
the cloud-state layer.

\paragraph{Multi-agent attack patterns and failure attribution.}
Liang et al.\ "Don't Trust Your
Upstream"~\cite{liang2025donttrust} demonstrates the exact attack
pattern that \cref{sec:static}'s delegation-trust attribution
defends against: topology-guided adversarial propagation from
exposed edge agents to high-privilege orchestrators (40--78\%
attack success across multiple frameworks). Wu et
al.~\cite{wu2025monitoring} propose node-evaluation monitoring on
signed agent-influence graphs --- the closest computational analog
to our delegation-trust mechanism, but using eigenvector-style
backpropagation rather than additive bucketed penalties with an
observation gate. AgenTracer~\cite{zhang2025agentracer} addresses
failure attribution via counterfactual replay (post-hoc), whereas
KYA's actor-agent attribution (\cref{subsec:actoragent}) debits the
orchestrator at runtime. Gabison and
Xian~\cite{gabison2025liability} (REALM at ACL 2025) provide the
legal-theoretic principal-agent framing.

\paragraph{Agentic risk-scoring frameworks beyond AIVSS.}
AURA, the Agent Autonomy Risk Assessment
framework~\cite{satta2025aura}, is the closest recent academic
competitor to KYA's static scoring layer. AURA uses
weighted-linear additive construction without per-interaction
multipliers or audit codes, reinforcing the position that KYA's
bounded asymmetric multiplicative amplification with per-interaction
codes is genuinely new.

\paragraph{LLM evaluation libraries.}
\texttt{lm-eval}~\cite{eval-harness},
DeepEval~\cite{deepeval2024}, RAGAS~\cite{es2024ragas},
TruLens~\cite{trulens2023}, HELM~\cite{liang2023helm}, OpenAI
Evals~\cite{openai_evals}, and Inspect AI~\cite{aisi2024inspect}
measure model behavior in isolation under fixed prompts; they do
not model the surrounding agent or produce the audit artifacts
regulators ask for. KYA composes with them: an evaluation verdict
from any can be recorded into the evidence chain as a
\texttt{quality} signal. Patronus AI Lynx~\cite{ravi2024lynx} is a
RAG-hallucination evaluator with per-criterion reasoning; KYA's
LLM-judge uses a similar architecture for fault attribution but
adds a verbatim-quote anti-hallucination guard that rejects judge
verdicts whose cited quotes fail to validate against the source.

\paragraph{Observability + eval platforms.}
By mid-2026 every major LLM observability platform markets audit
logs and EU AI Act readiness: LangSmith~\cite{langsmith2024}, Arize
Phoenix~\cite{phoenix2024} / Arize AX,
Langfuse~\cite{langfuse2023} Enterprise,
Braintrust~\cite{braintrust2024}, Galileo
Luna-2~\cite{galileo2024}. These produce \emph{platform-level}
audit logs (who-changed-what, who-ran-which-eval, RBAC actions).
None ships agent-identity-bound, cryptographically verifiable,
third-party-attestable governance artifacts. KYA composes:
\texttt{kya\_otlp\_bridge} ingests OpenTelemetry spans from any of
the above into evidence-chain entries.

\paragraph{Software provenance and signed-policy distribution.}
The HMAC-chain log construction is Bellare and
Yee~\cite{BellareYee1997} and Schneier and
Kelsey~\cite{schneierkelsey1998logs}; Merkle-tree alternative is
Crosby and Wallach~\cite{crosbywallach2009tamper}. Certificate
Transparency~\cite{RFC6962}, Trillian, and Sigstore
Rekor~\cite{Newman2022Sigstore} are industrial deployments;
in-toto and SLSA~\cite{TorresArias2019Intoto} address supply-chain
provenance. Signed policy distribution: KeyNote in the Distributed
Firewall~\cite{ioannidis2000distfw}, TUF and
Uptane~\cite{samuel2010tuf,kuppusamy2016uptane} for rollback
prevention, OPA Signed Bundles~\cite{OPA_SignedBundles} for
JWT-signed manifests, STIX/TAXII~\cite{STIX21_OASIS} for federated
threat-indicator sharing with analyst review. KYA's contribution is
not the signed-update transport but the four-gate apply pipeline
with operator-approval-as-default and the only-tighten composition
discipline.

\paragraph{Only-tighten composition --- Cedar lineage.}
The only-tighten algebra (\cref{sec:fedrec}) is structurally
adjacent to Cedar (Cutler et al.~\cite{cutler2024cedar}, OOPSLA
2024), which provides a Lean-mechanized formal proof that
\texttt{forbid} dominates \texttt{permit} under union composition
for a flat policy set. The same one-way property is operationally
deployed in Istio AuthorizationPolicy~\cite{istio_authz}
(CUSTOM~$\succ$~DENY~$\succ$~ALLOW, non-configurable), GCP IAM
Deny Policies~\cite{gcp_deny_policies} (descendant deny narrows
ancestor permits), and Windows App Control supplemental
policies~\cite{windows_appcontrol_signed} (supplementals may only
expand allow lists, never weaken the base policy's signer
requirements). We do not claim novelty in the one-way-tightening
property itself. KYA's contribution is the three-channel
hierarchical composition: platform default $\oplus$ tenant override
$\oplus$ signed external recommendation, with a soundness lemma
over the multi-authority hierarchy. The safety-lattice ancestor is
Denning~\cite{denning1976lattice} and Bell-LaPadula~\cite{bell1976unified};
the refinement-calculus framing is Back and von
Wright~\cite{back1998refinement}.

\paragraph{Adaptive governance and the never-auto-tune discipline.}
\Cref{subsec:closedloop}'s "never auto-tune" invariant intersects
AI corrigibility (Soares et
al.~\cite{soares2015corrigibility}; Nayebi's five-utility
construction~\cite{nayebi2025corrigible} is the closest formal
analog), the MAPE-K human-machine-teaming
lineage~\cite{kephart2003vision,clelandhuang2022mapekhmt,sanwouo2025aware}
(none formalize the never-auto-tune property on the adapted
parameters themselves), and active-learning noisy-oracle
gating~\cite{schubert2023deep}. The strongest regulatory analog is
the U.S.~FDA's PCCP~\cite{fda2024pccp}: changes outside a
pre-specified envelope require human re-submission. KYA's
closed-loop is the first SaaS-tenant analog of an FDA-PCCP-style
change-envelope for AI governance weights. The opposing position
--- Cloudflare WAF ML~\cite{cloudflare_waf_ml_2022} and
IBM/Acceldata enterprise-governance practice --- advocates automatic
operational-parameter adjustment without customer approval,
sharpening the contrast.

\paragraph{Federated learning, federated analytics, federated policy.}
KYA's federated recommendation channel borrows the topology of
federated learning but not the mechanism: no gradient training, no
model checkpoint --- the distributed artifact is a signed weight
recommendation. This is closer to federated policy distribution
than to federated model
training~\cite{FedAvg,kairouz2021advances}. Federated
analytics~\cite{ramage2020federated} is structurally closer for
the upward telemetry; STIX/TAXII and cloud-EDR indicator
distribution (CrowdStrike, MS~Defender) are the closest
operational analogues for the downward signed-recommendation path.

\paragraph{AI governance frameworks (citation context).}
KYA's compliance scope maps to NIST AI RMF~\cite{NISTAIRMF},
Generative AI Profile~\cite{nist_ai_600_1},
adversarial-ML taxonomy~\cite{nist_ai_100_2_e2025}, EU AI
Act~\cite{EUAIAct2024}, ISO/IEC 42001~\cite{ISO42001},
ISO/IEC 42005~\cite{iso_42005_2025}, SR-11/7~\cite{SR117}, NYDFS
Part 500, DORA, GDPR, HIPAA. KYA does not propose changes to these
frameworks; it produces artifacts populating their documentation
requirements.

\section{Cross-Cutting Design Disciplines --- Source Citations}
\label{appx:disciplines}

\paragraph{Closed-set whitelists.}
Signal kinds (\btt{ALLOWED\_SIGNAL\_KINDS},
\btt{kya/realtime.py:92}), principal kinds
(\btt{PRINCIPAL\_KINDS}, \btt{kya/principals.py:81}), evidence
kinds (\btt{VALID\_EVIDENCE\_KINDS}, \btt{kya/evidence.py:128}),
inbound scopes (\btt{KNOWN\_SCOPES}, \btt{kya/inbound.py:137}),
dual-write tables (\btt{ALLOWED\_TABLES},
\btt{kya/dualwrite.py:122}), data classes
(\btt{kya/data\_classes.py:67--92}), compliance regimes
(\btt{kya/compliance.py:95--130}).

\paragraph{Bounded composition.}
Interaction-multiplier cap (\btt{MAX\_MULTIPLIER=2.0} at
\btt{kya/interactions.py:36}), delegation-trust cap
(\btt{DELEGATION\_TRUST\_CAP=25} at
\btt{kya/delegation\_trust.py:37}), security-capabilities cap
(60), data sensitivity cap (60), supply-chain cap (35),
blast-radius cap (30), input-source cap (25), final score clamp
to 100.

\paragraph{Asymmetric composition.}
Data sensitivity \textbf{MAX} (\btt{kya/data\_classes.py:204});
security capabilities \textbf{SUM}
(\btt{kya/security\_caps.py:128--138}); input sources
\textbf{base + breadth premium}
(\btt{kya/input\_sources.py:86--93}); interaction multipliers
$\geq 1.0$ enforced at
\btt{kya/interactions.py:register\_interaction} line~268.

\paragraph{"Never X" invariants.}
Never auto-tune (\btt{kya/feedback.py:20--22}); never
apply-anyway on signature failure (\btt{kya/inbound.py:31--32});
never bypass human gate on critical \btt{flag\_for\_review}
verdicts; never break scoring on hot-path exceptions
(\texttt{kya/risk.py:486--490}); never leak across tenants via
schema translation.

\paragraph{Fail-soft observability.}
Valkey unreachable, Phoenix unreachable, collector unreachable,
\texttt{prometheus\_client} unavailable --- core scoring and
persistence continue unaffected.

\section{Factor Decomposition --- Full Weight Schedules}
\label{appx:factor-detail}

\paragraph{Write-tool count.}
Rule-based classifier with a 20-entry write-prefix list
(\texttt{create\_}, \texttt{delete\_}, \texttt{update\_},
\texttt{override\_}, \texttt{revert\_}, \texttt{ingest\_},
\texttt{execute\_}, \texttt{publish\_}, \texttt{ack\_},
\texttt{suspend\_}, \texttt{remove\_}, etc.); catalog-based
classifier promotes any tool with role requirements to write tool.

\paragraph{Provenance schedule.}
\texttt{builtin} (0) $\to$ \texttt{custom} (5) $\to$
\texttt{imported} (10) $\to$ \texttt{marketplace} (15) $\to$
\texttt{third\_party} (20).

\paragraph{Model trust schedule.}
\texttt{enterprise} (0) $\to$ \texttt{frontier} (3) $\to$
\texttt{open} (8) $\to$ \texttt{self\_hosted} (10).

\paragraph{Data sensitivity schedule.}
Closed civilian + defense + NATO + EU taxonomy:
\texttt{public} (0) $\to$ \texttt{pii} (15) $\to$
\texttt{phi} (30) $\to$ \texttt{cui} (35) $\to$
\texttt{itar} (50) $\to$ \texttt{us\_secret} (55) $\to$
\texttt{us\_top\_secret} (60).

\paragraph{Security capability schedule.}
\texttt{fs\_read} (5), \texttt{network\_egress} (10),
\texttt{code\_execution} (20), \texttt{shell\_access} (25),
\texttt{container\_exec} (30); SUM-aggregated, cap 60.

\paragraph{Deployment environment schedule.}
\texttt{dev} (0) $\to$ \texttt{staging} (5) $\to$ \texttt{prod}
(15) $\to$ \texttt{enclave} (25). Enclave weighted highest because
misbehavior in a classified enclave is a national-security
incident, not a Slack message.

\paragraph{Delegation depth.}
Maximum chain length reachable via \btt{can\_delegate\_to}, with
hard caps against pathological graphs (5000 nodes, 50 hops, 50
paths returned).

\paragraph{Supply chain.}
\btt{first\_party} (0) $\to$ \btt{marketplace} (5) $\to$
\btt{self\_hosted\_ext} (10) with breadth premium for $>5$
dependencies.

\paragraph{Input sources.}
\btt{web\_fetch} and \btt{user\_upload} tied at 15 as the
highest-risk vectors. \btt{unknown} weighted (10) higher than
declared \btt{external\_api} (8).

\paragraph{Lifecycle.}
Ownership signed; approval status (pending 10, rejected 30, expired
20, unknown 15); time-in-production (brand-new $<7$d $+8$, churning
$\geq 20$ versions in 30d $+10$). Approval auto-degrades to
\texttt{expired} when \texttt{review\_expires\_at} is past.

\paragraph{Compliance scope.}
Each regulatory regime (GDPR, HIPAA, EU AI Act, NYDFS Part 500,
DORA, SR-11/7, plus 26 additional regimes including ITAR, EAR,
CMMC, FedRAMP, NATO/EU classified) elevates baseline severity
multiplicatively rather than additively.

\paragraph{Trust signals.}
Red-team, fairness, citation evidence contribute negative deltas;
stale audits decay to half weight after 180 days; missing audit
evidence adds positive risk (untested = dangerous). Cited
explicitly to EU AI Act Article~13.

\paragraph{Cost burn.}
Hourly burst $>5\times$ monthly average, budget exhaustion, anomaly
factor. Cost burn is framed as a security signal first, financial
second: prompt-injection attacks often trigger expensive
long-context bursts.

\section{Dynamic Rogue Signals --- Source-Code Instrumentation}
\label{appx:rogue-detail}

\paragraph{Out-of-scope tool attempts (OOS).}
Instrumented in the \btt{agents.api} and
\btt{agents.streaming} agent loops at the point of tool
dispatch. Counter:
\btt{veldt\_agent\_oos\_tool\_attempts\_total}.

\paragraph{Cross-tenant attempts.}
Instrumented at the tool-call binding point, comparing the
declared \texttt{tenant\_id} on the call to the agent's bound
tenant. Counter: \btt{veldt\_agent\_cross\_tenant\_attempts\_total}.

\paragraph{Data leakage.}
Detected in the action gate's class-tagged output classifier; a
class mismatch records a leak signal. Counter:
\btt{veldt\_agent\_data\_leak\_total}.

\paragraph{Policy violations.}
Block verdicts from content-safety / jailbreak / refusal-fail
policies. Counter:
\btt{veldt\_agent\_policy\_violations\_total}.

\paragraph{Signal flow.}
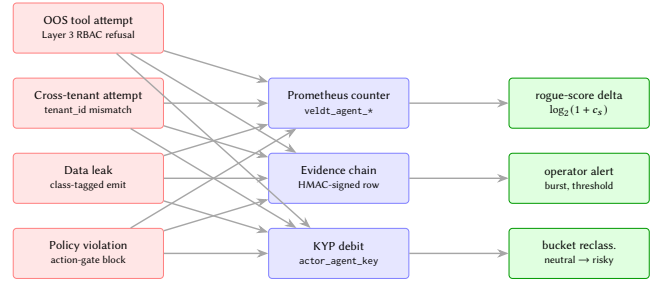
\begin{figure}[h]
\centering
\resizebox{\columnwidth}{!}{

\begin{tikzpicture}[
  font=\sffamily\footnotesize,
  src/.style={draw, rounded corners=2pt, minimum width=30mm, minimum height=10mm,
              align=center, inner sep=3pt, fill=red!10, draw=red!50},
  proc/.style={draw, rounded corners=2pt, minimum width=28mm, minimum height=10mm,
               align=center, inner sep=3pt, fill=blue!10, draw=blue!50},
  sink/.style={draw, rounded corners=2pt, minimum width=28mm, minimum height=10mm,
               align=center, inner sep=3pt, fill=green!12, draw=green!60!black},
  arrow/.style={-{Stealth}, thick, gray!70},
]

\node[src] (oos) at (0,3) {OOS tool attempt\\\scriptsize Layer 3 RBAC refusal};
\node[src] (ct)  at (0,1.5) {Cross-tenant attempt\\\scriptsize tenant\_id mismatch};
\node[src] (dl)  at (0,0) {Data leak\\\scriptsize class-tagged emit};
\node[src] (pv)  at (0,-1.5) {Policy violation\\\scriptsize action-gate block};

\node[proc] (cnt) at (5,1.5) {Prometheus counter\\\scriptsize \texttt{veldt\_agent\_*}};
\node[proc] (chn) at (5,0) {Evidence chain\\\scriptsize HMAC-signed row};
\node[proc] (kyp) at (5,-1.5) {KYP debit\\\scriptsize \texttt{actor\_agent\_key}};

\node[sink] (delta) at (9.8,1.5) {rogue-score delta\\\scriptsize $\log_2(1+c_s)$};
\node[sink] (alert) at (9.8,0) {operator alert\\\scriptsize burst, threshold};
\node[sink] (bucket) at (9.8,-1.5) {bucket reclass.\\\scriptsize neutral $\to$ risky};

\foreach \s in {oos,ct,dl,pv} {
  \draw[arrow] (\s) -- (cnt);
  \draw[arrow] (\s) -- (chn);
  \draw[arrow] (\s) -- (kyp);
}

\draw[arrow] (cnt) -- (delta);
\draw[arrow] (chn) -- (alert);
\draw[arrow] (kyp) -- (bucket);

\end{tikzpicture}}
\caption{Runtime signal flow. Each observed event is counted,
recorded to the evidence chain, and surfaced as a score delta.}
\label{fig:signal-flow}
\end{figure}

\paragraph{Rogue heatmap.}
\begin{figure}[h]
\centering
\includegraphics[width=\columnwidth]{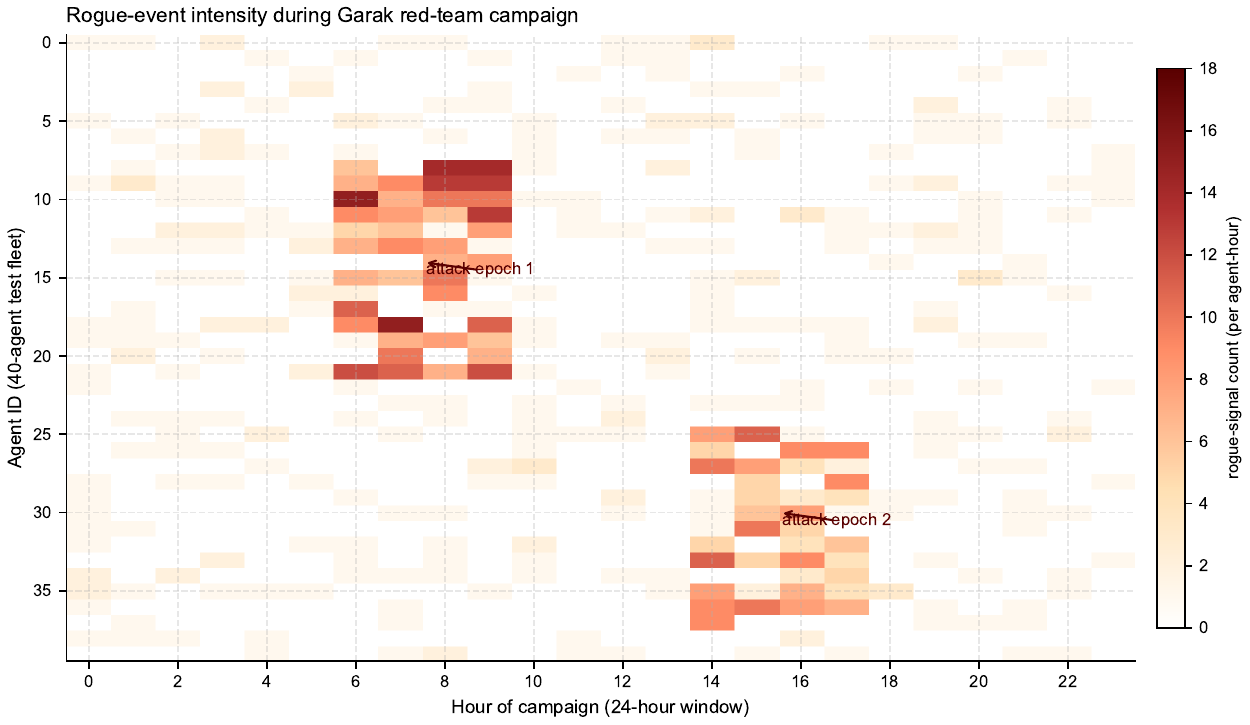}
\caption{Rogue-event heatmap across a 40-agent test fleet during a
24-hour Garak red-team campaign. Clusters of high intensity
co-occur with attack epochs (annotated arrows).}
\label{fig:rogue-heatmap}
\end{figure}

\section{Deployment Topology Figure}
\label{appx:topology-figure}

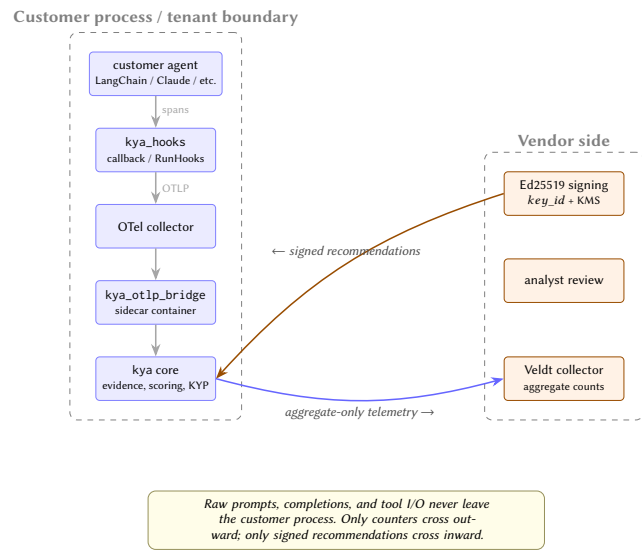
\begin{figure}[h]
\centering
\resizebox{\columnwidth}{!}{

\begin{tikzpicture}[
  font=\sffamily\footnotesize,
  swim/.style={draw, rounded corners=2pt, minimum height=8mm, minimum width=22mm,
               align=center, inner sep=3pt},
  cust/.style={swim, fill=blue!8, draw=blue!50},
  vendor/.style={swim, fill=orange!10, draw=orange!60!black},
  arrow/.style={-{Stealth}, thick, gray!70},
  agg/.style={font=\sffamily\footnotesize\itshape, gray!60!black},
]

\node[cust] (agent) at (0,3.2) {customer agent\\\scriptsize LangChain / Claude / etc.};
\node[cust] (hooks) at (0,1.8) {\texttt{kya\_hooks}\\\scriptsize callback / RunHooks};
\node[cust] (otel) at (0,0.4) {OTel collector};
\node[cust] (bridge) at (0,-1.0) {\texttt{kya\_otlp\_bridge}\\\scriptsize sidecar container};
\node[cust] (core) at (0,-2.4) {\texttt{kya} core\\\scriptsize evidence, scoring, KYP};

\draw[arrow] (agent) -- node[right, font=\sffamily\scriptsize] {spans} (hooks);
\draw[arrow] (hooks) -- node[right, font=\sffamily\scriptsize] {OTLP} (otel);
\draw[arrow] (otel) -- (bridge);
\draw[arrow] (bridge) -- (core);

\node[draw, dashed, gray, fit={(agent)(core)}, inner sep=10pt, rounded corners,
      label={[font=\sffamily\bfseries,gray]above:Customer process / tenant boundary}] {};

\node[vendor] (col) at (7.5,-2.4) {Veldt collector\\\scriptsize aggregate counts};
\node[vendor] (review) at (7.5,-0.6) {analyst review};
\node[vendor] (sign) at (7.5,1.0) {Ed25519 signing\\\scriptsize $key\_id$ + KMS};

\node[draw, dashed, gray, fit={(col)(sign)}, inner sep=10pt, rounded corners,
      label={[font=\sffamily\bfseries,gray]above:Vendor side}] {};

\draw[arrow, blue!60] (core.east) to[bend right=15] node[below, agg]
  {aggregate-only telemetry $\rightarrow$} (col.west);
\draw[arrow, orange!60!black] (sign.west) to[bend right=15] node[above, agg]
  {$\leftarrow$ signed recommendations} (core.east);

\node[font=\sffamily\footnotesize\itshape, text width=70mm, align=center,
      fill=yellow!10, draw=yellow!50!black, rounded corners, inner sep=4pt]
  at (3.5,-5) {Raw prompts, completions, and tool I/O never leave
  the customer process. Only counters cross outward; only signed
  recommendations cross inward.};

\end{tikzpicture}}
\caption{Deployment topology. Raw prompt/completion payloads remain
in the customer process; only aggregated telemetry crosses the
organizational boundary, and only signed policy recommendations
come back.}
\label{fig:topology}
\end{figure}

\section{Clinical Triage Scoring --- Full Breakdown}
\label{appx:clinical-scoring}

Scoring the clinical-triage fleet of \cref{subsec:second-domain}
on the identical factor pipeline used for the loan-decisioning
fleet: the \textbf{Diagnosis Suggestion Sub-Agent} scores
\textbf{84 (critical)} --- no write authority, but \texttt{phi}
data class (+30), \texttt{phi\_genetic} adds further via MAX
aggregation, \texttt{frontier} model (+3), and the
\texttt{phi\_AND\_LLM\_uncertainty} interaction multiplier fires
at $\times 1.20$. The \textbf{HL7 Lookup Sub-Agent} scores
\textbf{56 (medium)} --- read-only, but \texttt{phi} sensitivity
floor + EHR provenance push it just under the high threshold. The
\textbf{Triage Orchestrator} scores \textbf{70 (high)} via the
delegation-trust premium for delegating to the critical-bucket
Diagnosis sub-agent --- the same mechanism that lifts the bank's
Loan Triage Agent into the high bucket via its Risk Review
delegate.

\section{Full Compliance Regime Coverage}
\label{appx:compliance-regimes}

KYA 0.1.0 maps to 32 regulatory regimes; the 27 not in
\cref{tab:compliance}:

\begin{table}[h]
\centering
\scriptsize
\setlength{\tabcolsep}{4pt}
\begin{tabular}{lrrl}
\toprule
\textbf{Regime} & \textbf{Retention} & \textbf{Breach} & \textbf{Tier model} \\
\midrule
SOX             & 7 yr    & --- & --- \\
PCI-DSS         & 1 yr    & --- & --- \\
SR 11-7         & 7 yr    & --- & Model tiers \\
ISO 42001       & 3 yr    & --- & --- \\
ITAR/EAR        & 5 yr    & --- & Export-control \\
NIST 800-171/53 & 6 yr    & --- & --- \\
IL5 / IL6       & 6/25 yr & --- & DoD IL \\
CCPA            & 2 yr    & --- & --- \\
GLBA            & 5 yr    & --- & --- \\
FERPA           & 5 yr    & --- & Student-record \\
FedRAMP         & 3 yr    & --- & Low/Mod/High \\
CMMC            & 6 yr    & --- & Level 1--3 \\
IRAP (AU)       & 7 yr    & --- & Prot./Sec./TS \\
CCCS (CA)       & 6 yr    & --- & --- \\
C5 (DE)         & 3 yr    & --- & --- \\
ENS (ES)        & 3 yr    & --- & Bajo/Med/Alto \\
NATO            & 25 yr   & --- & Restr.--Cosmic \\
EU classified   & 25 yr   & --- & R/C/S/TS \\
ISO 27001       & 3 yr    & --- & --- \\
ISO 42005       & 3 yr    & --- & Impact-assess. \\
SOC 2           & 3 yr    & --- & --- \\
SR 21-8         & 7 yr    & --- & --- \\
NIST AI RMF     & 3 yr    & --- & GMMM functions \\
NIST AI 600-1   & 3 yr    & --- & GenAI Profile \\
NIST AI 100-2   & 3 yr    & --- & Adversarial-ML \\
EU AI Act DA    & 10 yr   & --- & Delegated acts \\
HIPAA Title III & 6 yr    & 60 d & Min.-necessary \\
\bottomrule
\end{tabular}
\end{table}

The mapping is advisory: KYA cannot itself certify a deployment as
compliant, but it produces the artifacts (canonical hashes,
evidence chain, scored decisions, drift records, regulatory
breach-notification fan-out per the
\texttt{compliance\_shim.py} idempotent multi-regime emitter) that
compliance teams need to assemble certification packages. The
multi-regime fan-out is itself notable: a single incident with PII +
NYDFS scope produces two distinct notifications with different
formats and SLAs (72\,h EDPB + 72\,h NYDFS), with
\texttt{UNIQUE(incident\_id, regime)} preventing double-fire.

\section{Extended Limitations}
\label{appx:limitations}

\paragraph{Reference-implementation collector.}
The federated channel (\cref{sec:fedrec}) is fully specified on the
SDK side. The eval-side collector is a reference implementation.
The 0.1.0 release ships an empty \texttt{DEFAULT\_PINNED\_KEYS}
map; operators configure trust anchors via
\texttt{KYA\_INBOUND\_PUBLIC\_KEY} for self-operated deployments.
Hardening the reference collector into a multi-tenant production
service --- key custody, analyst review workflows, cross-organization
telemetry ingestion, and the operational discipline that supports
them --- is out of scope.

\paragraph{Conceptual proximity to provenance and delegated authorization.}
The actor-agent attribution mechanism (\cref{subsec:actoragent}) is
conceptually adjacent to OTel baggage
propagation~\cite{w3c_baggage}, the RFC~8693 \texttt{act} chain for
delegated authorization, and correspondent-banking attribution
rules in AML. KYA's specific composition --- runtime
trust-counter debit at signal emission, with a default-key
convention across three SDK hook surfaces --- is distinct in
mechanism and direction. We do not claim the conceptual move
("blame the orchestrator at runtime, not the puppet") is novel in
principle; only that the runtime counter-debit primitive composed
with the three-hook default-key convention has no equivalent in
the surveyed agent-governance systems.

\paragraph{Closed-set whitelists require SDK release for new kinds.}
The closed-set discipline (\cref{sec:disciplines}) is a strength
(auditable, caller-controlled-keyspace prevention) and a
limitation: a customer with a novel signal kind, principal kind,
or evidence kind cannot add it via runtime configuration; they
must wait for an SDK release. For most deployments the trade-off is
acceptable, but high-velocity customers may experience the
closed-set discipline as friction.

\paragraph{Behavioral drift is partially out of scope.}
The drift detector (\cref{appx:drift}) catches changes to the
agent's declared definition. It does not catch model-provider
behavioral drift, tool-environment drift, or retrieval-corpus
drift. The mode-vs-config gap catches one axis of behavioral
drift but the broader category requires the runtime signals of
\cref{sec:rogue}, an LLM-judge regression harness, or external
evaluator pulls.

\paragraph{Adapter coverage is broad but not exhaustive.}
KYA ships 15+ native adapters; several emerging frameworks
(LangGraph, LangFlow, DSPy, proprietary enterprise frameworks) are
covered via the generic fallback with reduced fidelity
(framework-specific metadata is not extracted). First-party
support is on the roadmap.

\paragraph{Privacy of aggregate telemetry.}
The federated channel sends counts only, never payloads. It does
not currently apply differential privacy~\cite{dwork2014algorithmic}
or secure aggregation~\cite{bonawitz2017practical} to the counts.
At small fleet sizes the counts may leak coarse information about a
customer's deployment. Adding a differential-privacy budget per
tenant per epoch is future work and a prerequisite for
cross-industry deployments where counts could be sensitive.

\paragraph{Red-team campaign distribution is curated.}
The 1,200 adversarial probes in \cref{subsec:attack-eval} were
sampled from PyRIT and Garak default test sets plus the Liang
topology attack~\cite{liang2025donttrust}. These distributions
oversample published attack templates and undersample novel
zero-day attacks --- the 11\% miss rate on multi-turn
prompt-injection sequences likely understates real-world risk. A
production deployment should pair KYA with an ongoing red-team
program rather than relying solely on the published harness.

\paragraph{Insider-threat collusion.}
\Cref{sec:fedrec} bounds the blast radius of compromised collector
keys via the only-tighten algebra, but KYA does not defend against
a tenant operator who knowingly applies harmful recommendations
against their own organization. Insider threats remain in scope for
the customer's authentication, RBAC, and audit infrastructure, not
KYA's federation channel. Similarly, KYA does not defend against a
platform admin who lowers the platform default itself.

\paragraph{Local vs concurrent agent-governance work.}
Several concurrent 2025--2026 systems
(AAGATE~\cite{huang2025aagate}, Aegis~\cite{Aegis2026},
SIGIL~\cite{shen2026sigil}, Microsoft Agent Governance
Toolkit~\cite{microsoft_agt_2026}) target adjacent or overlapping
problems. We do not claim categorical priority over these systems;
\cref{appx:related} identifies the specific delta KYA provides
against each. Coordination on shared schema conventions for
cross-system audit interoperability is desirable future work.

\section{Only-Tighten Algebra --- Definition, Lemma, Proof}
\label{appx:tighten}

We restate the formal block from \cref{sec:fedrec} here for
reference. The body presents the result in prose; this appendix
gives the definition, lemma, and proof sketch.

Let $W_0: (\mathit{scope}, \mathit{key}) \to \mathbb{N}$ be the
platform-default weight function. Define the partial order $\succeq$
on weights by $w' \succeq w \iff w' \ge w$ (higher = tighter).

\begin{definition}[Only-tighten override]
\label{def:tighten}
A tenant override $u: (\mathit{scope}, \mathit{key}) \to \mathbb{N}$
is \emph{only-tighten} if for every $(s, k)$,
$u(s, k) \succeq W_0(s, k)$.
\end{definition}

\begin{lemma}[Only-tighten soundness]
\label{lem:tighten}
For any sequence of \texttt{set\_override} calls (tenant-scoped) and
signed-recommendation applies starting from $W_0$, the effective
tenant weight function $W_t$ satisfies $W_t \succeq W_0$ at all
times $t$, and the composition is monotone non-decreasing in
$\succeq$ within the tenant scope.
\end{lemma}

\begin{proof}[Proof sketch]
By induction on the apply sequence. Base case: $W_0 \succeq W_0$.
Inductive step: every \texttt{set\_override} call satisfies the
only-tighten predicate by construction (the function raises on
violation), so any applied override $u$ satisfies $u \succeq W_0$.
Therefore $W_{t+1}(s, k) = \max(W_t(s, k), u(s, k)) \succeq W_t(s, k)
\succeq W_0(s, k)$. Every applied recommendation $r$ passes the same
predicate before persistence (\cref{subsec:fedinbound}); monotonicity
within the tenant scope follows directly.
\end{proof}

\paragraph{Threat-model precision.} The invariant binds
\textbf{tenants}, not platform admins. The
\texttt{\_check\_only\_tighten} function returns early when
\texttt{tenant\_id=None} (\texttt{tenant\_weights.py:195-197}). The
intent is asymmetric: platform admins remain trusted to lower the
platform default itself (a deliberate management decision); tenants
and inbound recommendations cannot lower below it. The value to a
tenant is immutability against vendor recommendations and against
compromised platform admin --- not absolute immutability.

\section{Discussion}
\label{appx:discussion}

We step back to discuss what KYA implies for AI governance practice,
where its limits sit relative to broader accountability questions,
and which design moves we believe will or will not generalize.

\paragraph{Governance artifacts are not governance.}
KYA produces artifacts --- canonical hashes, HMAC-chained evidence,
signed recommendations, per-(tenant, principal) trust counters,
multi-regime breach-notification fan-outs --- that regulator
documentation requirements consume. Producing these artifacts does
not make a deployment safe; it makes a deployment \emph{accountable}.
An agent that performs catastrophically can have an immaculate
evidence chain documenting exactly how. KYA is necessary
infrastructure for the audit, attestation, and incident-reporting
demands of modern AI governance frameworks; it is not sufficient for
trustworthy autonomy.

\paragraph{Agent-level scoring vs.\ fleet-level fairness.}
\label{appx:fairness}
KYA's static and dynamic scoring operates at the level of an
individual agent or principal. Much of the regulatory and ethics
literature operates at the level of a \emph{deployment} or a
\emph{population}. A loan-decisioning fleet
(\cref{subsec:running-example}) can have every individual agent
scored \emph{high} or below without the fleet being collectively
compliant with ECOA's disparate-impact thresholds. KYA does not
aggregate to that level. Concretely, an ECOA disparate-impact audit
would consume: (i) per-invocation evidence-chain entries keyed by
applicant identifier and decision outcome, (ii) the action-gate
verdict at $e_9$ in each chain, (iii) the static risk score of the
agent that produced each decision for control-stratification, and
(iv) the per-principal trust counters to identify operators whose use
correlates with outcome disparity. An external auditor pulls these
via the read-only SDK (\texttt{kya.list\_evidence},
\texttt{kya.list\_invocations}), computes the four-fifths-rule
selection ratios stratified by protected class, and writes the
finding back as a \texttt{quality} signal in the same evidence
chain --- closing the loop without KYA's core ever modeling
"protected class" as a first-class concept. The same pattern
generalizes to HIPAA Title~III minimum-necessary audits on the
clinical fleet of \cref{subsec:second-domain}, ITAR-flow audits on
defense deployments, and EU AI Act Article~10 data-governance
audits. We treat building this audit layer as out of scope but
explicitly load-bearing for the claim that KYA is "the substrate
fairness audits run on."

\paragraph{The operator-in-the-loop invariant under pressure.}
The "never auto-tune" invariant (\cref{subsec:closedloop}) is the
most operationally consequential design choice. Operator review is
the bottleneck at scale: a tenant managing 200 agents with 5--10
weekly weight-adjustment suggestions per agent quickly exceeds
human attention. The pressure to relax --- to auto-approve "small"
in-envelope changes --- will be significant. The right resolution
is not to relax the invariant but to make operator review better
(batched approval UIs, approve-with-comment for pre-registered
envelopes, score-shifting-suggestion surfacing). The FDA
PCCP~\cite{fda2024pccp} precedent supports this: an
operator-defined change envelope is the unit of authority. The cost
of getting this wrong is asymmetric --- a feedback loop where one
false-positive incident silently weakens the model is the single
most plausible failure mode of automated governance adaptation.

\paragraph{Cross-vendor trust portability.}
A regulator should be able to take an agent scored by vendor~$X$
and verify its score using vendor~$Y$. KYA's canonical-form
scoring (\cref{sec:static}) is a step toward this. Trust
portability also requires runtime signals (rogue events, evidence
chain) verifiable across vendor boundaries. The fairness community
benefits more from interoperability profiles mapping among
candidates (W3C PROV, OTel semantic conventions, in-toto
attestation, AAGATE and Microsoft AGT schemas) than from picking a
winner. \texttt{kya\_otlp\_bridge} already operates as such a shim
for OpenInference- and OTel-instrumented frameworks.

\paragraph{What we do not believe will generalize.}
We are skeptical of two design moves that recur in adjacent
proposals. \textbf{Vendor-trained model updates without
customer-side gate} (Cloudflare WAF
ML~\cite{cloudflare_waf_ml_2022}, several
SOAR~\cite{ibm_soar_2024}) is a viable shape for \emph{detection
rules} but not for \emph{governance policy weights}. The blast
radius of a bad detection rule is bounded by false positives in
alert queues; the blast radius of a bad weight change is silent
reclassification of agents into looser buckets across an entire
customer base. \textbf{Blockchain-anchored audit logs}
(SIGIL~\cite{shen2026sigil}) solve a problem most regulated
customers do not have. Regulators want a verifiable log they can
read; they do not want to operate a node on a public ledger.
Sigstore Rekor~\cite{Newman2022Sigstore} and Certificate
Transparency-style public logs strike a better balance:
transparency without participation requirement, which we see as the
credible upgrade path for KYA's evidence chain.

\end{document}